\documentclass[10pt]{iopart}

\pdfoutput=1

\usepackage{iopams}
\usepackage{ifpdf}

\ifpdf
   \usepackage{thumbpdf}
   \usepackage[pdftex]{graphicx}
   \usepackage[pdftex]{color} 
   \usepackage[pdftex,colorlinks]{hyperref}
   \def\myfigpng#1#2{\pdfimage height #2 {#1.png}}
   
   \def\myfig#1#2{\pdfimage height #2 {#1.pdf}}
\else
   \usepackage{graphicx}
   \usepackage{color}
   \def\href#1#2{#2}
   \def\myfigpng#1#2{\includegraphics[height=#2]{#1.eps}}
   
   \def\myfig#1#2{\includegraphics[height=#2]{#1.eps}}
\fi

\def\tbar{|\hspace*{-0.15em}|\hspace*{-0.15em}|}
\def\Ahatstar{{}^{*}\!\!\hat{A}}
\def\text#1{{\rm #1}}

\def\calg#1{{\mathcal #1}}
\def\bbbb#1{{\mathbb #1}}
\def\PLTMG{{\sc PLTMG}}
\def\MC{{\sc MC}}

\def\MCLite{{\sc MCLite}}

\def\FETK{{\sc FETK}}
\def\Re{{\mathbb R}}
\newtheorem{theorem}{Theorem}[section]
\newtheorem{lemma}{Lemma}[section]
\newtheorem{definition}{Definition}[section]
\newtheorem{remark}{Remark}[section]
\newtheorem{algorithm}{Algorithm}
\newenvironment{proof}
{\begin{trivlist} \item[\hskip\labelsep{\it Proof.}]}
{\vbox{\hrule\hbox{\vrule height1.4ex\hskip1.4ex\vrule}\hrule}\end{trivlist}}

\begin{document}

\title[Generating Initial Data in GR using Adaptive FEM]
      {Generating Initial Data in General Relativity
       using Adaptive Finite Element Methods}
\author{B.\ Aksoylu$^1$, D.\ Bernstein$^2$, S.D.\ Bond$^3$, M.\ Holst$^4$}

\address{$^1$ Department of Mathematics and Center for Computation and 
  Technology, Louisiana State University, Baton Rouge, LA 70803, USA\\
  $^2$ Silicon Clocks, 39141 Civic Center Dr., Suite 450 Fremont, CA 94538, USA\\
  $^3$ Department of Computer Science,
  University of Illinois at Urbana-Champaign, Urbana, IL, 61801, USA\\
  $^4$ Department of Mathematics,
  University of California at San Diego, La Jolla, CA 92093, USA}

\ead{$^1$ burak@cct.lsu.edu, $^2$ david.h.bernstein@gmail.com, $^3$ sdbond@uiuc.edu,
  $^4$ mholst@sobolev.ucsd.edu}

\begin{abstract}
The conformal formulation of the Einstein constraint equations is
first reviewed, and we then consider the design, analysis, and implementation 
of adaptive multilevel finite element-type numerical methods for the 
resulting coupled nonlinear elliptic system.
We derive weak formulations of the coupled constraints, and review some
new developments in the solution theory for the constraints
in the cases of constant mean extrinsic curvature (CMC) data,
near-CMC data, and arbitrarily prescribed mean extrinsic curvature data.
We then outline some recent results on {\em a priori} and {\em a posteriori} 
error estimates for a broad class of Galerkin-type approximation methods 
for this system which includes techniques such as finite element, wavelet,
and spectral methods.
We then use these estimates to construct an adaptive finite 
element method (AFEM) for solving this system numerically, and
outline some new convergence and optimality results.
We then describe in some detail an implementation of the methods using 
the \FETK{} software package, which is an adaptive multilevel 
finite element code designed to solve nonlinear elliptic and
parabolic systems on Riemannian manifolds.
We finish by describing a simplex mesh generation algorithm for compact 
binary objects, and then look at a detailed example showing the use of 
\FETK{} for numerical solution of the constraints.
\end{abstract}


\maketitle

{\tiny\tableofcontents}
\pagestyle{headings}
\markboth{Generating Initial Data in GR using Adaptive FEM}
         {Generating Initial Data in GR using Adaptive FEM}

\section{Introduction}

One of the primary goals of numerical relativity is to compute solutions
of the full Einstein equations and to compare the results of such
calculations to the expected data from the current or next generation of 
interferometric gravitational wave observatories such as LIGO.
Such observatories are predicted to make observations primarily of
``burst'' events, mainly supernovae and collisions between compact
objects, e.g., black holes and neutron stars.
Computer simulation of such events in three space dimensions is a 
challenging task; a sequence of reviews that give a fairly complete
overview of the state of the field of numerical relativity at the time 
they appeared are~\cite{Seid96,CoTe99,Lehn01}.
While there is currently tremendous activity in this area of computational
science, even the mathematical properties of the Einstein system are still 
not completely understood; in fact, even the question of which is the most 
appropriate mathematical formulation of the constrained evolution system 
for purposes of accurate longtime numerical simulation is still being hotly 
debated (cf.~\cite{LiSc03,NOR04}).

As is well known, solutions to the Einstein equations are constrained in
a manner similar to Maxwell's equations, in that the initial data for a
particular spacetime must satisfy a set of purely spatial equations which are
then preserved throughout the evolution (cf.~\cite{MTW70}).
As in electromagnetism, these constraint equations may be put in the form
of a formally elliptic system; however for the Einstein equations, the 
resulting system is a set of four coupled nonlinear equations which are 
generally non-trivial to solve numerically; moreover, the solution theory 
is only partially understood at the moment (cf.~\cite{IsMo96,BaIs03,IsMu03}).
Until recently, most results were developed only in the case of 
constant mean extrinsic curvature (CMC) data, leading to a decoupling
of the constraints that allows them to be analyzed independently.
A complete characterization of CMC solutions on closed manifolds
appears in~\cite{Isen95}; see also the survey~\cite{rBjI04} and the
recent work on weak solutions~\cite{dM05,dM06}.
Of the very few non-CMC results available for the coupled system are
those for closed manifolds in~\cite{IsMo96}; these and all other 
known non-CMC results~\cite{pAaCjI07,jInOM04} hold only in the case that 
the mean extrinsic curvature is nearly constant (referred to as near-CMC).
Some very recent results on weak and strong solutions to the constaints
in the setting of both CMC and near-CMC solutions on compact manifolds 
with boundary appear in~\cite{HKN07a}.
Related results on weak and strong solutions to the constaints
in the setting of CMC, near-CMC, and truly non-CMC (far-from-CMC)
solutions on closed manifolds appear in~\cite{HNT07a,HNT07b}.

Numerous efforts to develop effective numerical techniques for the
constraint equations, and corresponding high-performance implementations,
have been undertaken over the last twenty years; the previously
mentioned numerical relativity reviews~\cite{Seid96,CoTe99,Lehn01} give a 
combined overview of the different discretization and solver technology 
developed to date.
A recent review that focuses entirely on the constraints is~\cite{Pfei04};
this work reviews the conformal decomposition technique and its more recent
incarnations, and also presents an overview of the current state of 
numerical techniques for the constraints in one of the conformal forms.
While most previous work has involved finite difference and spectral
techniques, both non-adaptive and adaptive, there have been previous 
applications of finite element (including adaptive) techniques to the 
scalar Hamiltonian constraint; these 
include~\cite{Mann83,Mann93,AMP97,Mukh96,HSS97}.
An initial approach to the coupled system using adaptive
finite element methods appears in~\cite{BeHo96}.
A complete theoretical analysis of adaptive finite element methods (AFEM)
for a general class of geometric PDE appears in~\cite{Hols2001a,HoTs07a},
and a complete analysis of AFEM for the Einstein constraints appears
in~\cite{HoTs07b}, including proofs of convergence and optimality.

While yielding many useful numerical results and new insights into the 
Einstein equations, most of the approaches previously used for the
constraints suffer from one or more of the following limitations:
\begin{itemize}
\item The $3$--metric must be conformally flat;
\item The extrinsic curvature tensor must be traceless or vanish altogether;
\item The momentum constraint must have an exact solution;
\item The problem must be spherically symmetric or axisymmetric;
\item The domain must be covered by a single coordinate system which
      may contain singularities;
\item The conformal metric must have a scalar curvature which can
      be easily computed analytically;
\item The gauge conditions and/or source terms must be chosen so that the 
      constraints decouple, become linear, or both; 
\item There is very little control of the approximation error, or even
      a rigorously derived {\em a priori} error estimate;
\item One must have access to a large parallel computer in order
      to obtain reasonably accurate solutions on domains of physical interest.
\end{itemize}
In this paper, we describe a general approach based on adaptive finite
element methods that allows one to avoid most of the limitations above.

More precisely, we develop a class of adaptive finite element methods 
for producing high-quality numerical approximations to solutions of the 
constraint equations in the setting of general domain topologies and general
non-constant mean curvature coupling of the two constraint equations.
Our approach is based on the adaptive approximation framework developed
in~\cite{Hols2001a,HoTs07a,HoTs07b}, which is a theoretical framework and 
corresponding implementation for adaptive multilevel finite element methods 
for producing high-quality reliable approximations to solutions of 
nonlinear elliptic systems on manifolds.
As in earlier work, we employ a 3+1 splitting of spacetime and use the
York conformal decomposition formalism to produce a covariant nonlinear
elliptic system on a 3-manifold.
This elliptic system is then written in weak form, which offers various
advantages for developing solution theory, approximation theory, 
and numerical methods.
In particular, this allows us to establish {\em a priori} error estimates
for a broad class of Galerkin-type methods which include not only
finite element methods, but also wavelet-based methods as well as
spectral methods.
Such estimates provide a base approximation theory for establishing
convergence of the underlying discretization approach.
We also derive {\em a posteriori} error estimates which can be
used to drive adaptive techniques such as local mesh refinement.
We outline a class of simplex-based adaptive algorithms for 
weakly-formulated nonlinear elliptic PDE systems, involving 
error estimate-driven local mesh refinement.

We then describe in some detail a particular implementation of the adaptive 
techniques in the software package \FETK{}~\cite{Hols2001a}, 
which is an adaptive 
multilevel finite element code based on simplex elements.
This software is designed to produce provably accurate numerical solutions 
to a large class of nonlinear covariant elliptic systems of tensor
equations on 2- and 3-manifolds in an optimal or nearly-optimal way.
It employs {\em a posteriori} error estimation, adaptive simplex
subdivision, unstructured algebraic multilevel methods, global inexact Newton
methods, and numerical continuation methods for the highly accurate numerical
solution of nonlinear covariant elliptic systems on 2- and 3-manifolds.
The \FETK{} implementation has several unusual features (described in 
Section~\ref{sec:fetk}) which make it ideally suited for solving problems
such as the Einstein constraint equations in an adaptive way.
Applications of \FETK{} to problems in other areas such as biology
and elasticity can be found in~\cite{HBW99,BHW99,BaHo98a}.

\subsection*{Outline of the paper}
We review the classical York conformal decomposition
in Section~\ref{sec:york_decomp}.
In Section~\ref{sec:weak_forms} we give a basic framework for
deriving weak formulations.
In Section~\ref{sec:sobolev_spaces} we briefly outline the
notation used for the relevant function spaces.
In Section~\ref{sec:elliptic_model} we go over a simple weak
formulation example.
In Section~\ref{sec:weakForm_gr} we derive an appropriate
symmetric weak formulation of the coupled constraint equations, 
and summarize a number of basic theoretical results.
In Section~\ref{sec:linearization_gr} we also derive the
linearized bilinear form of the nonlinear weak form for use with
stability analysis or Newton-like numerical methods.
A brief introduction to finite element methods for
nonlinear elliptic systems is presented in Section~\ref{sec:fem_details}.
Adaptive methods are described in Section~\ref{sec:fem_adaptivity},
and residual-type error indicators are derived in 
Section~\ref{sec:fem_residual}.
A derivation of the {\em a posteriori} error indicator for the constraints
is give in Section~\ref{sec:fem_residual_gr}.
We give two give an overview of {\em a priori} error estimates 
from~\cite{Hols2001a,HoTs07a,HoTs07b}
for general Galerkin approximations to solutions equations such as the 
momentum and Hamiltonian constraints in
Section~\ref{sec:fem_apriori}.
The numerical methods employed by \FETK{} are described in detail in
Section~\ref{sec:fetk}, including the finite element discretization,
the residual-based {\em a posteriori} error estimator, the adaptive
simplex bisection strategy, the algebraic multilevel solver, and the
Newton-based continuation procedure for the solution of the nonlinear
algebraic equations which arise.
Section~\ref{sec:fetk} describes a mesh generation algorithm for modeling
compact binary objects,
outlines an algorithm for computing conformal Killing vectors, 
describes the numerical approximation of the ADM mass, and
gives an example showing the use of \FETK{}
for solution of the coupled constraints in the setting of a binary
compact object collision.
The results are summarized in Section~\ref{sec:conclusions}.

\section{The Constraint Equations in General Relativity}
  \label{sec:theory}

\subsection{Notation}
  \label{sec:notation}

Let $(\calg{M},\gamma_{ab})$ be a connected compact Riemannian $d$-manifold 
with boundary $(\partial \calg{M},\sigma_{ab})$,
where the boundary
metric $\sigma_{ab}$ is inherited from $\gamma_{ab}$.
In this paper we are interested only in the Riemannian case, where
we assume $\gamma_{ab}$ is strictly positive {\em a.e.} in $\calg{M}$.
Later we will require some additional assumptions
on $\gamma_{ab}$ such as its smoothness class.
To allow for general boundary conditions, 
we will view the boundary $(d-1)$-submanifold $\partial \calg{M}$
(which we assume to be oriented)
as being formed from two disjoint
submanifolds $\partial_0 \calg{M}$ and $\partial_1 \calg{M}$, 
i.e.,
\begin{equation}
  \label{eqn:bndry}
\partial_0 \calg{M} \cup \partial_1 \calg{M} = \partial \calg{M},
\ \ \ \ \ \ \ \ 
\partial_0 \calg{M} \cap \partial_1 \calg{M} = \emptyset.
\end{equation}
When convenient in the discussions below, one of the two submanifolds
$\partial_0 \calg{M}$ or $\partial_1 \calg{M}$ may be allowed to shrink to
zero measure, leaving the other to cover $\partial \calg{M}$.
An additional technical assumption at times will be non-intersection
of the closures of the boundary sets:
\begin{equation}
  \label{eqn:bndry_overlap}
\overline{\partial_0 \calg{M}} \cap \overline{\partial_1 \calg{M}} = \emptyset,
\end{equation}
This condition is trivially satisfied if either $\partial_0 \calg{M}$ or
$\partial_1 \calg{M}$ shrinks to zero measure.
It is also satisfied in practical situations such as black hole models,
where $\partial_0 \calg{M}$ represents the outer-boundary of a truncated
unbounded manifold, and where $\partial_1 \calg{M}$ represents the 
inner-boundary at one or more black holes.
In any event, in what follows it will usually be necessary to make 
some minimal smoothness assumptions about the entire boundary
submanifold $\partial \calg{M}$, such as Lipschitz continuity
(for a precise definition see~\cite{Adam78}).

We will employ the abstract index notation (cf.~\cite{Wald84}) and
summation convention for tensor expressions below, with indices
running from $1$ to $d$ unless otherwise noted.
Covariant partial differentiation of a tensor
$t^{a_1\cdots a_p}_{~~~~~~~b_1\cdots b_q}$
using the connection provided by the metric $\gamma_{ab}$ will be denoted
as $t^{a_1\cdots a_p}_{~~~~~~~b_1\cdots b_q;c}$
or as $D_c t^{a_1\cdots a_p}_{~~~~~~~b_1\cdots b_q}$.
Denoting the outward unit normal to $\partial \calg{M}$ as $n_b$,
recall the Divergence Theorem for a vector
field $w^b$ on $\calg{M}$ (cf.~\cite{Lee97}):
\begin{equation}
   \label{eqn:divergence_basic}
\int_{\calg{M}} w^b_{~;b} ~dx = \int_{\partial \calg{M}} w^b n_b ~ds,
\end{equation}
where $dx$ denotes the measure on $\calg{M}$ generated by the volume
element of $\gamma_{ab}$:
\begin{equation}
   \label{eqn:volume_element}
dx = \sqrt{\text{det}~\gamma_{ab}}~dx^1 \cdots dx^d,
\end{equation}
and where $ds$ denotes the boundary measure on $\partial \calg{M}$
generated by the boundary volume element of $\sigma_{ab}$.
Making the choice $w^b = u_{a_1\ldots a_k} v^{a_1\ldots a_kb}$
and forming the divergence $w^b_{~;b}$ by applying the product rule
leads to a useful integration-by-parts formula for certain
contractions of tensors:
\begin{equation}
\hspace*{-2.2cm}
   \label{eqn:divergence}
  \int_{\calg{M}}
    u_{a_1\ldots a_k} v^{a_1\ldots a_kb}_{~~~~~~~~~;b} ~dx
= \int_{\partial \calg{M}}
    u_{a_1\ldots a_k} v^{a_1\ldots a_kb} n_b ~ds
- \int_{\calg{M}}
    v^{a_1\ldots a_kb} u_{a_1\ldots a_k;b} ~dx.
\end{equation}
When $k=0$ this reduces to the familiar case where $u$ and $v$ are scalars.

\subsection{The York Decomposition}
  \label{sec:york_decomp}

As discussed in the introduction, the most common form of the initial data
problem for numerical work is the coupled elliptic system obtained from the
York decomposition~\cite{York79,York73}.
We employ the standard notation whereby the
spatial 3--metric is denoted $\gamma_{ab}$ with indices running from 1 to 3.
The Hamiltonian and momentum constraints are
\begin{equation}
R + ({\rm tr} K)^2 - K_{ab} K^{ab} = 16 \pi \rho
\label{Hamiltonian1}
\end{equation}
and
\begin{equation}
D_b (K^{ab} - \gamma^{ab} {\rm tr} K) = 8 \pi j^a,
\label{momentum1}
\end{equation}
where $R$ is the scalar curvature of $\gamma_{ab}$, $K_{ab}$ is the
extrinsic curvature of the initial hypersurface, and $\rho$ and $j^a$ are the
mass density and current.

As is well known, (\ref{Hamiltonian1})~and~(\ref{momentum1}) form an 
under-determined
system for the components of the 3--metric and extrinsic curvature.
The York conformal decomposition method identifies some of these components 
as ``freely specifiable,'' i.e., source terms similar to the matter terms 
$\rho$ and $j^a$, and the remaining components as constrained.
The main strength of the formalism is that it does this in such a way that 
the constraint system becomes a coupled set of nonlinear elliptic equations
in these components.
In only the most general setting do the fully coupled equations have to be 
solved; in many applications they can be decoupled and further simplification
of the problem may eliminate one or more of the non-linearities.
The primary disadvantage of the formalism is the difficulty of controlling 
the physics of the initial data generated; the matter terms as well as the 
unconstrained parts of $\gamma_{ab}$ and $K_{ab}$ are not related directly to 
important physical and geometrical properties of the initial data.
One therefore generally relies on an evolution, i.e., a computation of the
entire spacetime, to determine ``what was in'' the initial data to begin with.
Some of the difficulties with use of the formalism have been recently 
overcome; cf.~\cite{Pfei04} for a recent survey.

The formalism is summarized in a number of 
places (see, e.g.,~\cite{York79,York73,Pfei04}), in this 
section we give only a brief description of the most standard form.
The manifold $\calg{M}$ is endowed with a Riemannian metric 
$\hat{\gamma}_{ab}$ 
and the solution to the initial data problem is, in part, a
metric $\gamma_{ab}$  related to $\hat{\gamma}_{ab}$ by
\begin{equation}
\gamma_{ab} = \phi^{4} \hat{\gamma}_{ab}.
\end{equation}
In what follows all hatted quantities are formed out of $\hat{\gamma}_{ab}$ 
in the usual way, e.g., the covariant derivative $\hat{D}_a$, the Riemann 
tensor $\hat{R}_{abcd}$, etc., while unhatted quantities are formed out 
of $\gamma_{ab}$.

There are four constraint equations and 12 components of $\gamma_{ab}$ and 
$K_{ab}$.  The eight freely specifiable components consist of the conformal 
part of the 3--metric and the trace and the transverse-traceless parts of the 
extrinsic curvature.
The remaining constrained parts are the conformal factor, $\phi$, 
and the ``longitudinal potential,'' $W^a$, of the extrinsic curvature.
Specifically, we decompose the extrinsic curvature tensor as 
\begin{eqnarray}
K^{ab} & = & \phi^{-10} (\Ahatstar^{ab} + (\hat{L}W)^{ab}) + 
\frac{1}{3} \phi^{-4} \hat{\gamma}^{ab} {\rm tr} K,
\\
A^{ab} & = & K^{ab} - \frac{1}{3} \gamma^{ab} {\rm tr}K
  =   \phi^{-10} \hat{A}^{ab},
\\
\hat{A}^{ab} & = & \Ahatstar^{ab} + (\hat{L}W)^{ab},
\\
(\hat{L}W)^{ab} & = & \hat{D}^a W^b + \hat{D}^b W^a 
                - \frac{2}{3} \hat{\gamma}^{ab} \hat{D}_c W^c,
\end{eqnarray}
where ${\rm tr} K = \gamma^{ab}K_{ab}$.
Note that $\hat{A}^{ab}$ and $(\hat{L}W)^{ab}$ are traceless by construction
and $\Ahatstar^{ab}$ is a freely specifiable transverse-traceless tensor.
The matter terms are decomposed as
\begin{equation}
\rho = \hat{\rho}\phi^{-8},
\;\;\;
j^a = \hat{j}^a \phi^{-10}.
\end{equation}
In these variables the Hamiltonian constraint~(\ref{Hamiltonian1}) becomes
\begin{equation}
   \label{eqn:ham}
\hat{\Delta}\phi = \frac{1}{8} \hat{R} \phi 
+ \frac{1}{12} ({\rm tr} K)^2 \phi^5
- \frac{1}{8} (\Ahatstar_{ab} + (\hat{L}W)_{ab})^2 \phi^{-7}
- 2 \pi \hat{\rho} \phi^{-3},
\end{equation}
where $\hat{\Delta}\phi = \hat{D}_a \hat{D}^a \phi$
and where $(T_{ab})^2 = T^{ab}T_{ab}$,
following the notation introduced in~(\ref{eqn:tensor_norm}).
In these variables the momentum constraint~(\ref{momentum1}) becomes
\begin{equation}
\hat{D}_b(\hat{L}W)^{ab} = \frac{2}{3} \phi^6 \hat{D}^a {\rm tr}K 
                         + 8 \pi \hat{j}^a.
   \label{eqn:mom}
\end{equation}

In this article, we are mainly interested in formulations of the
constraints on manifolds with boundary.
The primary motivation for this is the desire to develop techniques for
producing high-quality numerical solutions to the constraints, which require 
the use of mathematical formulations of the constraints involving finite 
domains with boundary.
In order to completely specify the strong (and later, the weak) forms
of the constraints on manifolds with boundary, we need to specify the
boundary conditions.
This is very problem dependent; however, we would like to at least
include the case of the vector Robin condition for asymptotically
flat initial data given in~\cite{YoPi82}.
This is
\begin{equation}
(\hat{L}W)^{bc} \hat{n}_c \left( \delta^a_b - \frac{1}{2}\hat{n}^a
\hat{n}_b \right)
+\frac{6}{7R} W^b \left(\delta^a_b - \frac{1}{8}\hat{n}^a \hat{n}_b
\right)
=O(R^{-3})
\end{equation}
where $R$ is the radius of a large, spherical domain.
The right hand side could be taken to be zero.  Noting that
$\delta^a_b + \hat{n}^a \hat{n}_b$ is the inverse of
$\delta^a_b - 1/2\; \hat{n}^a \hat{n}_b$ we may compute
\begin{equation}
   \label{eqn:robin_york}
(\hat{L}W)^{ab} \hat{n}_b
+\frac{6}{7R} W^b \left(\delta^a_b + \frac{3}{4}\hat{n}^a \hat{n}_b
\right) = 0.
\end{equation}
Hence we will consider the following linear Robin-like condition,
which is general enough to include~(\ref{eqn:robin_york}) and more 
recently proposed boundary conditions~\cite{CoPf04}:
\begin{equation}
   \label{eqn:robin_mom}
(\hat{L}W)^{ab} \hat{n}_b
 + C^a_{~b} W^b = Z^a ~\text{~on~} \partial_1 \calg{M}.
\end{equation}
Similarly, we are interested in analyzing the case of a 
Robin-like boundary condition with the
Hamiltonian constraint:
\begin{equation}
   \label{eqn:robin_ham}
\hat{n}_a \hat{D}^a \phi + c \phi = z ~\text{~on~} \partial_1 \calg{M}.
\end{equation}
Equations (\ref{eqn:ham})--(\ref{eqn:mom}) are known to be well-posed only
for restricted problem data and manifold
topologies~\cite{Lich44,Isen95,MuYo73,MuYo74a,MuYo74b,YoPi82,CIM92,IsMo96,CIY02,BaIs03,IsMu03}; most of the existing results are for the case of
constant mean extrinsic curvature (CMC) data on closed (compact without
boundary) manifolds, with some results for near-CMC data~\cite{IsMo96}.
Some very recent results on weak and strong solutions to the constaints
in the setting of CMC and near-CMC solutions on compact manifolds with
boundary appear in~\cite{HKN07a}.
Related results on weak and strong solutions to the constaints
in the setting of CMC, near-CMC, and truly non-CMC (far from-CMC)
solutions on closed manifolds appear in~\cite{HNT07a,HNT07b}.

\subsection{General Weak Formulations of Nonlinear Elliptic Systems}
\label{sec:weak_forms}

Consider now a general second-order elliptic system
of tensor equations in strong divergence form
over a Riemannian manifold $\calg{M}$ with boundary:
\begin{eqnarray}
\label{eqn:str1}
  - A^{ia}(x^b,u^j,u^k_{~;c})_{;a} + B^i(x^b,u^j,u^k_{~;c})
      & = & 0 ~\text{~in~} \calg{M}, \\ 
\label{eqn:str2}
    A^{ia}(x^b,u^j,u^k_{~;c}) n_a + C^i(x^b,u^j,u^k_{~;c})
      & = & 0 ~\text{~on~} \partial_1 \calg{M}, \\ 
\label{eqn:str3}
u^i(x^b) & = & E^i(x^b) ~\text{~on~} \partial_0 \calg{M},
\end{eqnarray}
where
\begin{equation}
1 \le a,b,c \le d,
\ \ \ 
1 \le i,j,k \le n,
\end{equation}
\begin{equation}
A : \calg{M} \times \bbbb{R}^n \times \bbbb{R}^{nd}
            \mapsto \bbbb{R}^{nd},
\ \ \ 
B : \calg{M} \times \bbbb{R}^n \times \bbbb{R}^{nd}
            \mapsto \bbbb{R}^n,
\end{equation}
\begin{equation}
C : \partial_1 \calg{M} \times \bbbb{R}^n \times \bbbb{R}^{nd} 
            \mapsto \bbbb{R}^n,
\ \ \ 
E : \partial_0 \calg{M} \times \mapsto \bbbb{R}^n.
\end{equation}
The divergence-form system~(\ref{eqn:str1})--(\ref{eqn:str3}), together
with the boundary conditions, can be viewed as 
an operator equation
\begin{equation}
   \label{eqn:parm_pre}
F(u) = 0,
\ \ \ \ \
F : \calg{B}_1 \mapsto \calg{B}_2^*,
\end{equation}
for some Banach spaces $\calg{B}_1$ and $\calg{B}_2$, where $\calg{B}_2^*$
denotes the dual space of $\calg{B}_2$.

Our interest here is primarily in coupled systems
of one or more scalar field equations
and one or more $d$-vector field equations.
The unknown $n$-vector $u^i$ then in general consists
of $n_s$ scalars and $n_v$ $d$-vectors, so that $n = n_s + n_v \cdot d$.
To allow the $n$-component system~(\ref{eqn:str1})--(\ref{eqn:str3})
to be treated notationally as if it were a single $n$-vector equation,
it will be convenient to introduce the following notation for the
unknown vector $u^i$ and for the metric of the product space of
scalar and vector components of $u^i$:
\begin{equation}
   \label{eqn:product_metric}
\hspace*{-1.0cm}
\calg{G}_{ij} = \left[ \begin{array}{ccc}
         \gamma_{ab}^{(1)} &        & 0      \\
                      & \ddots &              \\
         0            &        & \gamma_{ab}^{(n_e)} \\
         \end{array} \right],
~~~~~~
u^i = \left[ \begin{array}{c}
         u^a_{(1)}      \\
         \vdots         \\
         u^a_{(n_e)}      \\
         \end{array} \right],
~~~~~n_e=n_s+n_v.
\end{equation}
If $u^a_{(k)}$ is a $d$-vector
we take $\gamma_{ab}^{(k)} = \gamma_{ab}$; if $u^a_{(k)}$ is a scalar
we take $\gamma_{ab}^{(k)}=1$.

The weak form of~(\ref{eqn:str1})--(\ref{eqn:str3}) is obtained by taking the
$L^2$-inner-product between a vector $v^j$
(vanishing on $\partial_0 \calg{M}$) lying in a product space
of scalars and tensors, and the residual of the
tensor system~(\ref{eqn:str1}), yielding:
\begin{equation}
   \label{eqn:weak_almost}
\int_{\calg{M}} \calg{G}_{ij} \left( B^i - A^{ia}_{\ \ ;a} \right) v^j 
  ~dx = 0.
\end{equation}
Due to the definition of $\calg{G}_{ij}$ in~(\ref{eqn:product_metric}),
this is simply a sum of integrals of scalars, each of which is a contraction
of the type appearing on the left side in~(\ref{eqn:divergence}).
Using then~(\ref{eqn:divergence}) and~(\ref{eqn:str2}) together
in~(\ref{eqn:weak_almost}),
and recalling that $v^i=0$ on $\partial_0 \calg{M}$, yields
\begin{equation}
    \int_{\calg{M}} \calg{G}_{ij} A^{ia} v^j_{~;a}~dx
  + \int_{\calg{M}} \calg{G}_{ij} B^i v^j~dx
  + \int_{\partial_1 \calg{M}} \calg{G}_{ij} C^i v^j~ds
  = 0.
\end{equation}
This gives rise to a covariant weak formulation of the problem:
\begin{equation}
   \label{eqn:weak}
\text{Find}~ u \in \bar{u} + \calg{B}_1 ~\text{s.t.}~
     \langle F(u),v \rangle = 0,
   \ \ \forall~ v \in \calg{B}_2,
\end{equation}
for suitable Banach spaces of functions $\calg{B}_1$ and $\calg{B}_2$,
where the nonlinear weak form $\langle F(\cdot),\cdot \rangle$
can be written as:
\begin{equation}
\label{eqn:weakForm}
  \langle F(u),v \rangle
    = \int_{\calg{M}} \calg{G}_{ij} (A^{ia} v^j_{~;a} + B^i v^j)~dx
           + \int_{\partial_1 \calg{M}} \calg{G}_{ij} C^i v^j ~ds.
\end{equation}
The notation $\langle w,v \rangle$ will represent the duality pairing
of a function $v$ in a Banach space $\calg{B}$ with
a bounded linear functional (or {\em form}) $w$
in the dual space $\calg{B}^*$.
Depending on the particular function spaces involved,
the pairing may be thought of as coinciding with the $L^2$-inner-product
through the Riesz Representation Theorem~\cite{Yosi80}.
The affine shift tensor $\bar{u}$ in~(\ref{eqn:weak})
represents the essential or Dirichlet part
of the boundary condition if there is one; the existence of $\bar{u}$ such that 
$E = \bar{u} |_{\partial_0 \calg{M}}$ in the sense
of the Trace operator is guaranteed by the Trace Theorem
for Sobolev spaces on manifolds with boundary~\cite{Wlok92},
as long as $E^i$ in~(\ref{eqn:str3}) and $\partial_0 \calg{M}$ are smooth
enough (see~\S\ref{sec:sobolev_spaces} below).


The (Gateaux) linearization $\langle DF(u)w,v \rangle$ 
of the nonlinear form $\langle F(u),v \rangle$,
necessary for both local solvability analysis and
Newton-like numerical methods (cf.~\cite{Hols2001a}),
is defined formally as:
\begin{equation}
    \label{eqn:gateaux_definition}
\langle DF(u)w,v \rangle=\left.\frac{d}{d\epsilon}
    \langle F(u+\epsilon w),v \rangle \right|_{\epsilon=0}.
\end{equation}
This form is easily computed from most nonlinear forms
$\langle F(u),v \rangle$ which
arise from second order nonlinear elliptic problems, although the
calculation can be tedious in some cases (as we will see shortly
in the case of the constraints).

The Banach spaces which arise naturally as solution spaces for the
class of nonlinear elliptic systems in~(\ref{eqn:weak}) are product
spaces of the Sobolev spaces $W_{0,D}^{k,p}(\calg{M})$, or the
related Besov spaces $B^{k,p}_q(\calg{M})$.
This is due to the fact that under suitable growth conditions
on the nonlinearities in $F$, it can be shown
(essentially by applying the H\"{o}lder inequality)
that there exists $p_k,q_k,r_k$ satisfying $1 < p_k,q_k,r_k < \infty$
such that the choice
$$
\calg{B}_1 = W_{0,D}^{1,r_1}(\calg{M})
      \times \cdots \times W_{0,D}^{1,r_{n_e}}(\calg{M}),
~~~~
\calg{B}_2 = W_{0,D}^{1,q_1}(\calg{M})
      \times \cdots \times W_{0,D}^{1,q_{n_e}}(\calg{M}),
$$
\begin{equation}
   \label{eqn:banach}
\frac{1}{p_k} + \frac{1}{q_k} = 1,
~~~~
r_k \ge \text{min} \{p_k,q_k\},
~~~~
k=1,\ldots,n_e,
\end{equation}
ensures $\langle F(u),v \rangle$ in~(\ref{eqn:weakForm})
remains finite for all arguments~\cite{FuKu80,Hols2001a}.

\subsection{The Sobolev Spaces $W^{k,p}(\calg{M})$
            and $H^k(\calg{M})$ on Manifolds with Boundary}
  \label{sec:sobolev_spaces}

For a type $(r,s)$-tensor
$T^I_{~J} = T^{a_1a_2 \cdots a_r}_{~~~~~~~~~b_1b_2\cdots b_s}$,
where $I$ and $J$ are (tensor) multi-indices satisfying
$|I|=r$, $|J|=s$, define
\begin{equation}
   \label{eqn:tensor_norm}
|T^I_{~J}| = \left( T^I_{~J} T^L_{~M} \gamma_{IL}\gamma^{JM} \right)^{1/2}.
\end{equation}
Here, $\gamma_{IJ}$ and $\gamma^{IJ}$ are generated from the
Riemannian $d$-metric $\gamma_{ab}$ on $\calg{M}$ as:
\begin{equation}
   \label{eqn:metric_lebesgue}
\gamma_{IJ} = \gamma_{ab}\gamma_{cd} \cdots \gamma_{pq},
~~~~~~~~
\gamma^{IJ} = \gamma^{ab}\gamma^{cd} \cdots \gamma^{pq},
\end{equation}
where $|I|=|J|=m$, producing $m$ terms in each product.
This is just an extension of the Euclidean $l^2$-norm for vectors
in $\bbbb{R}^d$.
For example, in the case of a 3-manifold,
taking $|I|=1$, $|J|=0$, $\gamma_{ab} = \delta_{ab}$, gives
$|K^I_{~J}| = |K^a| = \left( K^a K^b \gamma_{ab} \right)^{1/2}
  = \left( K^a K^b \delta_{ab} \right)^{1/2}
  = \| K^a \|_{l^2(\bbbb{R}^3)}$.

Employing the measure $dx$ on $\calg{M}$ defined
in~(\ref{eqn:volume_element}),
the $L^p$-norm of a tensor on $\calg{M}$, $p \in [1,\infty)$, is defined as:
\begin{equation}
\hspace*{-1.5cm}
   \|T^I_{~J}\|_{L^p(\calg{M})} 
     = \left( \int_{\calg{M}} | T^I_{~J} |^p ~dx \right)^{1/p},
~~~~
\|T^I_{~J}\|_{L^{\infty}(\calg{M})}
    = \text{ess}~\sup_{x \in \calg{M}} |T^I_{~J}(x)|.
\end{equation}
The resulting $L^p$-spaces for $1 \le p \le \infty$ are denoted as:
\begin{equation}
   \label{eqn:lpspaces}
L^p(\calg{M}) = \left\{~ T^I_{~J} ~:~
   \|T^I_{~J}\|_{L^p(\calg{M})} < \infty ~\right\}.
\end{equation}

Covariant (distributional) differentiation of order $m=|L|$
(for some tensor multi-index $L$) using a connection generated by
$\gamma_{ab}$, or generated by possibly a different metric, will be denoted
as either of:
\begin{equation}
   \label{eqn:metric_distribution}
D^m T^I_{~J} = T^I_{~J;L},
\end{equation}
where $m$ should not be confused with a tensor index.
The Sobolev semi-norm of a tensor is defined through~(\ref{eqn:lpspaces}) as:
\begin{equation}
   \label{eqn:sobolev_semi_norm}
|T^I_{~J}|_{W^{m,p}(\calg{M})}^p
= \sum_{|L|=m} \| T^I_{~J;L} \|_{L^p(\calg{M})}^p,
\end{equation}
and the Sobolev norm is subsequently defined as:
\begin{equation}
   \label{eqn:sobolev_norm}
\|T^I_{~J}\|_{W^{k,p}(\calg{M})}
= \left( \sum_{0\le m\le k}
     | T^I_{~J} |_{W^{m,p}(\calg{M})}^p \right)^{1/p}.
\end{equation}
The resulting Sobolev spaces of tensors are then defined as:
$$
W^{k,p}(\calg{M}) = \left\{~ T^I_{~J} ~:~
   \|T^I_{~J}\|_{W^{k,p}(\calg{M})} < \infty ~\right\},
$$
\begin{equation}
   \label{eqn:sobolev_space}
W^{k,p}_0(\calg{M}) 
     = \left\{~ \text{Completion~of}~ C_0^\infty(\calg{M})
    \text{~w.r.t.~} \|\cdot\|_{W^{k,p}(\calg{M})} ~\right\},
\end{equation}
where $C_0^{\infty}(\calg{M})$ is the space of $C^{\infty}$-tensors with
compact support in $\calg{M}$.
The space $W^{k,p}_0(\calg{M})$ is a special case of
$W^{k,p}_{0,D}(\calg{M})$, which can be characterized as:
\begin{equation}
  \hspace*{-1.0cm}
W^{k,p}_{0,D}(\calg{M}) = \left\{~ T^I_{~J} \in W^{k,p} ~:~
             \text{tr}~T^I_{~J;L} = 0 ~\text{on}~ \partial_0 \calg{M},
   ~ |L| \le k-1 ~\right\}.
\end{equation}
The spaces $W^{k,p}(\calg{M})$ and $W^{k,p}_{0,D}(\calg{M})$
are separable ($1\le p < \infty$) and reflexive ($1 < p < \infty$)
Banach spaces.
The dual space of bounded linear functionals on $W^{k,p}(\calg{M})$
can be shown (in the sense of distributions, cf.~\cite{KJF77}) to be
$W^{-k,q}(\calg{M})$, $1/p + 1/q=1$, which is itself a 
Banach space when equipped with the dual norm:
\begin{equation}
  \label{eqn:dual_norm}
\|f\|_{W^{-k,q}(\calg{M})} = \sup_{0\ne u \in W^{k,p}(\calg{M})}
       \frac{|f(u)|}{\|u\|_{W^{k,p}(\calg{M})}},
~~~~ \frac{1}{p} + \frac{1}{q} = 1.
\end{equation}

The Hilbert space special case of $p=2$ is given a simplified notation:
\begin{equation}
H^k(\calg{M}) = W^{k,2}(\calg{M}),
~~~~~~~~
H^{-k}(\calg{M}) = W^{-k,2}(\calg{M}),
\end{equation}
with the same convention used for the various subspaces of $H^k(\calg{M})$
such as $H_0^k(\calg{M})$ and $H_{0,D}^k(\calg{M})$.
The norm on $H^k(\calg{M})$ defined above is induced
by an $L^2$-based inner-product as follows:
$\| T^I_{~J} \|_{H^k(\calg{M})}
   = (T^I_{~J},T^I_{~J})_{H^k(\calg{M})}^{1/2}$,
where
\begin{equation}
    \label{eqn:L2_inner_product}
(T^I_{~J},S^I_{~J})_{L^2(\calg{M})}
     = \int_{\calg{M}}
      T^I_{~J} S^L_{~M} \gamma_{IL}\gamma^{JM} ~dx,
\end{equation}
and where
\begin{equation}
(T^I_{~J},S^I_{~J})_{H^k(\calg{M})}
 = \sum_{0\le m\le k} ( D^m T^I_{~J}, D^m S^I_{~J} )_{L^2(\calg{M})}.
\end{equation}
The Banach spaces $W^{k,p}$ and their various subspaces satisfy various
relations among themselves, with the $L^p$-spaces, and with classical
function spaces such as $C^k$ and $C^{k,\alpha}$.
These relationships are characterized by a collection of results known
as embedding, compactness, density, trace, and related theorems.
A number of these are fundamental to approximation theory for elliptic
equations, and will be recalled when needed below.

\subsection{Weak Formulation Example}
\label{sec:elliptic_model}

Before we derive a weak formulation of the Einstein constraints,
let us consider a simple example to illustrate the idea.
Here we assume the 3--metric to be flat so that $\nabla$ is the ordinary
gradient operator and $\cdot$ the usual inner product.  
Let $\calg{M}$ represent the unit sphere centered at the origin (with a 
single chart inherited from the canonical Cartesian coordinate system
in $\bbbb{R}^3$), and let $\partial \calg{M}$ denote the boundary,
satisfying the boundary assumption in equation~(\ref{eqn:bndry}).
Consider now the following semilinear equation on $\calg{M}$:
\begin{eqnarray}
  \label{eqn:strong}
- \nabla \cdot (a(x) \nabla u(x)) + b(x,u(x)) 
      & = & 0 ~\text{~in~} \calg{M}, \\
  \label{eqn:strong_robin}
n(x) \cdot (a(x) \nabla u(x)) + c(x,u(x))
      & = & 0 ~\text{~on~} \partial_1 \calg{M}, \\
  \label{eqn:strong_dirichlet}
 u(x) & = & f(x) ~\text{~on~} \partial_0 \calg{M},
\end{eqnarray}
where $n(x) : \partial \calg{M} \mapsto \bbbb{R}^d$
is the unit normal to $\partial \calg{M}$, and where
\begin{eqnarray}
a : \calg{M} \mapsto \bbbb{R}^{3 \times 3},
& \ \ \ &
b : \calg{M} \times \bbbb{R} \mapsto \bbbb{R}, \\
c : \partial_1 \calg{M} \times \bbbb{R} \mapsto \bbbb{R},
& \ \ \ &
f : \partial_0 \calg{M} \mapsto \bbbb{R}.
\end{eqnarray}

To produce a weak formulation, we first multiply by a test function
$v \in H^1_{0,D}(\calg{M})$ (the subspace of $H^1(\calg{M})$ which
vanishes on the Dirichlet portion of the boundary $\partial_0 \calg{M}$),
producing:
\begin{equation}
\int_{\calg{M}} 
  \left(-\nabla \cdot (a \nabla u) + b(x, u) \right) v ~dx
 = 0.
\end{equation}
After applying the flat space version of the
divergence theorem, this becomes:
\begin{equation}
  \label{eqn:weak1}
\int_{\calg{M}} (a \nabla u)  \cdot \nabla v ~dx
- \int_{\partial \calg{M}} v (a \nabla u) \cdot n ~ds
+ \int_{\calg{M}} b (x,u) v ~dx
   = 0.
\end{equation}
The boundary integral is reformulated using the boundary conditions
as follows:
\begin{equation}
  \label{eqn:boundary_int}
\int_{\partial \calg{M}} v (a \nabla u) \cdot n ~ds
=
- \int_{\partial_1 \calg{M}} c(x,u) v ~ds.
\end{equation}

If the boundary function $f$ is regular enough so that
$f \in H^{1/2}(\partial_0 \calg{M})$, then from the
Trace Theorem~\cite{Adam78}, there exists $\bar{u} \in H^1(\calg{M})$
such that $f = \bar{u}|_{\partial_0 \calg{M}}$ in the sense of the
Trace operator.
Employing such a function $\bar{u} \in H^1(\calg{M})$,
the weak formulation has the form:
\begin{equation}
   \label{eqn:weak_model}
\text{~Find~} u \in \bar{u} + H_{0,D}^1(\calg{M})
   ~\text{~s.t.~}
     \langle F(u),v \rangle = 0,
~~\forall~ v \in H^1_{0,D}(\calg{M}),
\end{equation}
where from equations~(\ref{eqn:weak1}) and~(\ref{eqn:boundary_int}),
the nonlinear form is defined as:
\begin{equation}
\langle F(u),v \rangle = \int_{\calg{M}} \left(a \nabla u \cdot \nabla v
    + b (x,u) v \right) ~dx
  + \int_{\partial_1 \calg{M}} c(x,u) v~ds.
\end{equation}
The ``weak'' formulation of the problem given by
equation~(\ref{eqn:weak_model}) imposes only one order of
differentiability on the solution $u$, and only in the weak or
distributional sense.
Under suitable growth conditions on the nonlinearities $b$ and $c$,
it can be shown that this weak formulation makes sense, in that the
form $\langle F(\cdot),\cdot \rangle$ is finite for all arguments.

To analyze linearization stability, or to apply a numerical
algorithm such as Newton's method, we will need the
bilinear linearization form $\langle DF(u)w,v \rangle$,
produced as the formal Gateaux derivative of the nonlinear form
$\langle F(u),v \rangle$:
$$
\langle DF(u)w,v \rangle =
   \left.\frac{d}{d\epsilon}
 \langle F(u+\epsilon w),v \rangle \right|_{\epsilon=0}
$$
$$
= \frac{d}{d \epsilon} \left(
  \int_{\calg{M}} \left(a \nabla (u+\epsilon w) \cdot \nabla v
+ b (x,u+\epsilon w) v \right) dx \right.
+ \left.\left. \int_{\partial_1 \calg{M}}
     c(x,u+\epsilon w) vds \right)
    \right|_{\epsilon=0}
$$
\begin{equation}
   \label{eqn:linearize_example}
= \int_{\calg{M}} \left(a \nabla w \cdot \nabla v
+ \frac{\partial b(x,u)}{\partial u} w v \right) ~dx 
+ \int_{\partial_1 \calg{M}}
  \frac{\partial c(x,u)}{\partial u} w v ~ds.
\end{equation}
Now that the nonlinear weak form $\langle F(u),v \rangle$
and the associated bilinear
linearization form $\langle DF(u)w,v \rangle$
are defined as integrals, they can be evaluated using
numerical quadrature to assemble a Galerkin-type discretization;
this is described in some detail below in the case of a finite
element-based Galerkin method.

\subsection{Weak Formulation of the Constraints}
   \label{sec:weakForm_gr}

The Hamiltonian constraint~(\ref{eqn:ham}) as well as the
momentum constraint~(\ref{eqn:mom}), taken separately or as a system,
fall into the class of
second-order divergence-form elliptic systems of tensor
equations in~(\ref{eqn:str1})--(\ref{eqn:str3}).
Therefore, we will follow the same plan as
in~\S\ref{sec:sobolev_spaces} in order to produce the weak
formulation~(\ref{eqn:weak})--(\ref{eqn:weakForm}).
However, we now employ the conformal metric $\hat{\gamma}_{ab}$ from the
preceding section to define the volume element
$dx = \sqrt{\text{det}~\hat{\gamma}_{ab}}~dx^1 dx^2 dx^3$
and the corresponding boundary volume element $ds$,
and for use as the manifold connection for covariant differentiation.
The notation for covariant differentiation using the conformal connection
will be denoted $\hat{D}_a$ as in the previous section,
and the various quantities from~\S\ref{sec:sobolev_spaces}
will now be hatted to denote our use of the conformal metric.
For example, the unit normal to $\partial \calg{M}$ will now
be denoted $\hat{n}^a$.

Consider now the principle parts of the
Hamiltonian and momentum constraint operators of the previous section:
\begin{equation}
\hspace*{-1.0cm}
\hat{\Delta} \phi = \hat{D}_a \hat{D}^a \phi,
\ \ \ \ \ \ 
\hat{D}_b (\hat{L}W)^{ab} = \hat{D}_b ( \hat{D}^a W^b + \hat{D}^b W^a 
                - \frac{2}{3} \hat{\gamma}^{ab} \hat{D}_c W^c).
\label{covStrong}
\end{equation}
Employing the covariant divergence theorem
in equation~(\ref{eqn:divergence})
leads to covariant versions of the Green identities
\begin{equation}
    \label{eqn:green_ham}
 \int_{\calg{M}} \psi \hat{\Delta} \phi  ~dx
+ \int_{\calg{M}} (\hat{D}_a \phi) (\hat{D}^a \psi) ~dx
= \int_{\partial \calg{M}} \hat{n}_a \psi \hat{D}^a \phi ~ds
\end{equation}
and
\begin{equation}
\hspace*{-1.0cm}
    \label{eqn:green_mom}
 \int_{\calg{M}} V_a \hat{D}_b (\hat{L}W)^{ab} ~dx
+  \int_{\calg{M}} (\hat{L}W)^{ab} \hat{D}_b V_a~dx
= \int_{\partial \calg{M}} \hat{n}_b V_a (\hat{L}W)^{ab} ~ds,
\end{equation}
for smooth functions in $C^{\infty}(\calg{M})$.
These identities extend to $W^{1,p}(\calg{M})$ using
a standard approximation argument, since $C^{\infty}(\calg{M})$ is
dense in $W^{1,p}(\calg{M})$ (cf.~Theorem~2.9 in~\cite{Aubi82}).

Due to the symmetries of $(\hat{L}W)^{ab}$ and $\hat{\gamma}^{ab}$, the second
integrand in~(\ref{eqn:green_mom}) can be rewritten in a
completely symmetric form.
To do so, we first borrow the
linear strain or (symmetrized) deformation operator
from elasticity:
\begin{equation}
   \label{eqn:strain}
(\hat{E}V)^{ab} = \frac{1}{2} \left( \hat{D}^b V^a + \hat{D}^a V^b \right).
\end{equation}
We can then write the operator $(\hat{L}W)^{ab}$ in terms
of $(\hat{E}W)^{ab}$ as follows:
\begin{equation}
\hspace*{-1.0cm}
   \label{eqn:elasticity_operator}
(\hat{L}W)^{ab} = 2 \mu (\hat{E}W)^{ab}
   + \lambda \hat{\gamma}^{ab} \hat{D}_c W^c,
~~~ ~\text{with}~ \mu = 1 ~\text{and}~ \lambda = - \frac{2}{3}.
\end{equation}
This makes it clear that the momentum constraint operator
 $\hat{D}_b (\hat{L}W)^{ab}$
is a covariant version of the linear elasticity operator
for a homogeneous isotropic material (cf.~\cite{Ciar88}),
for a particular choice of Lam\'{e} constants.
In particular, in the flat space case where
$\hat{\gamma}_{ab} \equiv \delta_{ab}$, the operator becomes the usual
linear elasticity operator:
\begin{equation}
    \label{eqn:elasticity}
 \hspace*{-1.5cm}
 \hat{D}_b (\hat{L}W)^{ab}
  \rightarrow
  \left( 2 \mu e_{ab}(W) + \lambda e_{cc}(W) \delta_{ab} \right)_{,b},
~~~ e_{ab}(W) = \frac{1}{2} \left( W_{a,b} + W_{b,a} \right),
\end{equation}
where $W_{a,b} = \partial W_a/\partial x^b$
denotes regular partial differentiation.
Employing the operator $(\hat{E}W)^{ab}$ leads to a
symmetric expression in the Green identity~(\ref{eqn:green_mom}):
\begin{eqnarray}
\hspace*{-1.0cm}
(\hat{L}W)^{ab} \hat{D}_b V_a
  & = & \frac{1}{2} \left( (\hat{L}W)^{ab} \hat{D}_b V_a
                      + (\hat{L}W)^{ab} \hat{D}_a V_b \right)
\\
  & = & (\hat{L}W)^{ab} (\hat{E}V)_{ab}
\\
  & = & 2 \mu (\hat{E}W)^{ab} (\hat{E}V)_{ab}
        + \frac{1}{2} \lambda \hat{\gamma}^{ab} \hat{D}_c W^c
          \left( \hat{D}_b V_a + \hat{D}_a V_b \right)
\\
  & = & 2 \mu (\hat{E}W)^{ab} (\hat{E}V)_{ab}
        + \lambda \hat{D}_a W^a \hat{D}_b V^b.
   \label{eqn:sym_trick}
\end{eqnarray}
While it is clear by inspection that the first operator in~(\ref{covStrong})
is formally self-adjoint with respect to the covariant $L^2$-inner-product
defined in~(\ref{eqn:L2_inner_product}),
reversing the procedure in~(\ref{eqn:sym_trick})
implies that the same is true for the second operator in~(\ref{covStrong}).
In other words, the following holds (ignoring the boundary terms):
\begin{eqnarray}
( \hat{\Delta} \phi, \psi )_{L^2(\calg{M})}
& = & (\phi ,  \hat{\Delta} \psi )_{L^2(\calg{M})},
~~~~~~~~~~~~~~\forall \phi,\psi,
\\
 ( \hat{D}_b (\hat{L}W)^{ab}, V^a )_{L^2(\calg{M})}
& = & (W^a ,  \hat{D}_b (\hat{L}V)^{ab})_{L^2(\calg{M})},
~~~\forall  W^a,V^a.
\end{eqnarray}

To make it possible to write the Hamiltonian constraint and various
related equations in a concise way, we now introduce the following
nonlinear function $P(\phi)=P(\phi,W^a,x^b)$, where the explicit dependence on
$x^b$ (and sometimes also the dependence on $W^a$) is suppressed to simplify
the notation:
\begin{eqnarray}
\hspace*{-1.0cm}
P(\phi)
        & = & \frac{1}{16} \hat{R} \phi^2
          + \frac{1}{72} ({\rm tr} K)^2 \phi^6
          + \frac{1}{48} (\Ahatstar_{ab} + (\hat{L}W)_{ab})^2 \phi^{-6}
          + \pi \hat{\rho} \phi^{-2}.
  \label{eqn:p_def}
\end{eqnarray}
The first and second functional derivatives with respect to $\phi$ are
then as follows:
\begin{eqnarray}
\hspace*{-1.0cm}
P'(\phi) & = & \frac{1}{8} \hat{R} \phi
         + \frac{1}{12} ({\rm tr} K)^2 \phi^5
         - \frac{1}{8} (\Ahatstar_{ab} + (\hat{L}W)_{ab})^2 \phi^{-7}
         - 2 \pi \hat{\rho} \phi^{-3},
  \label{eqn:dp_def}
\\
\hspace*{-1.0cm}
P''(\phi) & = & \frac{1}{8} \hat{R}
          + \frac{5}{12} ({\rm tr} K)^2 \phi^4
          + \frac{7}{8} (\Ahatstar_{ab} + (\hat{L}W)_{ab})^2 \phi^{-8}
          + 6 \pi \hat{\rho} \phi^{-4}.
  \label{eqn:ddp_def}
\end{eqnarray}
When working with the Hamiltonian constraint separately, we will often
use these expressions involving $P(\cdot)$; when working with the coupled
system, we will usually write the polynomial explicitly in order
to indicate the coupling terms.

We now take the inner product
of~(\ref{eqn:ham}) with a test function $\psi$, which
we assume to vanish on $\partial_0 \calg{M}$.
(Again, $\partial_0 \calg{M}$ may have zero measure, or it may
be all or only a piece of $\partial \calg{M}$.)
After use of the Green identity~(\ref{eqn:green_ham}) 
and the Robin boundary condition~(\ref{eqn:robin_ham})
we obtain the following form
$\langle F_{\text{H}}(\phi),\psi \rangle$
(nonlinear in $\phi$, but linear in $\psi$)
for use in the weak formulation in~(\ref{eqn:weak}):
\begin{equation}
\hspace*{-1.0cm}
   \label{eqn:weakHam}
\langle F_{\text{H}}(\phi),\psi \rangle =
    \int_{\partial_1 \calg{M}} (c\phi - z) \psi ~ds
+ \int_{\calg{M}} P^{\prime}(\phi) \psi ~dx
+ \int_{\calg{M}} \hat{D}_a \phi \hat{D}^a \psi~dx.
\end{equation}
For the momentum constraint we take the inner product of~(\ref{eqn:mom}) with 
respect to a test vector field $V^a$ (again assumed to vanish on
$\partial_0 \calg{M}$) and similarly use the
Green identity~(\ref{eqn:green_mom}) and the
Robin condition~(\ref{eqn:robin_mom}) to obtain a
form $\langle F_{\text{M}}(W^a),V^a \rangle$
(linear in both $W^a$ and $V^a$ in this case)
having the expression
$$
\langle F_{\text{M}}(W^a),V^a \rangle =
  \int_{\partial_1 \calg{M}} \left( C^a_{~b} W^b - Z^a \right) V_a ~ds
+ \int_{\calg{M}} 
    \left(
      \frac{2}{3} \phi^6 \hat{D}^a {\rm tr}K
    + 8 \pi \hat{j}^a
    \right) V_a ~dx
$$
\begin{equation}
   \label{eqn:weakMom}
+ \int_{\calg{M}}
  \left( 2 \mu (\hat{E}W)^{ab} (\hat{E}V)_{ab}
  + \lambda \hat{D}_a W^a \hat{D}_b V^b \right)
 ~dx,
\end{equation}
where we have used~(\ref{eqn:sym_trick}).
We will take $\mu = 1$ and $\lambda = -2/3$ in~(\ref{eqn:weakMom}),
but for the moment we will leave them
unspecified for purposes of the discussion below.

Ordering the Hamiltonian constraint first in the 
system~(\ref{eqn:str1}), and defining the product metric
$\calg{G}_{ij}$ and the vectors $u^i$ and $v^j$ appearing
in~(\ref{eqn:product_metric}) and~(\ref{eqn:weakForm}) as:
\begin{equation}
\calg{G}_{ij} = \left[ \begin{array}{ccc}
         1 & 0      \\
         0 & \gamma_{ab} \\
         \end{array} \right],
~~~~~
u^i = \left[ \begin{array}{c}
         \phi \\
         W^a  \\
         \end{array} \right],
~~~~~
v^j = \left[ \begin{array}{c}
         \psi \\
         V^b  \\
         \end{array} \right],
\end{equation}
produces a single nonlinear weak form for the coupled
constraints in the form required in~(\ref{eqn:weak}), where
$$
\langle F(u),v \rangle
= \langle F([\phi,W^a]),[\psi,V^a] \rangle
= \langle F_{\text{H}}(\phi),\psi \rangle
+ \langle F_{\text{M}}(W^a),V^a \rangle
$$
$$
=
  \int_{\partial_1 \calg{M}}
    \left( \left[c\phi - z\right]\psi
         + \left[C^a_{~b} W^b - Z^a\right] V_a \right) ~ds
+  \int_{\calg{M}}
   \left(
      \frac{2}{3} \phi^6 \hat{D}^a {\rm tr}K
    + 8 \pi \hat{j}^a
    \right) V_a ~dx
$$
$$
+  \int_{\calg{M}}
   \left(
      \frac{1}{8} \hat{R} \phi
    + \frac{1}{12} ({\rm tr} K)^2 \phi^5
    - \frac{1}{8} (\Ahatstar_{ab} + (\hat{L}W)_{ab})^2 \phi^{-7}
    - 2 \pi \hat{\rho} \phi^{-3}
   \right) \psi ~dx
$$
\begin{equation}
   \label{eqn:weakHamMom}
+  \int_{\calg{M}}
   \left(   \hat{D}_a \phi \hat{D}^a \psi
        + 2 \mu (\hat{E}W)^{ab} (\hat{E}V)_{ab}
       + \lambda \hat{D}_a W^a \hat{D}_b V^b
   \right)~dx.
\end{equation}

While we have completely specified the weak form of the separate and
coupled constraints on a manifold with boundary in a formal sense,
they can be shown to be well-defined (and individually well-posed)
in a more precise mathematical sense; see~\cite{HKN07a,HNT07a,HNT07b}
for an analysis and a survey of the collection of existence, uniqueness,
and stability results.
For a particular situation, we must specify the particular combination
of the boundary
conditions~(\ref{eqn:str1})--(\ref{eqn:str2}) on a splitting of
the boundary ($\partial \calg{M}$) into Dirichlet ($\partial_0 \calg{M}$) and
Robin ($\partial_1 \calg{M}$) parts.
This is quite problem dependent; in numerical simulation
one typically computes solutions to the constraints in the interior of 
a large box or sphere.
On the surface of the sphere one employs Robin and vector Robin conditions
similar to those given in~\cite{YoPi82}, which fit the framework
in~(\ref{eqn:robin_mom}) and~(\ref{eqn:robin_ham}).
In addition, one often constructs black holes topologically by requiring
the conformal metric to obey an isometry through one or more smaller
non-overlapping spheres internal to the domain boundary.
The isometry generates a boundary condition on the conformal factor
which is well understood only when $\hat{\gamma}_{ab}$ is flat.
Even in this case, the exact corresponding boundary condition on $W^a$
is not known, but is likely to appear in the
form of~(\ref{eqn:robin_mom}).
Solvability of both constraints
rests delicately on the boundary condition choices made.

\subsection{Gateaux Linearization of the Weak Formulation}
   \label{sec:linearization_gr}

We will take the formal Gateaux-derivative of the nonlinear form
$\langle F(\cdot),\cdot \rangle$ in equation~(\ref{eqn:weakHamMom}) above,
to produce
a linearization form for use in local solvability analysis through the
Implicit Function Theorem, and for use
in Newton-like iterative solution methods (cf.~\cite{Hols2001a}).
Defining an arbitrary variation direction $w=[\xi,X^a]$, we compute
the Gateaux-derivative of the nonlinear form as follows:
$$
\langle DF([\phi,W^a])[\xi,X^a],[\psi,V^a] \rangle
   = \left.\frac{d}{d\epsilon}
\langle F([\phi+\epsilon \xi,W^a + \epsilon X^a]),[\psi, V^a] \rangle
     \right|_{\epsilon=0}
$$
$$
\hspace*{-1.0cm}
=
  \left.\frac{d}{d\epsilon}
           \int_{\partial_1 \calg{M}}
           \left( \left[c(\phi + \epsilon \xi) - z\right] \psi
                  + \left[C^a_{~b} (W^b + \epsilon X^b) - Z^a\right] V_a
           \right) ~ds
  \right|_{\epsilon=0}
$$
$$
+ \left.\frac{d}{d\epsilon}
           \int_{\calg{M}}
           \left(   \hat{D}_a (\phi+\epsilon \xi) \hat{D}^a \psi
        + 2 \mu (\hat{E}(W + \epsilon X))^{ab} (\hat{E}V)_{ab}
  \right.  \right.
$$
$$
  \left.  \left.
       + \lambda \hat{D}_a (W^a + \epsilon X^a) \hat{D}_b V^b
           \right)~dx
  \right|_{\epsilon=0}
$$
$$
+ \left.\frac{d}{d\epsilon}
     \int_{\calg{M}}
       \left(
           \frac{2}{3} (\phi + \epsilon \xi)^6 \hat{D}^a {\rm tr}K
         + 8 \pi \hat{j}^a
         \right) V_a ~dx
  \right|_{\epsilon=0}
$$
$$
+ \left.\frac{d}{d\epsilon}
          \int_{\calg{M}}
             \left(
             \frac{1}{8} \hat{R} (\phi + \epsilon \xi)
           + \frac{1}{12} ({\rm tr} K)^2 (\phi + \epsilon \xi) ^5
           \right. \right.
$$
\begin{equation}
\hspace*{-1.0cm}
   \left.  \left.
           - \frac{1}{8} (\Ahatstar_{ab} + (\hat{L}(W+\epsilon X))_{ab})^2
                (\phi + \epsilon \xi)^{-7}
           - 2 \pi \hat{\rho} (\phi + \epsilon \xi)^{-3}
          \right) \psi ~dx
  \right|_{\epsilon=0}.
\end{equation}
After some simple manipulations using the product and chain rules,
we are left with the following bilinear form
(for fixed $[\phi,W^a]$), linear separately in each of the variables
$[\xi,X^a]$ and $[\psi,V^a]$:
$$
\hspace*{-2.0cm}
\langle DF([\phi,W^a])[\xi,X^a],[\psi,V^a] \rangle
=
  \int_{\partial_1 \calg{M}}
    \left( c \xi \psi
         + C^a_{~b} X^b V_a \right) ~ds
$$
$$
+ \int_{\calg{M}}
    \left( \hat{D}_a \xi \hat{D}^a \psi
        + 2 \mu (\hat{E}X)^{ab} (\hat{E}V)_{ab}
       + \lambda \hat{D}_a X^a \hat{D}_b V^b
           \right)~dx
$$
$$
+ \int_{\calg{M}}
             \left(
             \frac{1}{8} \hat{R}
           + \frac{5}{12} ({\rm tr} K)^2 \phi^4
           + \frac{7}{8} (\Ahatstar_{ab} + (\hat{L}W)_{ab})^2 \phi^{-8}
           + 6 \pi \hat{\rho} \phi^{-4}
          \right) \xi \psi ~dx
$$
\begin{equation}
\hspace*{-2.0cm}
   \label{eqn:weakHamMom_linearized}
- \int_{\calg{M}}
     \left(
     \frac{1}{4} (\Ahatstar_{ab} + (\hat{L}W)_{ab}) \phi^{-7} 
     \right) 
       (\hat{L}X)^{ab} \psi ~dx
+ \int_{\calg{M}}
       \left(
           4 \phi^5 \hat{D}^a {\rm tr}K
         \right) V_a \xi ~dx.
\end{equation}
Note that the first two volume integrals and the surface integral
are completely symmetric in their
arguments, and represent the symmetric part of the bilinear form.
The third and fourth volume integrals are nonsymmetric in their arguments;
the third volume integral represents the linearized coupling of $W^a$ 
into the Hamiltonian constraint, and the fourth volume integral represents
the linearized coupling of the conformal factor $\phi$ into the
momentum constraint.

\subsection{Weak Formulations Arising from Energy Functionals}
   \label{sec:energy_gr}

Due to the fact that the principle parts of the Hamiltonian and
momentum operators produced by the conformal decomposition are self-adjoint,
the weak formulations individually
arise naturally as the Euler conditions for stationarity of associated
(energy) functionals.
It is straight-forward to verify that the following energy functionals
\begin{equation}
  \label{eqn:ham_energy}
\hspace*{-2.0cm}
J_{\text{H}}(\phi) =
  \frac{1}{2} \int_{\partial_1 \calg{M}} c(\phi - z) \phi ~ds
+ \int_{\calg{M}} P(\phi,W^a) ~dx
+ \frac{1}{2}
  \int_{\calg{M}} \hat{D}_a \phi \hat{D}^a \phi ~dx,~~
\end{equation}
\begin{equation}
\hspace*{-2.0cm}
J_{\text{M}}(W^a) = 
  \frac{1}{2}
   \int_{\partial_1 \calg{M}} \left( C^a_{~b} W^b - Z^a \right) W_a ~ds
+ \int_{\calg{M}} 
    \left( \frac{2}{3} \phi^6 \hat{D}^a {\rm tr}K + 8 \pi \hat{j}^a
    \right) W_a ~dx
\end{equation}
\begin{equation}
+ \int_{\calg{M}}
  \frac{1}{2} \left( 2 \mu (\hat{E}W)^{ab} (\hat{E}W)_{ab}
  + \lambda \hat{D}_a W^a \hat{D}_b W^b \right)
 ~dx,
  \label{eqn:mom_energy}
\end{equation}
each separately give rise to the individual weak
forms~(\ref{eqn:weakHam}) and~(\ref{eqn:weakMom}), respectively.
One computes the Gateaux derivative of $J_{\text{H}}(\phi,W^a)$
with respect to $\phi$, and the Gateaux derivative
of $J_{\text{M}}(\phi,W^a)$ with respect to $W^a$,
and then sets them to zero:
\begin{eqnarray}
\left. \frac{d}{d \epsilon} J_{\text{H}}(\phi + \epsilon \psi)
   \right|_{\epsilon=0}
& = & \langle J_{\text{H}}^{\prime}(\phi),\psi \rangle
       = \langle F_{\text{H}}(\phi),\psi \rangle = 0,
\label{eqn:euler_ham}
\\
\left. \frac{d}{d \epsilon} J_{\text{M}}(W^a + \epsilon V^a)
   \right|_{\epsilon=0}
& = & \langle J_{\text{M}}^{\prime}(W^a),V^a \rangle
       = \langle F_{\text{M}}(W^a),V^a \rangle = 0.
\label{eqn:euler_mom}
\end{eqnarray}
This discussion is only formal, but can be made rigorous.

Unfortunately, while the individual constraints each arise as
Euler conditions for stationarity of the separate energy functionals above,
the coupled constraints do not arise in this way from any coupled energy.
This follows easily from the fact that the combined linearization
bilinear form in~(\ref{eqn:weakHamMom_linearized}) is not symmetric.
This can also be verified directly by considering the most general
possible expression for the total energy:
\begin{equation}
J_{total}(\phi,W^a) = J_{\text{H}}(\phi,W^a) + J_{\text{M}}(\phi,W^a)
                    + J_{\text{R}}(\phi,W^a),
\end{equation}
where $J_{\text{H}}(\phi,W^a)$ and $J_{\text{M}}(\phi,W^a)$
are as defined above,
and where $J_{\text{R}}(\phi,W^a)$ is the remainder term in the energy which
must account for the coupling terms in the combined weak
form~(\ref{eqn:weakHamMom}).
It is easy to verify that the Euler condition for stationarity places
separate conditions on the Gateaux derivative of
$J_{\text{R}}(\phi,W^a)$ which are
impossible to meet simultaneously.

This lack of a variational principle for the coupled constraints limits
the number of techniques available for analyzing solvability of the
coupled system; the existing near-CMC results~\cite{IsMo96,HKN07a} and
the non-CMC (far-from-CMC) results~\cite{HNT07a,HNT07b} are actually based 
on fixed-point arguments; however, variational arguments are used
to solve the individual constaint equations in~\cite{HKN07a,HNT07a,HNT07b}
as part of the overall fixed-point argument.

\section{Adaptive Finite Element Methods (AFEM)}
  \label{sec:fem}

In this section we give a brief description of Galerkin methods,
finite element methods, and adaptive techniques for covariant 
nonlinear elliptic systems.
We also derive {\em a posteriori} error indicators for driving
adaptivity, and finish the section with some {\em a priori} error
estimates for general Galerkin approximations to the Hamiltonian
and momentum constraints.
Expanded versions of this material, including proofs of all results,
can be found in~\cite{Hols2001a,HoTs07a,HoTs07b}.

\subsection{Petrov-Galerkin Methods, Galerkin Methods, 
            and Finite Element Methods}
  \label{sec:fem_details}

A {\em Petrov-Galerkin} approximation of the solution to
(\ref{eqn:weak}) is the solution to the following subspace problem:
\begin{equation}
\text{Find}~ (u_h - \bar{u}_h) \in U_h \subset \calg{B}_1 ~\text{s.t.}~
     \langle F(u_h),v \rangle = 0,
   \label{eqn:galerkin}
\end{equation}
$$
   \forall~ v \in V_h \subset \calg{B}_2,
$$
for some chosen subspaces $U_h$ and $V_h$,
where $\text{dim}(U_h) = \text{dim}(V_h) = n$, and where the
discrete Dirichlet function $\bar{u}_h$ approximates $\bar{u}$
(e.g. an interpolant).
A {\em Galerkin} approximation refers to the case that $U_h = V_h$.

A {\em finite element} method is simply a Petrov-Galerkin or Galerkin
method in which the subspaces $U_h$ and $V_h$ are chosen
to have the extremely simple form of continuous piecewise polynomials
with local support, defined over a disjoint covering of the domain 
manifold $\calg{M}$ by {\em elements}.
For example, in the case of continuous piecewise linear polynomials 
defined over a disjoint covering with 2- or 3-simplices
(cf. Figure~\ref{fig:element}),
the basis functions are easily defined element-wise using the unit
2-simplex (triangle) and unit 3-simplex (tetrahedron) as follows:
$$
\begin{array}{ccl}
\tilde{\phi}_0(\tilde{x},\tilde{y}) & = & 1 - \tilde{x} - \tilde{y} \\
\tilde{\phi}_1(\tilde{x},\tilde{y}) & = & \tilde{x} \\
\tilde{\phi}_2(\tilde{x},\tilde{y}) & = & \tilde{y} \\
\end{array}
\hspace*{0.2cm}
\begin{array}{ccl}
\tilde{\phi}_0(\tilde{x},\tilde{y},\tilde{z}) & = & 1
                   - \tilde{x} - \tilde{y} - \tilde{z} \\
\tilde{\phi}_1(\tilde{x},\tilde{y},\tilde{z}) & = & \tilde{y} \\
\tilde{\phi}_2(\tilde{x},\tilde{y},\tilde{z}) & = & \tilde{x} \\
\tilde{\phi}_3(\tilde{x},\tilde{y},\tilde{z}) & = & \tilde{z} \\
\end{array}.
$$
\begin{figure}[tbh]
\begin{center}
\mbox{\myfig{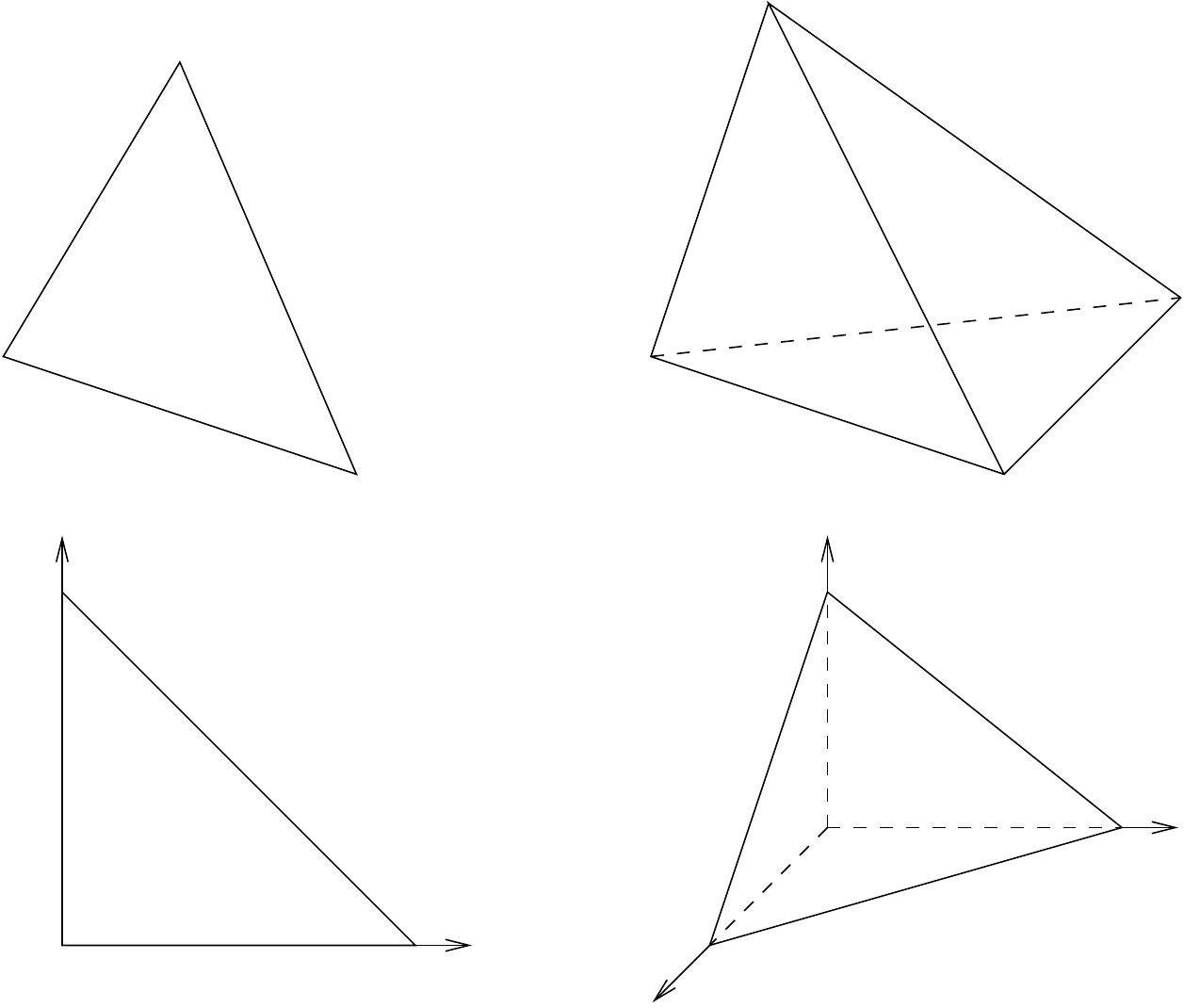}{1.2in}}
\hspace*{0.2cm}
\mbox{\myfig{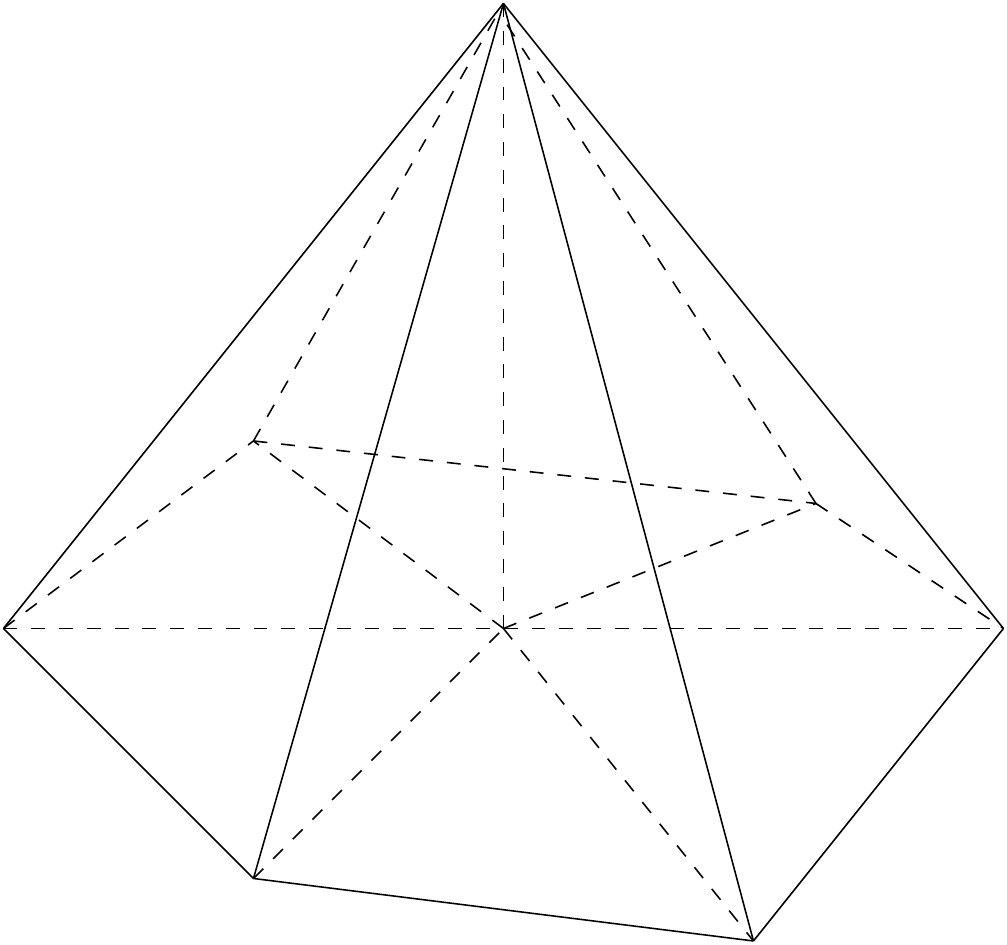}{1.2in}}
   \caption{Reference and arbitrary 2- and 3-simplex elements,
            and a global (2D) basis function.} 
\end{center}
   \label{fig:element}
\end{figure}
Global basis functions are then defined, as in the right-most picture in
Figure~\ref{fig:element}, by combining the support regions around a given
vertex, and extending the unit simplex basis functions to each arbitrary
simplex using coordinate transformations.
If the manifold domain can be triangulated exactly with simplex elements,
then the coordinate transformations are simply affine transformations.
Note that in this sense, finite element methods are by their very nature
defined in a chart-wise manner.
Quadratic and high-order basis functions are defined analogously.

The above basis functions clearly do not form any subspace of 
$\calg{C}^2(\calg{M})$, the space of twice continuously differentiable
functions on $\calg{M}$, which is the natural function space in which to
look for the solutions to second order elliptic equations.
This is due to the fact that they are discontinuous along simplex faces and
simplex vertices in the disjoint simplex covering of $\calg{M}$.
However, one can show~\cite{Ciar78} that in fact:
\begin{equation}
V_h = \text{span}\{\phi_1,\ldots,\phi_n\} \subset W^{1,p}_{0,D}(\calg{M}),
~~ \calg{M} \subset \bbbb{R}^d,
\end{equation}
so that these continuous, piecewise defined,
low-order polynomial spaces do in fact form a subspace of the solution space
to the weak formulation of the class of second order elliptic equations of
interest.
Making then the choice
$
U_h = \text{span}\{\phi_1,\phi_2,\ldots,\phi_n\},
$
$
V_h = \text{span}\{\psi_1,\psi_2,\ldots,\psi_n\},
$
equation~(\ref{eqn:galerkin}) reduces to a set of $n$ nonlinear algebraic 
relations (implicitly defined) for the $n$ coefficients $\{\alpha_j\}$
in the expansion
\begin{equation}
   \label{eqn:soln}
u_h = \bar{u}_h + \sum_{j=1}^n \alpha_j \phi_j.
\end{equation}
In particular, regardless of the complexity of the form
$\langle F(u),v \rangle$,
as long as we can evaluate it for given arguments $u$ and $v$, then we
can evaluate the nonlinear discrete residual of the finite element 
approximation $u_h$ as:
\begin{equation}
   \label{eqn:disc_residual}
    r_i = \langle F(\bar{u}_h + \sum_{j=1}^n \alpha_j \phi_j),\psi_i \rangle,
         \ \ \ \ i=1,\ldots,n.
\end{equation}
Since the form $\langle F(u),v \rangle$
involves an integral in this setting, if we
employ quadrature then we can simply sample the integrand at quadrature points;
this is a standard technique in finite element technology.
Given the local support nature of the functions $\phi_j$ and $\psi_i$,
all but a small constant number of terms in the sum 
$\sum_{j=1}^n \alpha_j \phi_j$ are zero at a particular spatial point
in the domain, so that the residual $r_i$ is inexpensive to evaluate when
quadrature is employed.

The two primary issues in applying this approximation method are then:
\begin{enumerate}
\item The approximation error $\|u-u_h\|_X$, for various norms $X$, and
\item The computational complexity of solving the $n$ algebraic equations.
\end{enumerate}
The first of these issues represents the core of finite element approximation
theory, which itself rests on the results of classical approximation theory.
Classical references to both topics include~\cite{Ciar78,DeLo93,Davi63}.
The second issue is addressed by the complexity theory of direct and
iterative solution methods for sparse systems of linear and nonlinear
algebraic equations, cf.~\cite{Hack94,OrRh70}, and by the use of
adaptive techniques to minimize the size $n$ of the discrete space 
that must be constructed to reach a specific approximation quality.

\subsection{{\em A Priori} Error Estimates for the Constraint Equations}
  \label{sec:fem_apriori}

We first outline an approximation result for general Galerkin approximations
to solutions of the momentum constraint.
It is referred to as a quasi-optimal error estimate, in that it establishes
a bound on the error $\|u-u_h\|_X$ that is within a constant of being the
error in the best possible approximation.

To understand this result, we begin with the two (Hilbert vector) spaces
$H$ and $V$, where in our setting of the momentum constraint,
we have $H=L^2(\calg{M})$ and $V=H^1_{0,D}(\calg{M})$.
We will stay with the abstract notation involving $H$ and $V$ for clarity.
The weak form of the momentum constraint can be shown to have
the following form:
\begin{equation}
  \label{eqn:mom_gal_cont}
\text{Find}~u \in V ~\text{s.t.}~
      A(u,v) = F(v), ~~\forall v \in V,
\end{equation}
where the bilinear form
$A(u,v) : V \times V \mapsto \bbbb{R}$ is bounded
\begin{equation}
  \label{eqn:mom_as2}
A(u,v) \le M \| u \|_V \| v \|_V,
      ~~\forall u,v \in V,
\end{equation}
and V-coercive (satisfying a G{\aa}rding inequality):
\begin{equation}
  \label{eqn:mom_as3}
m \|u\|_V^2 \le K \|u\|_H^2 + A(u,u),
      ~~\forall u \in V, ~~~\text{where}~~ m > 0,
\end{equation}
and where the linear functional $F(v) : V \mapsto \bbbb{R}$ is bounded
and thus lies in the dual space $V^*$:
$$
F(v) \le L \| v \|_V,
      ~~\forall v \in V.
$$
It can be shown that the weak formulation of the
momentum constraint
fits into this framework; to simplify the discussion, we have
assumed that any Dirichlet function $\bar{u}$ has been absorbed
into the linear functional $F(v)$ in the obvious way.
Our discussion can be easily modified to include approximation
of $\bar{u}$ by $\bar{u}_h$.

Now, we are interested in the quality of a Galerkin approximation:
\begin{equation}
\hspace*{-1.0cm}
  \label{eqn:mom_gal_disc}
\text{Find}~u_h \in V_h \subset V~\text{s.t.}~
      A(u_h,v) = A(u,v) = F(v), ~~\forall v \in V_h \subset V.
\end{equation}
We will assume that there exists a sequence of approximation subspaces
$V_h \subset V$ parameterized by $h$, with $V_{h_1} \subset V_{h_2}$
when $h_2 < h_1$, and that there exists a sequence $\{a_h\}$, with
$\lim_{h\rightarrow 0} a_h = 0$, such that
\begin{equation}
\hspace*{-1.0cm}
  \label{eqn:mom_as4}
\|u-u_h\|_H \le a_h \|u-u_h\|_V, 
   ~\text{when}~ A(u-u_h,v) = 0, ~\forall v \in V_h \subset V.
\end{equation}
The assumption~(\ref{eqn:mom_as4}) is very natural;
in our setting, it is the assumption
that the error in the approximation converges to zero more quickly in the
$L^2$-norm than in the $H^1$-norm.
This is easily verified in the setting of piecewise polynomial approximation
spaces, under very mild smoothness requirements on the solution~$u$.
Under these assumptions, we have the following {\em a priori} error estimate.
\begin{theorem}
   \label{thm:mom_approx}
For $h$ sufficiently small, there exists a unique approximation
$u_h$ satisfying~(\ref{eqn:mom_gal_disc}),
for which the following quasi-optimal
{\em a priori} error bounds hold:
\begin{eqnarray}
\|u-u_h\|_V 
   & \le & C \inf_{v \in V_h} \|u-v\|_V,
\label{eqn:schatz} \\
\|u-u_h\|_H 
   & \le & C a_h \inf_{v \in V_h} \|u-v\|_V,
\label{eqn:schatz_L2}
\end{eqnarray}
where $C$ is a constant independent of $h$.
If $K \le 0$ in~(\ref{eqn:mom_as3}), then the above holds for all $h$.
\end{theorem}
\begin{proof}
See~\cite{Hols2001a,HoTs07a,HoTs07b}; also~\cite{Scha74}.
\end{proof}

As we did previously for the momentum constraint,
we now outline an approximation result for general Galerkin approximations
to solutions of the Hamiltonian constraint.
Again, it is referred to as a quasi-optimal error estimate, in that it 
establishes a bound on the error $\|u-u_h\|_X$ that is within a constant 
of being the error in the best possible approximation.

We begin again with the two Hilbert spaces
$H$ and $V$, where again
we have $H=L^2(\calg{M})$ and $V=H^1_{0,D}(\calg{M})$.
We are given the following nonlinear variational problem:
\begin{equation}
  \label{eqn:ham_gal_cont}
\text{Find}~u \in V ~\text{s.t.}~
      A(u,v) + \langle B(u),v \rangle = F(v), ~~\forall v \in V,
\end{equation}
where the bilinear form
$A(u,v) : V \times V \mapsto \bbbb{R}$ is bounded
\begin{equation}
  \label{eqn:ham_as2}
A(u,v) \le M \| u \|_V \| v \|_V,
      ~~\forall u,v \in V,
\end{equation}
and V-elliptic:
\begin{equation}
  \label{eqn:ham_as3}
m \|u\|_V^2 \le A(u,u),
      ~~\forall u \in V, ~~~\text{where}~~ m > 0,
\end{equation}
where the linear functional $F(v) : V \mapsto \bbbb{R}$ is bounded
and thus lies in the dual space $V^*$:
$$
F(v) \le L \| v \|_V,
      ~~\forall v \in V,
$$
and where the nonlinear form
$\langle B(u),v \rangle : V \times V \mapsto \bbbb{R}$
is assumed to be monotonic:
\begin{equation}
  \label{eqn:ham_as3a}
0 \le \langle B(u)-B(v),u-v \rangle,
      ~~\forall u,v \in V,
\end{equation}
where we have used the notation:
\begin{equation}
\label{eqn:weakFormDiff}
  \langle B(u)-B(v),w \rangle = \langle B(u),w \rangle
                              - \langle B(v),w \rangle.
\end{equation}
We are interested in the quality of a Galerkin approximation:
\begin{equation}
  \label{eqn:ham_gal_disc}
\text{Find}~u_h \in V_h ~\text{s.t.}~
      A(u_h,v) + \langle B(u_h),v \rangle = F(v),
      ~~\forall v \in V_h,
\end{equation}
where $V_h \subset V$.
We will assume that $\langle B(u),v \rangle$
is bounded in the following weak sense:
If $u \in V$ satisfies~(\ref{eqn:ham_gal_cont}),
if $u_h \in V_h$ satisfies~(\ref{eqn:ham_gal_disc}),
and if $v \in V_h$, then there exists a constant $K>0$ such that:
\begin{equation}
  \label{eqn:ham_as3b}
\langle B(u)-B(u_h),u-v \rangle \le K \| u-u_h \|_V \| u-v \|_V.
\end{equation}
It can be shown that the weak formulation of the
Hamiltonian constraint fits into this framework.
We have again assumed that any Dirichlet function $\bar{u}$ has been
absorbed into the various forms in the obvious way to simplify the discussion.
The discussion can be modified to include approximation
of $\bar{u}$ by $\bar{u}_h$.

Again, we are interested in the quality of a
Galerkin approximation $u_h$ satisfying~(\ref{eqn:ham_gal_disc}),
or equivalently:
$$
    A(u-u_h,v) + \langle B(u)-B(u_h),v \rangle = 0,
      ~\forall v \in V_h \subset V.
$$
As before,
we will assume that there exists a sequence of approximation subspaces
$V_h \subset V$ parameterized by $h$, with $V_{h_1} \subset V_{h_2}$
when $h_2 < h_1$, and that there exists a sequence $\{a_h\}$, with
$\lim_{h\rightarrow 0} a_h = 0$, such that
\begin{equation}
  \label{eqn:ham_as4}
\|u-u_h\|_H \le a_h \|u-u_h\|_V, 
\end{equation}
holds whenever $u_h$ satisfies~(\ref{eqn:ham_gal_disc}).
The assumption~(\ref{eqn:ham_as4}) is again very natural;
see the discussion above following~(\ref{eqn:mom_as4}).
Under these assumptions, we have the following {\em a priori} error estimate.
\begin{theorem}
   \label{thm:ham_approx}
The approximation $u_h$ satisfying~(\ref{eqn:ham_gal_disc})
obeys the following quasi-optimal {\em a priori} error bounds:
\begin{eqnarray}
\|u-u_h\|_V 
   & \le & C \inf_{v \in V_h} \|u-v\|_V,
\label{eqn:holst} \\
\|u-u_h\|_H 
   & \le & C a_h \inf_{v \in V_h} \|u-v\|_V,
\label{eqn:holst_L2}
\end{eqnarray}
where $C$ is a constant independent of $h$.
\end{theorem}
\begin{proof}
See~\cite{Hols2001a,HoTs07a,HoTs07b}.
\end{proof}

\subsection{Adaptive Finite Element Methods (AFEM)}
  \label{sec:fem_adaptivity}

{\em A priori} error analysis for the finite element method for addressing
the first issue is now a well-understood subject~\cite{Ciar78,BrSc94}.
Much activity has recently been centered around
{\em a posteriori} error estimation and its use in adaptive mesh
refinement algorithms~\cite{PLTMG,BaRh78a,BaRh78b,Verf94,Verf96,XuZh97}.
These estimators include weak and strong residual-based 
estimators~\cite{BaRh78a,BaRh78b,Verf94}, as well as estimators based on
the solution of local problems~\cite{BaSm93,BaWe85}.
The challenge for a numerical method is to be as efficient as possible,
and {\em a posteriori} estimates are a basic tool in deciding which
parts of the solution require additional attention.
While the majority of the work on {\em a posteriori} estimates has been 
for linear problems, nonlinear extensions are possible through linearization
theorems~(cf.~\cite{Verf94,Verf96} and the discussion of the error
estimator employed by \FETK{} later in this paper).
The solve-estimate-refine structure in simplex-based adaptive finite element
codes such as \FETK{}~\cite{Hols2001a} and \PLTMG{}~\cite{PLTMG},
exploiting these {\em a posteriori} estimators, is as follows:
\begin{algorithm}
   \label{alg:MC}
(Adaptive multilevel finite elements)
{\small \begin{itemize}
\item While $(\|u - u_h \|_X > \epsilon)$ do:
    \begin{enumerate}
    \item Find $(u_h - \bar{u}_h) \in U_h \subset \calg{B}_1$ 
          such that: \\ $\langle F(u_h),v \rangle=0,
          ~\forall~ v \in V_h \subset \calg{B}_2$.
    \item Estimate $\|u-u_h\|_X$ over each element.
    \item Initialize two temporary simplex lists as empty: $Q1=Q2=\emptyset$.
    \item Place simplices with large error on the ``refinement'' list $Q1$.
    \item Bisect all simplices in $Q1$ (removing from $Q1$),
          and place any nonconforming simplices created on the list $Q2$.
    \item $Q1$ is now empty; set $Q1$ = $Q2$, $Q2 = \emptyset$.
    \item If $Q1$ is not empty, goto (5).
    \end{enumerate}
 \item End While.
\end{itemize} }
\end{algorithm}
The conformity loop (5)--(7), required to produce a globally ``conforming''
mesh (described below) at the end of a refinement step, is guaranteed to 
terminate in a finite number of steps (cf.~\cite{Riva84,Riva91}),
so that the refinements remain local.
Element shape is crucial for approximation quality; the bisection 
procedure in step~(5) is guaranteed to produce nondegenerate families 
if the longest edge is bisected in two dimensions~\cite{RoSt75,Styn80}, 
and if marking or homogeneity methods are used 
in three dimensions~\cite{AMP97,Mukh96,Bans91a,Bans91b,LiJo95,Maub95}.
Whether longest edge bisection is nondegenerate in three dimensions apparently
remains an open question.
Figure~\ref{fig:refine} shows a single subdivision of a 2-simplex or a
3-simplex using either 4-section (left-most figure),
8-section (fourth figure from the left),
or bisection (third figure from the left, and the right-most figure).
\begin{figure}[tbh]
\begin{center}
\mbox{\myfigpng{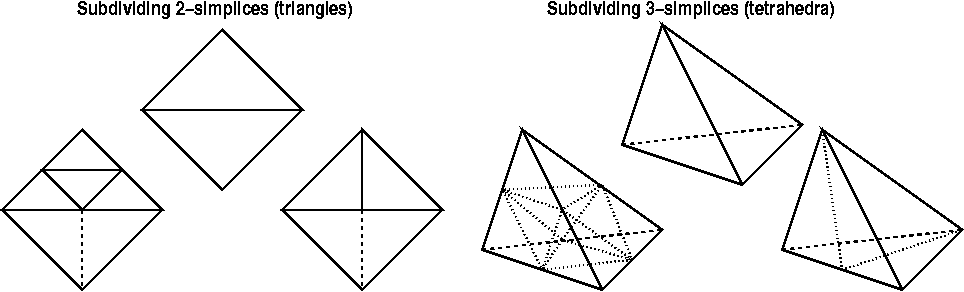}{1.0in}}
   \caption{Refinement of 2- and 3-simplices using 4-section,
            8-section, and bisection.}
   \label{fig:refine}
\end{center}
\end{figure}
The paired triangle in the 2-simplex case of Figure~\ref{fig:refine} 
illustrates the nature of conformity and its violation during refinement.
A globally conforming simplex mesh is defined as a collection of simplices
which meet only at vertices and faces; for example, removing the dotted
bisection in the third group from the left in Figure~\ref{fig:refine}
produces a non-conforming mesh.
Non-conforming simplex meshes create several theoretical as well as practical
implementation difficulties, and the algorithms in \FETK{} (as well as those
in \PLTMG{}~\cite{PLTMG} and similar simplex-based adaptive 
codes~\cite{Hols2001a,Mukh96,BER95,Bey96,BBJL98}) enforce conformity using
the above queue swapping strategy or a similar approach.

Addressing the complexity of Step~1 of the algorithm above,
Newton methods are often the most effective:
\begin{algorithm}
   \label{alg:newton}
(Damped-inexact-Newton)
{\small \begin{itemize}
 \item Let an initial approximation $u$ be given.
 \item While ($|\langle F(u),v \rangle| > \epsilon$ for any $v$) do:
    \begin{enumerate}
    \item Find $w$ such that: \\
          $\langle DF(u)w,v \rangle = - \langle F(u),v \rangle + r,
          ~\forall~ v$.
    \item Set $u = u + \lambda w$.
    \end{enumerate}
 \item End While.
\end{itemize} }
\end{algorithm}
The bilinear form $\langle DF(u)w,v \rangle$ in the algorithm above is simply
the (Gateaux) linearization of the nonlinear form $\langle F(u),v \rangle$,
defined formally as:
\begin{equation}
\langle DF(u)w,v \rangle=\left.\frac{d}{d\epsilon}
    \langle F(u+\epsilon w),v \rangle \right|_{\epsilon=0}.
\end{equation}
This form is easily computed from most nonlinear forms
$\langle F(u),v \rangle$ which
arise from second order nonlinear elliptic problems, although the
calculation can be tedious in some cases (as in the case
of the constraints in general relativity).
   The possibly nonzero ``residual'' term $r$ is to allow for inexactness
in the Jacobian solve for efficiency, which is quite effective in many
cases (cf.~\cite{BaRo82,DES82,EiWa92}).
The parameter $\lambda$ brings robustness to the 
algorithm~\cite{EiWa92,BaRo80,BaRo81}.
If folds or bifurcations are present, then the iteration is modified 
to incorporate path-following~\cite{Kell87,BaMi89}.

As was the case for the nonlinear residual
$\langle F(\cdot),\cdot \rangle$, the matrix
representing the bilinear form in the Newton iteration is easily assembled,
regardless of the complexity of the bilinear form 
$\langle DF(\cdot)\cdot,\cdot \rangle$.
In particular, the algebraic system for $w = \sum_{j=1}^n \beta_j \phi_j$
has the form:
\begin{equation} \label{algebraicSystem}
A U = F,
\ \ \ \ \ 
U_i = \beta_i,
\end{equation}
where
\begin{eqnarray}
A_{ij} &=& \langle DF( \bar{u}_h + \sum_{k=1}^n \alpha_k \phi_k )\phi_j,
         \psi_i \rangle, \\
F_i &=& \langle F( \bar{u}_h + \sum_{j=1}^n \alpha_j \phi_j),\psi_i \rangle.
\end{eqnarray}
As long as the integral forms
$\langle F(\cdot),\cdot \rangle$ and $\langle DF(\cdot)\cdot,\cdot \rangle$
can be evaluated at individual points in the domain, then quadrature can
be used to build the Newton equations, regardless of the complexity of
the forms.
This is one of the most powerful features of the finite element method,
and is exploited to an extreme in the code \FETK{}
   (see Section~\ref{sec:fetk} and~\cite{Hols2001a}).
It should be noted that there is a subtle difference between the approach
outlined here (typical for a nonlinear finite element approximation) and
that usually taken when applying a Newton-iteration to a nonlinear finite
difference approximation.
In particular, in the finite difference setting the discrete equations
are linearized explicitly by computing the Jacobian of the system of 
nonlinear algebraic equations.
In the finite element setting, the commutativity of linearization and
discretization is exploited; the Newton iteration is actually
performed in function space, with discretization occurring
``at the last moment'' in Algorithm~\ref{alg:newton} above.

   It can be shown that the Newton iteration above
is dominated by the computational 
complexity of solving the $n$ linear algebraic equations 
in each iteration (cf.~\cite{BaRo82,Hack85}).
   Multilevel methods are the only known provably optimal or nearly optimal
methods for solving these types of linear algebraic equations resulting
from discretizations of a large class of general linear elliptic 
problems~\cite{Hack85,BaDu81,Xu92a}.
   An obstacle to applying multilevel methods to the constraint equations
in general relativity and to similar equations is the presence of
geometrically complex domains.
   The need to accurately represent complicated domain features and boundaries
with an adapted mesh requires the use of very fine mesh simply to describe
the complexities of the domain.
   This may preclude the use of the solve-estimate-refine structure
outlined above in some cases, which requires starting with a coarse mesh in
order to build the approximation and linear algebra hierarchies as the problem
is solved adaptively.
   In this situation, algebraic or agglomeration/aggregation-based multilevel 
methods can be
employed~\cite{BaXu94,BaXu96,Bran86,CSZ94,CGZ97,BMR84,RuSt87,VMB94,VMB95}.
   A fully unstructured algebraic multilevel approach is taken in \FETK{},
even when the refinement hierarchy is present; see Section~\ref{sec:fetk}
below and also~\cite{Hols2001a} for a more detailed description.

\subsection{Residual-Based {\em A Posteriori} Error Indicators}
  \label{sec:fem_residual}

There are several approaches to adaptive error control, although the
approaches based on {\em a posteriori} error estimation are usually the
most effective and most general.
While most existing work on {\em a posteriori} estimates has been for linear 
problems, extensions to the nonlinear case can be made through linearization.
To describe one such result from~\cite{Hols2001a}, we assume that the 
$d$-manifold $\calg{M}$ has been exactly
triangulated with a set $\calg{S}$ of shape-regular $d$-simplices
(the finite dimension $d$ is arbitrary throughout this discussion).
A family of simplices will be referred to here as shape-regular if
for all simplices in the family the ratio of the diameter of the
circumscribing sphere to that of the inscribing sphere
is bounded by an absolute fixed constant, independent of the numbers and
sizes of the simplices that may be generated through refinements.
(For a more careful definition of shape-regularity and related concepts,
see~\cite{Ciar78}.)
It will be convenient to introduce the following notation:
\begin{center}
\begin{tabular}{lcl}
$\calg{S}$    & = & Set of shape-regular simplices
                    triangulating $\calg{M}$ \\
$\calg{N}(s)$ & = & Union of faces in simplex set
                    $s$ lying on $\partial_N \calg{M}$ \\
$\calg{I}(s)$ & = & Union of faces in simplex set
                    $s$ not in $\calg{N}(s)$ \\
$\calg{F}(s)$ & = & $\calg{N}(s) \cup \calg{I}(s)$ \\
$\omega_s$    & = & $~\bigcup~ \{~ \tilde{s} \in \calg{S} ~|~
                         s \bigcap \tilde{s} \ne \emptyset,
                     ~\text{where}~s \in \calg{S} ~\}$ \\
$\omega_f$    & = & $~\bigcup~ \{~ \tilde{s} \in \calg{S} ~|~
                     f \bigcap \tilde{s} \ne \emptyset,
                     ~\text{where}~f \in \calg{F} ~\}$ \\
$h_s$         & = & Diameter (inscribing sphere) of the simplex $s$ \\
$h_f$         & = & Diameter (inscribing sphere) of the face $f$.
\end{tabular}
\end{center}
When the argument to one of the face set
functions $\calg{N}$, $\calg{I}$, or $\calg{F}$ is in fact the entire
set of simplices $\calg{S}$, we will leave off the explicit dependence on
$\calg{S}$ without danger of confusion.
Referring forward briefly to Figure~\ref{fig:river} will be convenient.
The two darkened triangles in the left picture in Figure~\ref{fig:river}
represents the set $w_f$ for the face $f$ shared
by the two triangles.
The clear triangles in the right picture in Figure~\ref{fig:river}
represents the set $w_s$ for the darkened triangle $s$ in the center
(the set $w_s$ also includes the darkened triangle).

Finally, we will also need some notation to represent
discontinuous jumps in function values across faces interior to the
triangulation.
To begin, for any face $f \in \calg{N}$, let $n_f$ denote the
unit outward normal;
for any face $f \in \calg{I}$, take $n_f$ to be an arbitrary
(but fixed) choice of one of the two possible face normal orientations.
Now, for any $v \in L^2(\calg{M})$ such that
$v \in C^0(s) ~\forall s \in \calg{S}$, define the {\em jump function}:
$$
[v]_f(x) = \lim_{\epsilon \rightarrow 0^+} v(x+\epsilon n_f)
         - \lim_{\epsilon \rightarrow 0^-} v(x-\epsilon n_f).
$$

By analyzing the element-wise volume and surface integrals
in~(\ref{eqn:weakForm}) and using some technical results on interpolation
of functions, the following fairly standard result is derived 
in~\cite{Hols2001a}:
\begin{theorem}
Let $u \in W^{1,r}(\calg{M})$ be a regular solution
of~(\ref{eqn:str1})--(\ref{eqn:str3}), or equivalently
of~(\ref{eqn:weak})--(\ref{eqn:weakForm}),
where some additional minimal assumptions hold as described 
in~\cite{Hols2001a}.
Then the following {\em a posteriori} error estimate holds for
a Petrov-Galerkin approximation $u_h$ satisfying~(\ref{eqn:galerkin}):
\begin{equation}
   \label{eqn:indicator_residual}
 \| u - u_h \|_{W^{1,r}(\calg{M})} \le
   C \left( \sum_{s \in \calg{S}} \eta_s^p \right)^{1/p},
\end{equation}
where
$$
C = 2 \cdot \max_{\calg{S},\calg{F}} \{ C_s, C_f \} 
      \cdot \max_{\calg{S},\calg{F}} \{ D_s^{1/q}, D_f^{1/q} \}
\cdot \| DF(u)^{-1} \|_{\calg{L}(W^{-1,q},W^{1,p})},
$$
and where the element-wise error indicator $\eta_s$ is defined as:
\begin{eqnarray}
\eta_s 
  &=& \left(
      h_s^p \| B^i - A^{ia}_{~~;a} \|_{L^p(s)}^p
  + \frac{1}{2} \sum_{f \in \calg{I}(s)} h_f
     \| \left[ A^{ia} n_a \right]_f \|_{L^p(f)}^p
   \right.
\label{eqn:estimator} \\
& & \left.
   + \sum_{f \in \calg{N}(s)} h_f
     \| C^i + A^{ia} n_a \|_{L^p(f)}^p
   \right)^{1/p}.
\nonumber
\end{eqnarray}
\end{theorem}
\begin{proof}
See~\cite{Hols2001a,HoTs07a,HoTs07b}.
\end{proof}

\subsection{An {\em A Posteriori} Error Indicator for the Constraints}
  \label{sec:fem_residual_gr}

Here, we using the general indicator above, we can instantiate an
estimator specifically for the constraints in general relativity.
The Ph.D. thesis of Mukherjee~\cite{Mukh96} contains
a residual-based error estimator for the Hamiltonian constraint
that is equivalent to our estimator when the momentum constraint is
not involved, in the specific case of $p=q=r=2$.
We consider first the Hamiltonian constraint, which can be thought of
as an equation of the form~(\ref{eqn:str1})--(\ref{eqn:str3}).
The error indicator from above now takes the form:
$$
\eta_s^{\text{H}} = \left( h_s^p \|
      \frac{1}{8} \hat{R} \phi_h
    + \frac{1}{12} ({\rm tr} K)^2 \phi_h^5
    \right.
    - \frac{1}{8} (\Ahatstar_{ab} + (\hat{L}W)_{ab})^2 \phi_h^{-7}
    - 2 \pi \hat{\rho} \phi_h^{-3} 
    -  \hat{D}_a \hat{D}^a \phi_h
    \|_{L^p(s)}^p
$$
\begin{equation}
   \label{eqn:estimatorHam}
\hspace*{-2.2cm}
  + \frac{1}{2} \sum_{f \in \calg{I}(s)} h_f
     \| \left[ \hat{n}_a \hat{D}^a \phi_h \right]_f \|_{L^p(f)}^p
  \left.
  + \sum_{f \in \calg{N}(s)} h_f
     \| c \phi_h - z + \hat{n}_a \hat{D}^a \phi_h \|_{L^p(f)}^p
   \right)^{1/p}.
\end{equation}
The momentum constraint also has the
form~(\ref{eqn:str1})--(\ref{eqn:str3}),
and the error indicator in~\cite{Hols2001a} takes the form:
$$
\eta_s^{\text{M}} = \left( h_s^p \|
      \frac{2}{3} \phi^6 \hat{D}^a {\rm tr}K + 8 \pi \hat{j}^a
       - \hat{D}_b(\hat{L}W_h)^{ab}
    \|_{L^p(s)}^p
    \right.
  + \frac{1}{2} \sum_{f \in \calg{I}(s)} h_f
     \| \left[ \hat{n}_b (\hat{L}W_h)^{ab} \right]_f \|_{L^p(f)}^p
$$
\begin{equation}
   \label{eqn:estimatorMom}
   \left.
  + \sum_{f \in \calg{N}(s)} h_f
     \| C^a_{~b} W_h^b - Z^a + \hat{n}_b (\hat{L}W_h)^{ab} \|_{L^p(f)}^p
   \right)^{1/p}.
\end{equation}
Finally, the error indicator we employ for the coupled system is the
$l^p$-weighted average of the two estimators above:
\begin{equation}
   \label{eqn:estimatorHamMom}
   \eta_s^{HM} = \left( w_H (\eta_s^{H})^p
                      + w_M (\eta_s^{M})^p \right)^{1/p},
\end{equation}
where the weights $w_H$ and $w_M$ satisfying
$w_H + w_M = 1$ are determined heuristically.

Note that in this special case the weighted sum estimator
in~(\ref{eqn:estimatorHamMom}) can be derived from the general system
estimator in~\cite{Hols2001a} by defining the product space
metric described in~\cite{Hols2001a} as follows:
\begin{equation}
\calg{G}_{ij} = \left[ \begin{array}{ccc}
         w_H & 0      \\
         0 & w_M g_{ab} \\
         \end{array} \right],
u^i = \left[ \begin{array}{c}
         \phi \\
         W^a  \\
         \end{array} \right],
v^j = \left[ \begin{array}{c}
         \psi \\
         V^b  \\
         \end{array} \right],
\end{equation}
and employing the coupled system framework from~\cite{Hols2001a,HoTs07a}.
However, in general different components of a coupled system
could lie in different function spaces, and the norm appearing
in the estimator in~\cite{Hols2001a,HoTs07a} would need to be
modified.
Equivalently, an estimator built from a weighted sum of estimators
for individual components could be used as above.

\subsection{Convergence and Optimal Complexity of AFEM for the Constraints}
  \label{sec:fem_convergence}

The following convergence result for AFEM applied the a general class
of nonlinear elliptic equations that includes the Hamiltonian constraint
appears in~\cite{HoTs07b}.
More general results which apply to the coupled system 
appear in~\cite{HoTs07a}.
Below, the energy norm is used as the natural norm
for the analysis:
$$
\tbar u \tbar = A(u,u)^2,
$$
where the bilinear form $A(u,v)$ 
is defined in Section~\ref{sec:weakForm_gr}.
\begin{theorem}[Contraction]
   \label{T:convergence}
Let $\{P_k, S_k, U_k \}_{k\ge 0}$ be the
sequence of finite element meshes, spaces, and solutions,
respectively, produced by AFEM.
Let the initial mesh size $h_0$ be sufficiently small so that 
a quasi-orthogonality inequality holds (see~\cite{HoTs07b})
holds for $\{P_0, S_0, U_0 \}_{k\ge 0}$.
Then, there exist constants $\gamma > 0$ and $\alpha \in (0,1)$,
depending only on some freely specifiable AFEM parameters and the 
shape-regularity of the initial triangulation $P_0$, such that
$$
\tbar u - U_{k+1} \tbar^2 + \gamma \eta_{k+1}^2
   \le
\alpha \left( \tbar u - U_{k} \tbar^2 + \gamma \eta_{k}^2 \right).
$$
\end{theorem}
\begin{proof}
See~\cite{HoTs07b}.
\end{proof}
The strict contraction result above makes possible the
following complexity result from~\cite{HoTs07b}, which
guarantees optimal complexity of the AFEM iteration for the
Hamiltonian constraint equation under 
reasonable assumptions on the problem data.
The approximation class $\mathbb{A}_s$ is precisely characterized
in~\cite{HoTs07a,HoTs07b}.
\begin{theorem}[Optimality]
   \label{T:optimality}
If data lies in approximation class $\mathbb{A}_s$, then
there exists a constant $C$ such that
$$
\tbar u-U_{k+1}\tbar^2 + \mathrm{osc}_k
\le C \left( \#P_k - \#P_0 \right)^{-s}.
$$
\end{theorem}
\begin{proof}
See~\cite{HoTs07b}.
\end{proof}

\section{Fast Solvers and Preconditioners for AFEM}

We first introduce a superscript $(j)$ to indicate the level of the
stiffness matrix (or sometimes, the discretization operator) in the
algebraic system (\ref{algebraicSystem}). At certain places, we will
drop the superscript for simplicity.  The FE solution is sampled on
the nodes, hence the number of internal nodes is equal to the number
of degrees of freedom (DOF) or unknowns in the system
(\ref{algebraicSystem}).  The total number of DOF on the finest level
$J$ is denoted by $N_J=N$.

The stiffness matrix $A^{(j)}$ is ill-conditioned, with the condition
number $\kappa(A)$ of the system (\ref{algebraicSystem}) growing
${\cal O}(2^{2j})$ as $j \rightarrow \infty$ (in the case of second order
elliptic PDE). It is imperative to improve the
condition of (\ref{algebraicSystem}) by transforming the system to an
equivalent one, namely by preconditioning
$$
(C^{(j)} ~ A^{(j)}) ~ U^{(j)} = C^{(j)} ~ F^{(j)},
$$ 
where 
$\kappa({C^{(j)}}^{1/2} ~ A^{(j)} ~ {C^{(j)}}^{1/2}) \ll\kappa(A^{(j)})$.  
Moreover, the preconditioning matrix $C^{(j)}$ has to be positive
definite, and in some sense simple. One possible way to implement the
above strategy is to determine a positive definite matrix $B^{(j)}$,
having the following two properties:
\begin{itemize}
\item ${B^{(j)}}^{-1}$ can efficiently be computed (usually, ${\cal O}(N_j)$
is the desirable bound for the number of arithmetical operations when
solving a linear system with coefficient matrix $B^{(j)}$),

\item $A^{(j)}$ and $B^{(j)}$ are ``almost'' spectrally equivalent, i.e. 
$$
\lambda_B x^T B^{(j)}x \leq x^T A^{(j)}x \leq \Lambda_B x^T B^{(j)}x,
~~~x \in \Re^{N_j},
$$
with two positive constants $\lambda_B, \Lambda_B$ with a small ratio 
$\frac{\Lambda_B}{\lambda_B}$. 
\end{itemize}
Since $\kappa( {B^{(j)}}^{-1/2} ~ A^{(j)} ~ {B^{(j)}}^{-1/2}) \leq 
\frac{\Lambda_B}{\lambda_B}$, then $C^{(j)}={B^{(j)}}^{-1}$ will be a good 
preconditioner choice.

Solution of the algebraic system~(\ref{algebraicSystem}) by iterative
methods has been the subject of intensive research because of the
enormous practical impact on a number of application areas in
computational science.  For quality approximation in physical
simulation, one is required to use meshes containing very large
numbers of simplices leading to approximation spaces ${\cal S}_j$ with
very large dimension $N_j$.  Only iterative methods which scale well
with $N_j$ can be used effectively, which usually leads to the use of
multilevel-type iterative methods and preconditioners.  Even with the
use of such optimal methods for~(\ref{algebraicSystem}), which means
methods which scale linearly with $N_j$ in both memory and
computational complexity, the approximation quality requirements on
${\cal S}_j$ often force $N_j$ to be so large that only parallel
computing techniques can be used to solve~(\ref{algebraicSystem}).

To overcome this difficulty one employs adaptive methods, which
involves the use of {\em a posteriori} error estimation to drive
{\em local mesh refinement} algorithms.  This approach leads to
approximation spaces ${\cal S}_j$ which are adapted to the particular
target function $u$ of interest, and as a result can achieve a desired
approximation quality with much smaller approximation space dimension
$N_j$ than non-adaptive methods.  One still must solve the algebraic
system~(\ref{algebraicSystem}), but unfortunately most of the
available multilevel methods and preconditioners are no longer
optimal, in either memory or computational complexity.  This is due to
the fact that in the local refinement setting, the approximation
spaces $S_j$ do not increase in dimension geometrically as they do in
the uniform refinement setting.  As a result for example, a single
multilevel V-cycle no longer has linear complexity, and the same
difficulty is encountered by other multilevel methods.  Moreover,
storage of the discretization matrices and vectors for each
approximation space, required for assembling V-cycle and similar
iterations, no longer has linear memory complexity.

A partial solution to the problem with multilevel methods in the local
refinement setting is provided by the hierarchical basis (HB) method~
\cite{BaDu80,BDY88,Yser86b}.  This method is based on a direct or
hierarchical decomposition of the approximation spaces ${\cal S}_j$
rather than the overlapping decomposition employed by the multigrid
and Bramble-Pasciak-Xu (BPX)~\cite{BPX90} method, and therefore by
construction has linear memory complexity as well as linear
computational complexity for a single V-cycle-like iteration.
Unfortunately, the HB condition number is not uniformly bounded,
leading to worse than linear overall computational complexity.  While
the condition number growth is slow (logarithmic) in two dimensions,
it is quite rapid (geometric) in three dimensions, making it
ineffective in the 3D local refinement setting. Recent alternatives to
the HB method, including both BPX-like methods~\cite{BrPa93,BPX90} and
wavelet-like stabilizations of the HB methods~\cite{VaWa97a}, provide
a final solution to the condition number growth problem.  It was shown
in~\cite{DaKu92} that the BPX preconditioner has uniformly bounded
condition number for certain classes of locally refined meshes in two
dimensions, and more recently in~\cite{AkHo02,AkHo05a,ABH02a,ABH04} it
was shown that the condition number remains uniformly bounded for
certain classes of locally refined meshes in three spatial dimensions.
In~\cite{AkHo02,ABH02a,AkHo05b}, it was also shown that
wavelet-stabilizations of the HB method give rise to uniformly bounded
conditions numbers for certain classes of local mesh refinement in
both the two- and three-dimensional settings.

\subsection{Preliminaries on Optimal Preconditioners} \label{sec:pre}
In the uniform refinement setting, the parallelized or additive version of
the multigrid method, also known as the BPX preconditioner, is defined
as follows:
\begin{equation} \label{eq:bpxImp}
X u := \sum_{j=0}^J 
2^{j(d-2)} \sum_{i=1}^{N_j} (u,\phi_i^{(j)}) \phi_i^{(j)},~~~u \in {\cal S}_J.
\end{equation}
Only in the presence of a geometric increase in the number of DOF,
the same assumption for optimality of a single classical
(i.e.~smoother acting on all DOF) multigrid or BPX iteration, does the cost 
per iteration remain optimal. In the case of local refinement, the BPX 
preconditioner (\ref{eq:bpxImp}) (usually known as additive multigrid) 
can easily be suboptimal because of the suboptimal cost per iteration.  
If the smoother is restricted to the space generated by \emph{fine or newly 
created} basis functions,
i.e.~$\tilde{{\cal S}}_j := (I_j - I_{j-1})~{\cal S}_{j}$,
then~(\ref{eq:bpxImp}) corresponds to the additive HB
preconditioner in~\cite{Yser86b}:
\begin{equation} \label{eq:bpxHb}
X_{\mbox{{\tiny {\rm HB}}}} u =
\sum_{j=0}^J 2^{j(d-2)} \sum_{i=N_{j-1}+1}^{N_j} (u,\phi_i^{(j)}) \phi_i^{(j)},
~~~ u \in {\cal S}_J.
\end{equation}
However, the HB preconditioner suffers from a suboptimal iteration count.  
Namely, if the algebraic 
system~(\ref{algebraicSystem}) is preconditioned by the HB preconditioner, the 
arising condition number is ${\cal O}(J^2)$ and ${\cal O}(2^J)$ in 2D and 3D, respectively.

In the local refinement setting, in order to maintain optimal 
overall computational complexity, the remedy is to restrict the smoother 
to a local space $\tilde{{\cal S}}_j$ which is typically slightly larger 
than the one generated by basis functions corresponding to fine DOF:
\begin{equation} \label{subset:localSpace}
(I_j - I_{j-1})~{\cal S}_{j} \subseteq \tilde{{\cal S}}_j \subset {\cal S}_j,
\end{equation}
where $I_j:L_2(\Omega) \rightarrow {\cal S}_j$ denotes the finite
element interpolation operator. The subspace generated by the nodal
basis functions corresponding to \emph{fine or newly created} degrees
of freedom (DOF) on level $j$ is denoted by $(I_j - I_{j-1})~{\cal
S}_{j}$.  Nodes--equivalently, DOF--corresponding to $\tilde{{\cal
S}}_j$ and their cardinality will be denoted by $\tilde{{\cal N}}_j$
and $\tilde{N}_j$, respectively. The above deficiencies of the BPX and
HB preconditioners can be overcome by restricting the (smoothing)
operations to $\tilde{{\cal S}}_j$.  This leads us to define the BPX
preconditioner for the local refinement setting as:
\begin{equation} \label{id:BPX}
X u := \sum_{j=0}^J 2^{j(d-2)} \sum_{i \in \tilde{\calg{N}}_j}
(u,\phi_i^{(j)}) \phi_i^{(j)},~~~ u \in {\cal S}_J.
\end{equation}

\begin{remark}
In order to prove optimal results on convergence, the basic theoretical 
restriction on the refinement procedure is that the refinement regions
from each level forms a \emph{nested} sequence. 
Let $\Omega_j$ denote the refinement region, namely, the union of 
the supports of basis functions which are introduced at level $j$.
Due to nested refinement $\Omega_j \subset \Omega_{j-1}$. Then the 
following nested hierarchy holds:
$$
\Omega_J \subset \Omega_{J-1} \subset \cdots \subset \Omega_0 = \Omega.
$$
Simply, the restriction indicates that tetrahedra of level $j$ which
are not candidates for further refinement will never be touched in the
future.  In practice, depending on the situation, the above nestedness
restriction may or may not be enforced. We enforce the restriction in
Lemma~\ref{lemma:BEK}. In realistic situations, it is typically not
enforced.
\end{remark}

\subsection{Matrix Representations and Local Smoothing}

We describe how to construct the matrix representation of the
preconditioners under consideration. Let the prolongation operator from
level $j-1$ to $j$ be denoted by
$$
P_{j-1}^j \in \Re^{\tilde{N}_j \times \tilde{N}_{j-1}},
$$
and also denote the prolongation operator from level $j$ to $J$ as:
$$
P_j \equiv P_j^J = P_{J-1}^J  \ldots P_j^{j+1} \in 
\Re^{N_J \times \tilde{N}_j},
$$
where $P_J^J$ is defined to be the rectangular identity matrix 
$I \in \Re^{N_J \times \tilde{N}_{J-1}}$. 
Then the matrix representation of (\ref{eq:bpxImp}) becomes~\cite{XuQi94}:
$$
X = \sum_{j=0}^J 2^{j(d-2)} P_j P_j^t.
$$
One can also introduce a version with an explicit smoother $G_j$: 
$$
X = \sum_{j=0}^J P_j G_j P_j^t.
$$
Throughout this article, the smoother 
$G_j \in \Re^{\tilde{N}_j \times \tilde{N}_j}$ is a symmetric 
Gauss-Seidel  iteration. Namely, 
$G_j = (D_j + U_j)^{-1} D_j (D_j + L_j)^{-1}$ where
$A_j = D_j + L_j + U_j$ with $\tilde{N}_0 = N_0$.

The matrix representation of~(\ref{eq:bpxHb}) is formed from matrices $H_j$
which are simply the tails of the $P_j$ corresponding to newly introduced DOF 
in the fine space. In other words, 
$H_j \in \Re^{N_J \times (N_j - N_{j-1})}$ is given by only keeping the fine 
columns (the last $N_j - N_{j-1}$ columns of $P_j$). Hence,
the matrix representation of~(\ref{eq:bpxHb}) becomes:
$$
X_{\mbox{{\tiny {\rm HB}}}} = \sum_{j=0}^J 2^{j(d-2)} H_j H_j^t.
$$

If the sum over $i$ in (\ref{eq:bpxImp}) is restricted only to those 
nodal basis functions with supports that intersect the refinement region
\cite{BoYs93,BrPa93,DaKu92,Oswa94}, then we obtain the set called as 
{\em onering} of fine DOF. Namely, the set which contains fine DOF and their
immediate neighboring coarse DOF: 
$${\rm ONERING}^{(j)}=\{onering (ii):~ii=N_{j-1}+1, \ldots, {N_j} \},$$
where $onering (ii) = \{ii, fathers(ii)\}.$
Now, the generic preconditioner~(\ref{id:BPX}) for local refinement 
transforms into the following preconditioner:
\begin{equation} \label{eq:bpxLocal}
X u =
\sum_{j=0}^J 2^{j(d-2)}
\sum_{ i \in \mbox{{\tiny {\rm ONERING}}}^{(j)} }
(u,\phi_i^{(j)}) \phi_i^{(j)},
~~~ u \in {\cal S}_J.
\end{equation}

There are three popular choices for $\tilde{{\cal N}}_j$. 
We outline possible BPX choices by the following DOF corresponding to:
\begin{itemize}
\item {\bf (DOF-1)} The basis functions with supports that intersect 
the refinement region $\Omega_j$~\cite{BoYs93,BrPa93,DaKu92}. We call
this set $\mbox{{\tiny {\rm ONERING}}}^{(j)}$.
\item {\bf (DOF-2)} The basis functions with supports that are contained in 
$\Omega_j$~\cite{Oswa94}.
\item {\bf (DOF-3)} Created by red refinement and their corresponding coarse
DOF.
\end{itemize}
Here, red refinement refers to quadrasection or octasection in 2D and
3D, respectively.  Green refinement simply refers to bisection.

The interesting ones are DOF-1 and DOF-3 and we would like to
elaborate on these. In the numerical experiments reported
in~\cite{ABH02a}, DOF-1 was used. For the provably optimal
computational complexity result in Lemma~\ref{lemma:BEK} DOF-3 is
used.

\subsubsection{The Sets DOF-1, DOF-3 and Local Smoothing Computational Complexity}
\label{sec:onering}

The set DOF-1 can be directly determined by the sparsity pattern of
the fine-fine subblock $A_{22}^{(j)}$ of the stiffness matrix
in~(\ref{defn:2by2}). Then, the set of DOF over which the BPX
preconditioner ~(\ref{eq:bpxLocal}) smooths is simply the union of the
column locations of nonzero entries corresponding to fine DOF. Using
this observation, HB smoother can easily be modified to be a BPX
smoother.

DOF-3 is equivalent to the following set:
\begin{equation} \label{eq:BEK-DOF}
\hspace*{-1.5cm}
\tilde{{\cal N}}_j = \{i=N_{j-1}+1, \ldots, N_j \}
\bigcup \{ i: \phi_i^{(j)} \neq \phi_i^{(j-1)},~~i=1, \ldots, N_{j-1}\},
\end{equation}
and the corresponding space over which the smoother acts:
$$
\tilde{{\cal S}}_j = {\rm span} \left[
\bigcup \{ \phi_i^{(j)} \}_{i=N_{j-1}+1}^{N_j}
\bigcup \{ \phi_i^{(j)} \neq \phi_i^{(j-1)} \}_{i=1}^{N_{j-1}} 
\right].
$$

This set is used in the Bornemann-Erdmann-Kornhuber (BEK) refinement~\cite{BEK93} and we 
utilize this set for the estimates of 3D local refinement. Since green refinement 
simply bisects a simplex, the modified basis function is the same as the one before the
bisection due to linear interpolation. So the set of DOF in (\ref{eq:BEK-DOF})
corresponds to DOF created by red refinement and corresponding coarse DOF 
(father DOF). The following crucial result from~\cite{BEK93} establishes a 
bound for the number of nodes used for smoothing. This indicates that
the BPX preconditioner has provably optimal (linear) computational complexity
per iteration on the resulting 3D mesh produced by the BEK refinement 
procedure.

\begin{lemma} \label{lemma:BEK}
The total number of nodes used for smoothing satisfies the bound:
\begin{equation} \label{ineq:smoothBound}
\sum_{j=0}^J \tilde{N}_j \leq \frac{5}{3}N_J - \frac{2}{3}N_0.
\end{equation}
\end{lemma}
\begin{proof}
See~\cite[Lemma 1]{BEK93}.
\end{proof}

The above lemma constitutes the first computational complexity
optimality result in 3D for the BPX preconditioner as reported in
~\cite{AkHo02}. A similar result
for 2D red-green refinement was given by Oswald~\cite[page95]{Oswa94}.
In the general case of local smoothing operators which
involve smoothing over newly created basis functions plus some
additional set of local neighboring basis functions, one can extend
the arguments from~\cite{BEK93}
and~\cite{Oswa94} using shape regularity.

\subsection{Hierarchical Basis Methods and Their Stabilizations}
HB methods exploit a 2-level hierarchical decomposition of the DOF. They
are divided into the coarse (the ones inherited from previous levels) and 
the fine (the ones that are newly introduced) nodes. In fact, in the operator
setting, this decomposition is a direct consequence of the direct decomposition
of the finite element space as follows:
$${\cal S}_j = {\cal S}_{j-1} \oplus {\cal S}_j^f.$$
Hence, $A_j$ can be represented by a two-by-two block form:
\begin{equation}  \label{defn:2by2}
A_j =
\left[ \begin{array}{cc}
A_{j-1}    & A^{(j)}_{12} \\
A^{(j)}_{21} & A^{(j)}_{22}
\end{array} \right]
\begin{array}{ll}
\} & {\cal S}_{j-1} \\
\} & {\cal S}^f_j
\end{array},
\end{equation}
where $A_{j-1}$, $A^{(j)}_{12}$, $A^{(j)}_{21}$, and $A_{22}^{(j)}$
correspond to coarse-coarse, coarse-fine, fine-coarse, and fine-fine 
discretization operators respectively. The same 2-level decomposition 
carries directly to the matrix setting. 

As mentioned earlier, HB methods suffer from the condition number
growth.  This makes makes HB methods especially ineffective in the 3D
local refinement setting. As we mentioned earlier, wavelet-like
stabilizations of the HB methods~\cite{VaWa97a} provide a final
solution to the condition number growth problem. The motivation for
the stabilization hinges on the following idea.  The BPX decomposition
gives rise to basis functions which are not locally supported, but
they decay rapidly outside a local support region.  This allows for
locally supported {\em approximations} as illustrated in
Figures~\ref{figure:basisFuncs} and \ref{figure:basisFuncsLower}.

\begin{figure}[ht]
\begin{center}
\mbox{\myfig{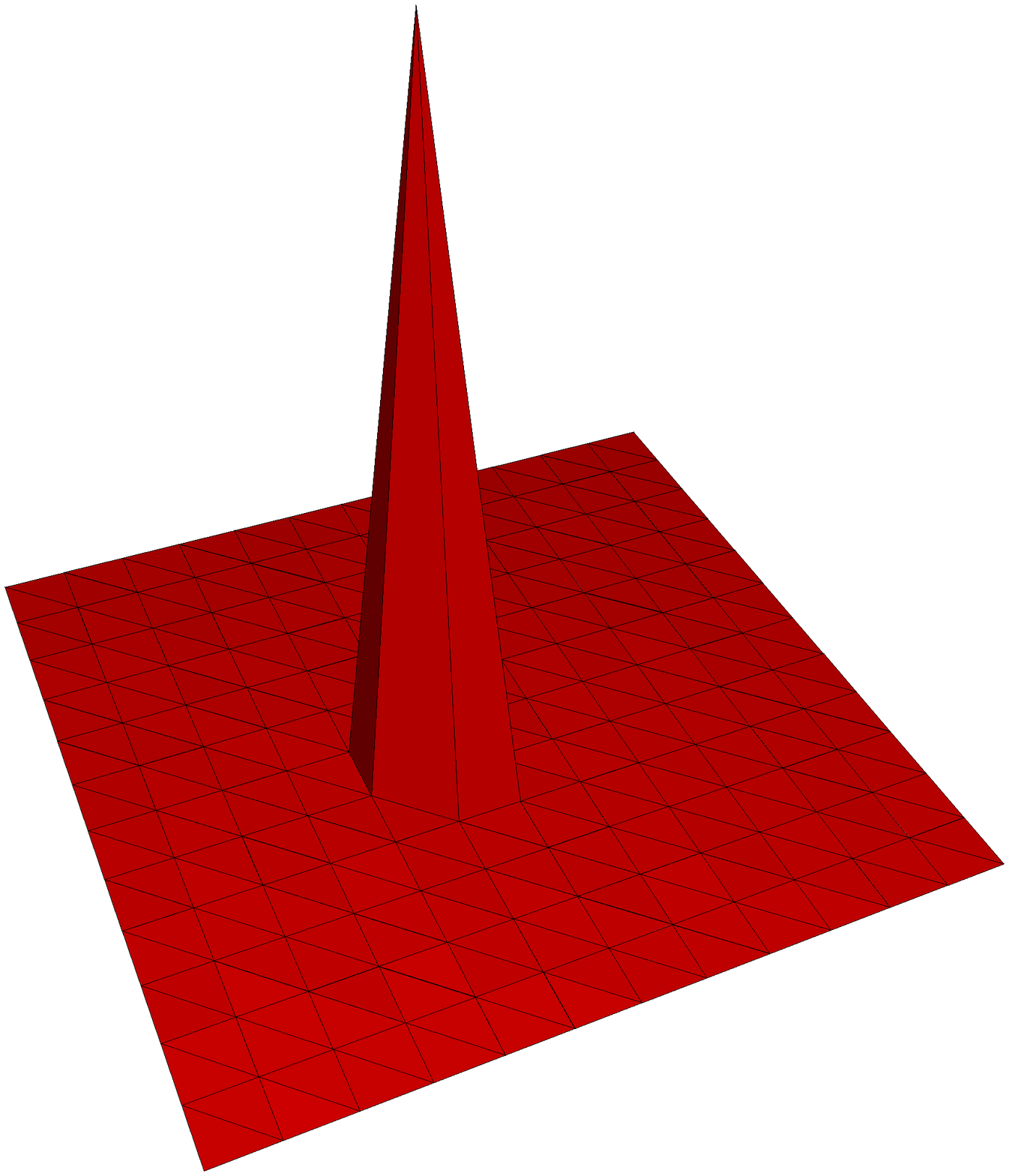}{5.25cm}}
\mbox{\myfig{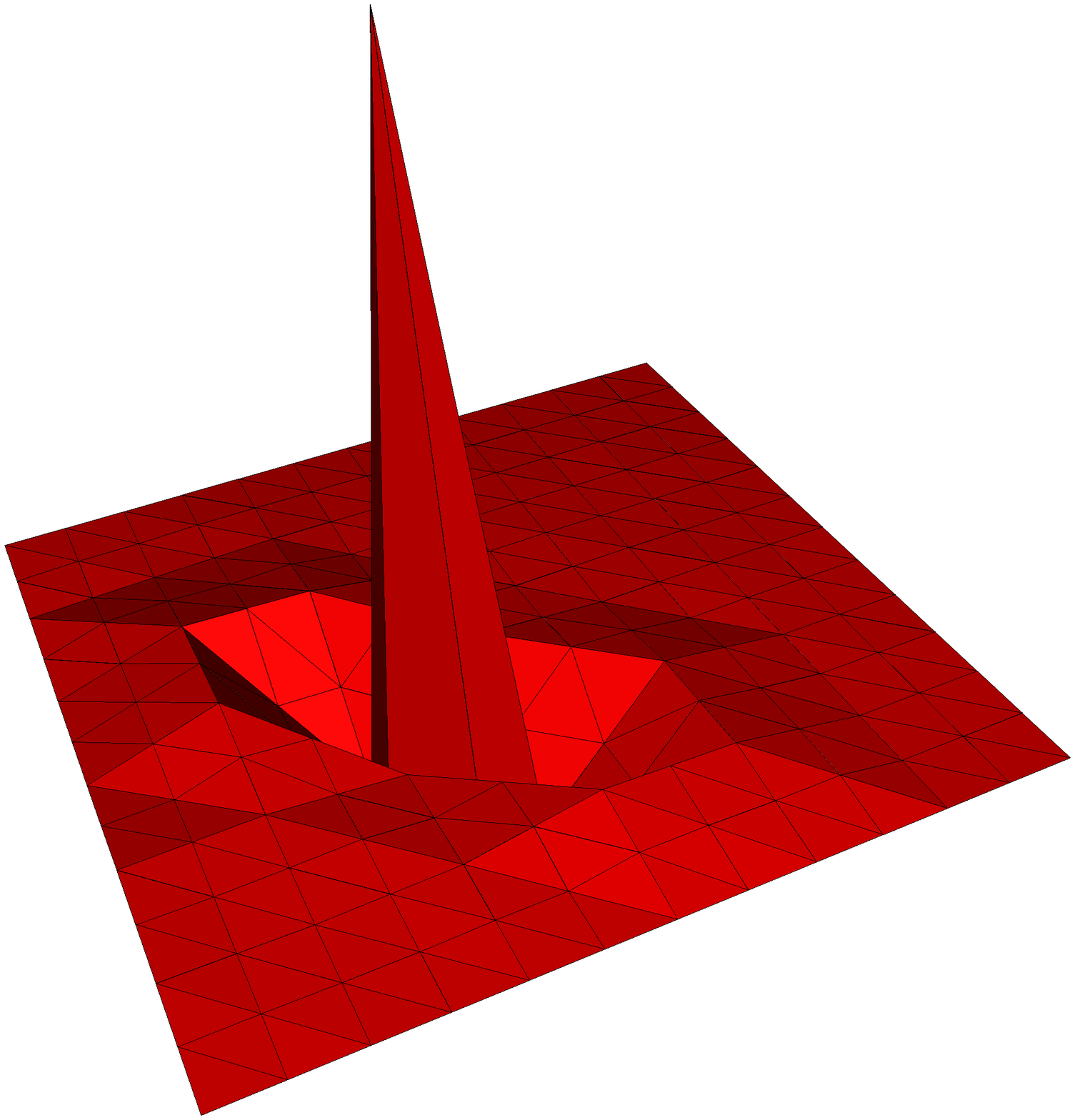}{5cm}}
\mbox{\myfig{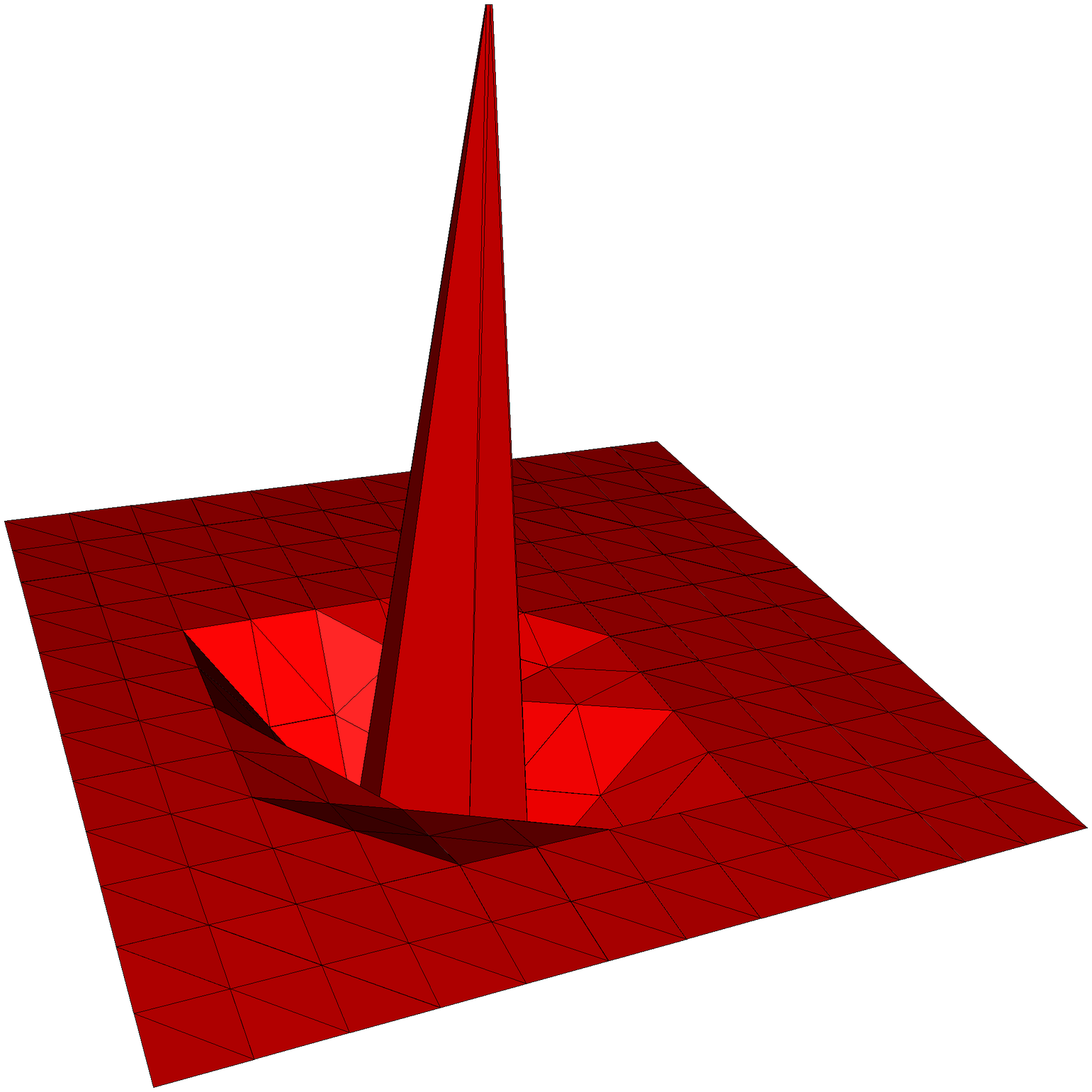}{5cm}}
\end{center}
\caption{Left: Hierarchical basis function without modification. 
Wavelet modified hierarchical basis functions. Middle: One iteration of
symmetric Gauss-Seidel approximation. Right: One iteration of Jacobi 
approximation.}
\label{figure:basisFuncs}
\end{figure}

\begin{figure}[ht]
\vspace*{-1cm}
\begin{center}
\mbox{\myfig{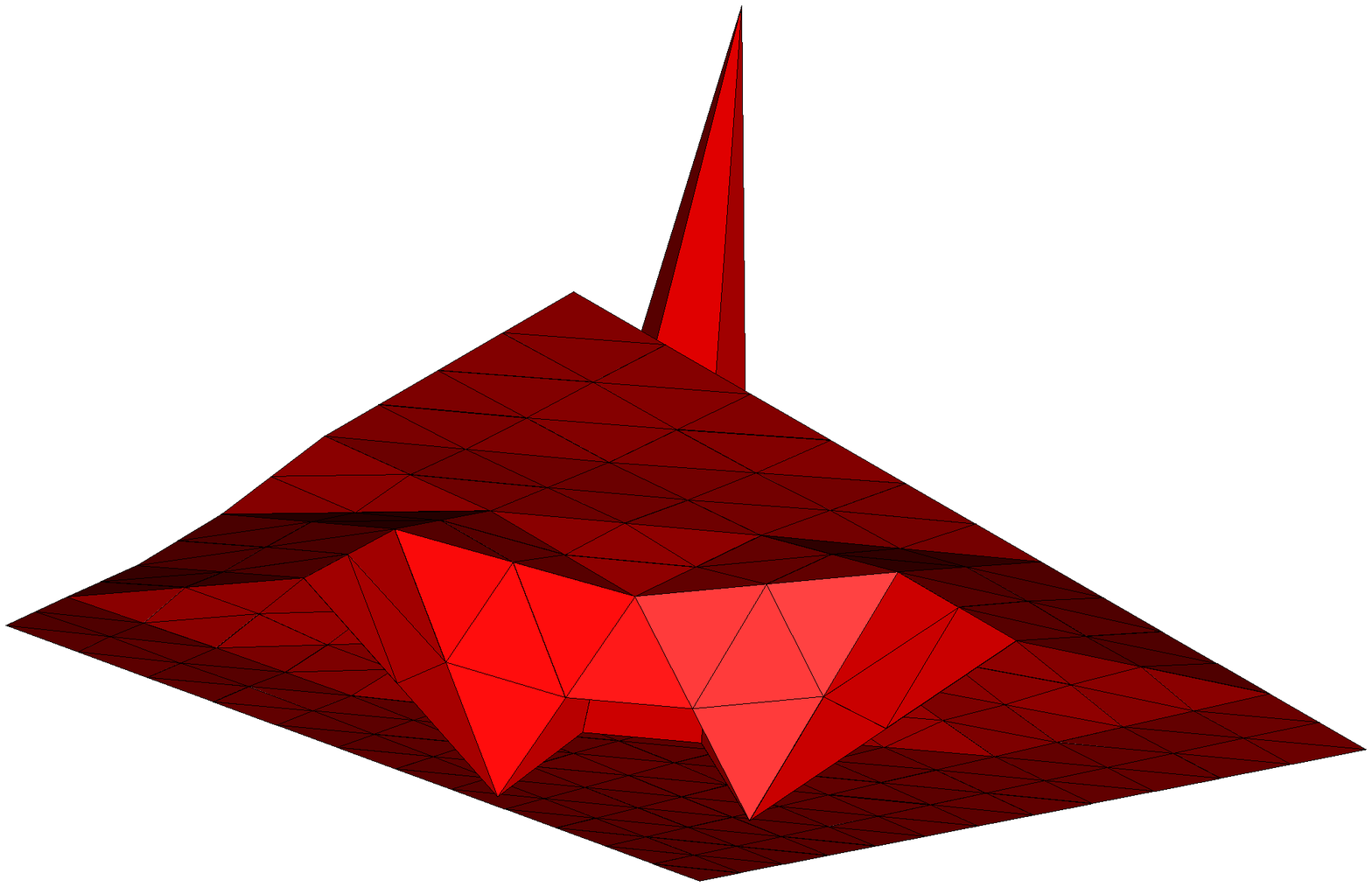}{5.5cm}}
\mbox{\myfig{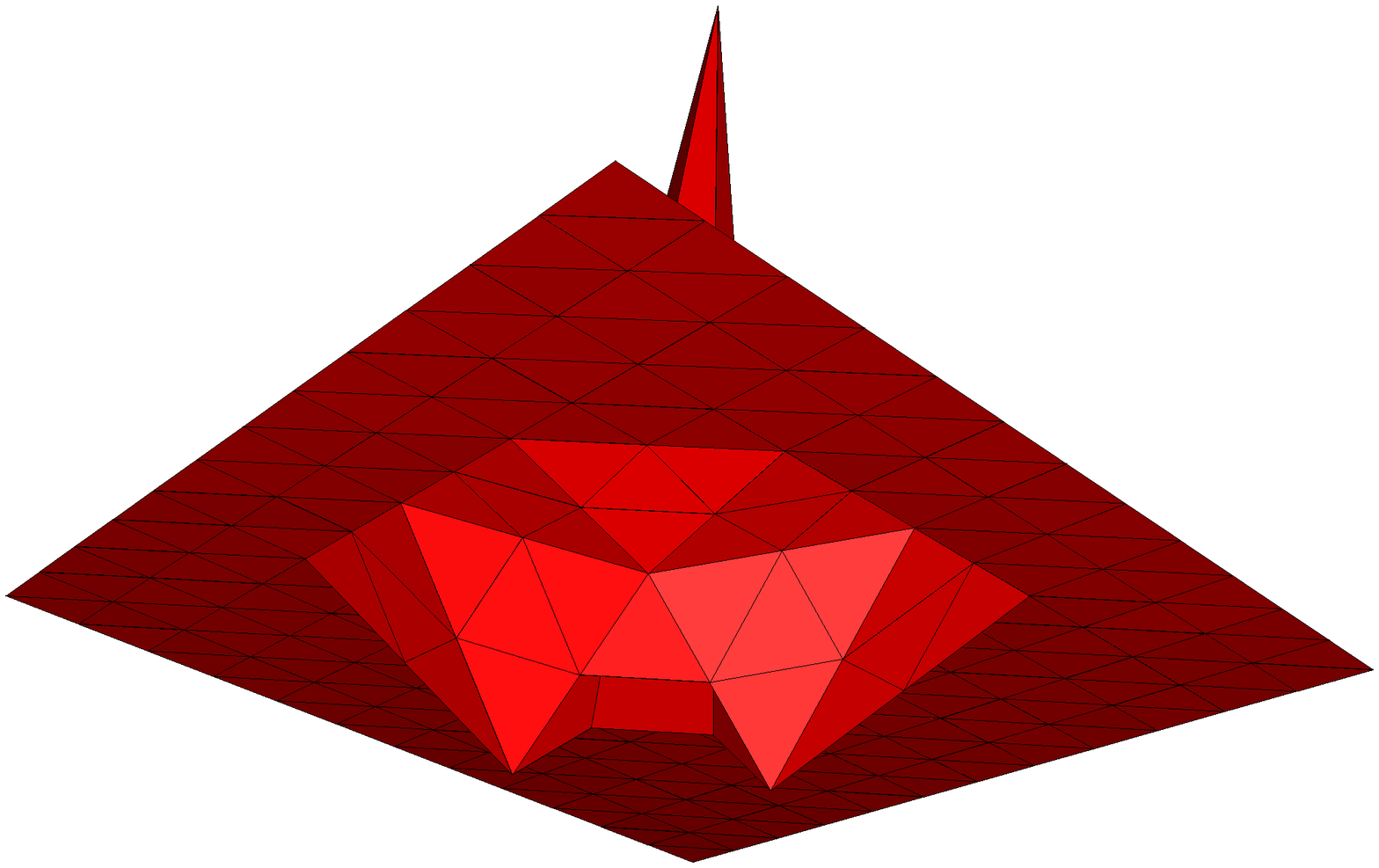}{5.5cm}}
\end{center}
\caption{Lower view of middle and left basis functions in 
Figure~\ref{figure:basisFuncs}.}
\label{figure:basisFuncsLower}
\end{figure}

The {\em wavelet modified hierarchical basis (WMHB)
methods}~\cite{VaWa97a,VaWa97b,VaWa98} can be viewed as an
approximation of the wavelet basis stemming from the BPX
decomposition~\cite{Jaff92}.  A similar wavelet-like multilevel
decomposition approach was taken in~\cite{Stev97}, where the orthogonal
decomposition is formed by a discrete $L_2$-equivalent inner
product. This approach utilizes the same BPX two-level
decomposition~\cite{Stev95,Stev97}.

For local refinement setting, the other primary method of interest is
the WMHB method.  The WMHB methods can be described as additive or
multiplicative Schwarz methods.  In one of the previous
papers~\cite{AkHo02}, it was shown that the
additive version of the WMHB method is optimal under certain types of
red-green mesh refinement.  Following the notational framework
in~\cite{AkHo02,ABH02a,AkHo05b,VaWa98}, this method is 
defined recursively as follows:
\begin{definition}
 \label{defn:Dj}
The additive WMHB method $D^{(j)}$ is defined for $j=1,\ldots,J$ as
$$ D^{(j)} \equiv \left[
\begin{array}{ll}
D^{(j-1)}    & 0 \\
0            & B^{(j)}_{22}
\end{array} \right],
$$
with $D^{(0)}=A^{(0)}$.
\end{definition}

With smooth PDE coefficients, optimal results were also established
for the multiplicative version of the WMHB method
in~\cite{AkHo02,AkHo05b}.  Our numerical experiments demonstrate such
optimal results. This method can be written recursively as:
\begin{definition} \label{defn:Bj}
The  multiplicative WMHB method $B^{(j)}$ is defined as
$$ B^{(j)} \equiv \left[
\begin{array}{ll}
B^{(j-1)}    & A^{(j)}_{12} \\
0            & B^{(j)}_{22}
\end{array} \right]
\left[ \begin{array}{ll}
I                            & 0 \\
B^{(j)^{-1}}_{22} A_{21}^{(j)} & I
\end{array} \right] =
\left[ \begin{array}{ll}
B^{(j-1)}+ A^{(j)}_{12}B^{(j)^{-1}}_{22} A_{21}^{(j)} & A^{(j)}_{12} \\
A_{21}^{(j)} & B^{(j)}_{22}
\end{array} \right],
$$
with $B^{(0)}=A^{(0)}$.
\end{definition}
$B^{(j)}_{22}$ denotes an approximation of
$A^{(j)}_{22}$, e.g. Gauss-Seidel or Jacobi approximation.
For a more complete description of these and related algorithms,
see~\cite{AkHo02,AkHo05b}.


\section{Practical Implementation of Fast Solvers}
The overall utility of any finite element code depends strongly on
efficient implementation of its core algorithms and data structures.
Finite element method becomes a viable tool in addressing realistic
simulations only when these critical pieces come together. Theoretical
results involving complexity are of little practical importance if the
methods cannot be implemented.  For algorithms involving data
structures, this usually means striking a balance between storage
costs and computational complexity.  For instance, finding a minimal
representation for a data set is only useful if the information can be
accessed efficiently. We elaborate on the data structure and the
implementation of the methods under consideration.

\subsection{Implementation  of Hierarchical Basis Methods}
In HB methods, nodal basis functions are transformed into hierarchical
basis functions via a nonsingular change of basis matrix: 
$$ 
Y = \left[ \begin{array}{cc} 
I & Y_{12} \\ Y_{21} & I + Y_{22}
\end{array} \right], 
$$
where $Y \in \Re^{N_j \times N_j}$, $Y_{12} \in \Re^{N_{j-1} \times n_j}$,
$Y_{21} \in \Re^{n_j \times N_{j-1}}$, and $Y_{22} \in \Re^{n_j \times n_j}$.
We denote the number of fine DOF at level $j$ as $n_j$.  Then, 
$N_j = n_1 + n_2 + \cdots + n_j$ is the total number of DOF at level $j$.
In this representation, we have assumed that the nodes are ordered with the
nodes $N_{j-1}$ inherited from the previous level listed first, and the $n_j$
new DOF listed second.  For both wavelet modified (WMHB) and
unmodified hierarchical basis (HB), the $Y_{21}$ block represents the
last $n_j$ rows of the prolongation matrix;
$$
P_{j-1}^{j} = 
\left[ \begin{array}{l}
I \\ Y_{21}^{(j)} 
\end{array} \right].
$$ 
In the HB case, the $Y_{12}$ and $Y_{22}$ blocks are zero resulting
in a very simple form:
\begin{equation}
Y_{\mathrm{hb}} = \left[ \begin{array}{cc} I & 0 \\ Y_{21} & I  \end{array} \right]
\end{equation}

Then, the
original system (\ref{algebraicSystem}) is related to the HB system
through $Y$ as follows.
\begin{eqnarray}
A_{{\mathrm hb}}~ U_{{\mathrm hb}} & = & F_{{\mathrm hb}}, \label{HBlinearSystem} \\
A_{{\mathrm hb}} & = & Y^T A Y, \nonumber \\
U       & = & Y  ~U_{{\mathrm hb}}, \nonumber \\
F_{{\mathrm hb}} & = & Y^T~F \nonumber.
\end{eqnarray}

It will be shown later that the HBMG algorithm operates on sublocks of
the HB system (\ref{HBlinearSystem}). We explicitly express the
sublocks as follows.
$$ \begin{array}{ccccc}
\left[ \begin{array}{cc}
A_{{\mathrm hb},11} & A_{{\mathrm hb},12} \\
A_{{\mathrm hb},21} & A_{{\mathrm hb},22}
\end{array} \right] & = &
\left[ \begin{array}{cc}
I & Y_{21}^T \\ 0 & I
\end{array} \right] &
\left[ \begin{array}{cc}
A_{11} & A_{12} \\
A_{21} & A_{22}
\end{array} \right] &
\left[ \begin{array}{cc}
I & 0 \\ Y_{21} & I
\end{array} \right],
\end{array} 
$$ 

$$\begin{array}{llllllllll}
A_{{\mathrm hb},11} & = & A_{11}& + & A_{12} Y_{21} & + & Y_{21}^T A_{21} & + & Y_{21}^T A_{22} Y_{21} \\
A_{{\mathrm hb},12} & = &       &   & A_{12}   &   &            & + & Y_{21}^T A_{22}  \\
A_{{\mathrm hb},21} & = &       &   &          &   & A_{21}     & + & A_{22} Y_{21} \\
A_{{\mathrm hb},22} & = &       &   &          &   &            &   & A_{22},
\end{array}
$$ 

Next, we briefly describe the HBMG Algorithm. The HBMG routine can interpreted
as an iterative process for solving the system (\ref{algebraicSystem})
with an initial guess of $U_j$.
\begin{algorithm} .\label{algo:hbmg}\\
\begin{enumerate}
\item smooth $A_{{\mathrm hb},22} U_{{\mathrm hb},2} = F_{{\mathrm hb},2} - (A_{{\mathrm hb},21} U_{{\mathrm hb},1})$
\item form residual $R_{{\mathrm hb},1} = F_{{\mathrm hb},1} - (A_{{\mathrm hb},11} U_{{\mathrm hb},1}) - A_{{\mathrm hb},12} U_{{\mathrm hb},2}$
\item solve $A_{{\mathrm hb},11} U_{{\mathrm hb},1} = R_{{\mathrm hb},1}$
\item prolongate $U_{{\mathrm hb}} = U_{{\mathrm hb}} + P U_{{\mathrm hb},1}$
\item smooth $A_{{\mathrm hb},22} U_{{\mathrm hb},2} = F_{{\mathrm hb},2}- (A_{{\mathrm hb},21} U_{{\mathrm hb},1})$
\end{enumerate}
\end{algorithm}
Smoothing involves the approximate solution of the linear system by a
fixed number of iterations (typically one or two) with a method such
as Gauss-Seidel, or Jacobi.  In order to use HBMG as a preconditioner
for CG, one has to make sure the pre-smoother is the adjoint of the
post-smoother.  One should also note that the algorithm can be
simplified by first transforming the linear system
(\ref{algebraicSystem}) into the equivalent system
$$ A (U - U_j) = F - A U_j.$$
In this setting, the initial guess is zero, and the HBMG algorithm
recursively iterates towards the {\em error} with given {\em residual}
on the right hand side.  Simplification comes from the fact that terms
in the parentheses are zero in Algorithm~\ref{algo:hbmg}.

\subsubsection{The Computational and Storage Complexity of the HBMG method}
\label{sec:storage} 
Utilizing the block structure of the stiffness matrix $A^{(j)}$
in (\ref{defn:2by2}), the storage is in the following fashion:
$$
A^{(j)} = \left[ \begin{array}{ll}
\left[ \begin{array}{llll}
A_{11}^{(j)}&\rightarrow&\rightarrow&\rightarrow \\
\rightarrow &\rightarrow&\rightarrow&\rightarrow \\
\rightarrow &\rightarrow&\rightarrow&\rightarrow \\
\rightarrow &\rightarrow&\rightarrow&\rightarrow \\
\end{array} 
\right]_{N_{j-1} \times N_{j-1}}
&
\left[ \begin{array}{ll}
A_{12}^{(j)}& \times \\
\times      & \times \\
\times      & \times \\
\times      & \times \\
\end{array} 
\right]_{N_{j-1} \times n_j}      
\\ 
& \\
\left[ \begin{array}{llll}
A_{21}^{(j)}&\rightarrow&\rightarrow&\rightarrow\\
\rightarrow &\rightarrow&\rightarrow&\rightarrow \\
\end{array} 
\right]_{n_j \times N_{j-1} }     
&
\left[ \begin{array}{ll}
A_{22}^{(j)}&\rightarrow \\
\rightarrow &\rightarrow \\
\end{array} 
\right]_{n_j \times n_j } 
\end{array} 
\right],
$$
where $\rightarrow$ indicates a block stored rowwise, and $\times$
indicates a block which is not stored. By the symmetry in the bilinear
form $\langle DF(u)w,v \rangle$ wrt $w$ and $v$, $A_{12}^T = A_{21}$,
hence $A_{12}$ does not need to be stored.

Just like MG, the HBMG is an algebraic multilevel method as
well. Namely, coarser stiffness matrices are formed algebraically
through the use of variational conditions
\begin{equation} \label{varCond}
A^{(j-1)}={P_{j-1}^j}^T A^{(j)} P_{j-1}^j,~~~j=1,\ldots,J;~~A: = A^{(J)},
\end{equation}
where $P_{j-1}^j$ denotes the prolongation operator from level $j-1$
to $j$. Then, the only matrix to be stored for HB method is
$P_{j-1}^j$, in which $I$ is implicit, and therefore, does not have to
be stored;
$$
P_{j-1}^j = \left[ \begin{array}{l}
\left[ \begin{array}{llll}
I      & \times & \times & \times \\
\times & \times & \times & \times \\
\times & \times & \times & \times \\
\times & \times & \times & \times \\
\end{array} 
\right]_{N_{j-1} \times N_{j-1}}
\\
\left[ \begin{array}{llll}
Y_{21}^{(j)}          &\rightarrow&\rightarrow&\rightarrow\\
\rightarrow&\rightarrow&\rightarrow&\rightarrow \\
\end{array} 
\right]_{n_j \times N_{j-1} }
\end{array}
\right].
$$

In an adaptive scenario, new DOF are introduced in parts of the mesh
where further correction or enrichment is needed.  Naturally, the
elements which are marked by the error estimator shrink in size by
subdivision and the basis functions corresponding to fine nodes are
forced to change more rapidly than the ones that correspond to the
coarse nodes. Such rapidly changing components of the solution are
described as the {\em high frequency} components.  {\em Smoothing} is
an operation which corrects the high frequency components of the
solution, and is an integral part of the MG-like solvers. In MG, all
DOF are exposed to smoothing which then requires to store all blocks
of the stiffness matrix $A$. Unlike MG, the distinctive feature of the
HBMG is that smoothing takes places only on basis functions
corresponding to fine nodes. This feature leads us to make the crucial
observation that V-cycle MG exactly becomes the HBMG method if
smoothing is replaced by fine smoothing.  One can describe HBMG style
smoothing as {\em fine smoothing}.  For fine smoothing, HBMG only
needs to store the {\em fine-fine} interaction block $A_{22}$ of
$A$. It is exactly the fine smoothing that allows HB methods to have
optimal storage complexity.

In a typical case, $n_j$ is a small constant relative to $N_j$ and has
no relation to it. The HB method storage superiority stems from the
fact that on every level $j$ the storage cost is ${\cal O}(n_j)$. Fine
smoothing is used for high frequency components, and this requires to
store $A_{22}$ block which is of size $n_j \times n_j$. $A_{22}$ is
stored rowwise, and the storage cost is ${\cal O}(n_j)$. Coarse grid
correction is used recursively for low frequency components, and this
requires to store $A_{12}$ and $A_{21}$ blocks which are of size
$N_{j-1} \times n_j$ and $n_j \times N_{j-1}$ respectively. Due to
symmetry of the bilinear form, $A_{12}^T$ block is substituted by
$A_{21}$, and $A_{21}$ is stored rowwise, and the storage cost is
again ${\cal O}(n_j)$.

Further strength of HB methods is that computational cost per cycle is
${\cal O}(n_j)$ on each level $j$. The preprocessing computational cost,
namely computation of $A_{{\mathrm hb},11}, A_{{\mathrm hb},12}, A_{{\mathrm hb},21},
A_{{\mathrm hb},22}$ is ${\cal O}(n_j)$.  Hence, in an adaptive refinement scenario,
the overall computational complexity is achieved to be 
${\cal O}( N )$.

Let us observe by a fictitious example how MG fails to maintain
optimal storage complexity.  Let us assume that the finest level is
$J$, then $N=N_J,~~N_{J-1}=N-n_J$, and count the total storage at each
level. Here we take complexity constants to be $1$ for simplicity.
$$\begin{array}{lcl} 
{\rm level}~~J   &  :   & N      \\
{\rm level}~~J-1 &  :   & N-n_J  \\
\vdots     &      & \vdots \\
{\rm level}~~2   &  :   & N-n_J-n_{J-1} - \cdots - n_3 \\
{\rm level}~~1   &  :   & N-n_J-n_{J-1} - \cdots - n_3 - n_2. \\
\end{array}
$$
Adding up the cost at all levels, overall cost is: 
\begin{equation}
(J)N-(J-1)n_J-(J-2)n_{J-1}- \cdots - 2n_3-n_2 \label{eq:cost}.
\end{equation}  
Since $n_j$'s are small constants, then (\ref{eq:cost}) $ \leq NJ$,
implying that overall storage is ${\cal O}( NJ )$. If numerous
refinements are needed, or precisely, in an asymptotic scenario
(i.e. $J \rightarrow \infty$), the storage poses a severe problem.

\subsection{Sparse Matrix Products and the WMHB Implementation}
Our implementation relies on a total of four distinct sparse matrix
data structures: compressed column (COL), compressed row (ROW),
diagonal-row-column (DRC), and orthogonal-linked list (XLN). For
detailed description of the data structures, the reader can refer
to~\cite{ABH02a}.  Each of these storage schemes attempts to record
the location and value of the nonzeros using a minimal amount of
information.  The schemes differ in the exact representation which
effects the speed and manner with which the data can be retrieved.
XLN is an orthogonal-linked list data structure format which is the only
dynamically ``fillable'' data structure used by our methods.  By using
variable length linked lists, rather than a fixed length array, it is
suitable for situations where the total number of nonzeros is not
known {\it a priori}.

The key preprocessing step in the hierarchical basis methods, is converting
the ``nodal'' matrices and vectors into the hierarchical basis.  This
operation involves sparse matrix-vector and matrix-matrix products for
each level of refinement.  To ensure that this entire operation has linear
cost, with respect to the number of unknowns, the per-level change of
basis operations must have a cost of ${\cal O}( n_j )$, where
$n_j := N_j - N_{j-1} $ is the number of ``new'' nodes on level $j$.  For
the traditional multigrid algorithm this is not possible, since enforcing
the variational conditions operates on {\em all} the nodes on each level,
not just the newly introduced nodes.

For WMHB, the $Y_{12}$ and $Y_{22}$ blocks are computed using the mass
matrix, which results in the following formula:
\begin{equation}
Y_{\mathrm{wmhb}} = \left[ \begin{array}{cc} I &  - \mathrm{inv}\left[ M_{\mathrm{hb},11} \right] M_{\mathrm{hb},12}  \\
Y_{21} & I   - Y_{21} \mathrm{inv}\left[ M_{\mathrm{hb},11} \right] M_{\mathrm{hb},12} \end{array} \right],
\end{equation}
where the $\mathrm{inv}\left[ \cdot \right]$ is some approximation to the
inverse which preserves the complexity.
For example, it could be as simple as the inverse of the diagonal, or a
low-order matrix polynomial approximation.  The $M^{\mathrm{hb}}$ blocks are
taken from the mass matrix in the HB basis:
\begin{equation}
M_{\mathrm{hb}} = Y_{\mathrm{hb}}^T M_{\mathrm{nodal}} Y_{\mathrm{hb}}.
\end{equation}
For the remainder of this section, we restrict our attention to the WMHB
case.  The HB case follows trivially with the two additional subblocks
of $Y$ set to zero.

To reformulate the nodal matrix representation of the bilinear form in
terms of the hierarchical basis, we must perform a triple matrix product
of the form:
\begin{eqnarray*}
A_{\mathrm{wmhb}}^{(j)} & = & Y^{(j)^T} A_{\mathrm{nodal}}^{(j)} Y^{(j)} \\
 & =  & \left( I + Y^{(j)^T} \right) A_{\mathrm{nodal}}^{(j)}
\left( I + Y^{(j)} \right).
\end{eqnarray*}
In order to keep linear complexity, we can only copy $A_{\mathrm{nodal}}$
a fixed number of times, i.e. it cannot be copied on every level.  Fixed
size data structures are unsuitable for storing the product, since predicting
the nonzero structure of $A_{\mathrm{wmhb}}^{(j)}$ is just as difficult
as actually computing it.  It is for these reasons that we have chosen
the following strategy:  First, copy  $A_{\mathrm{nodal}}$ on the
finest level, storing the result in an XLN which will eventually become $A_{\mathrm{wmhb}}$.  Second, form the product pairwise, contributing the result to
the XLN.  Third, the last $n_j$ columns and rows of $A_{\mathrm{wmhb}}$ are
stripped off, stored in fixed size blocks, and the operation is repeated on
the next level, using the $A_{11}$ block as the new  $A_{\mathrm{nodal}}$:
\begin{algorithm}
   \label{alg:WMHBSTAB}
(Wavelet Modified Hierarchical Change of Basis)
{\small \begin{itemize}
\item Copy $A_{\mathrm{nodal}}^{(J)} \rightarrow A_{\mathrm{wmhb}}$ in XLN
format.
\item While $j > 0$
    \begin{enumerate}
    \item Multiply $ A_{\mathrm{wmhb}} = A_{\mathrm{wmhb}} Y $ as
	$$ \left[ \begin{array}{cc} A_{11} & A_{12} \\
	A_{21} & A_{22} \end{array} \right] +=
\left[ \begin{array}{cc} A_{11} & A_{12} \\
	A_{21} & A_{22} \end{array} \right]
\left[ \begin{array}{cc} 0 & Y_{12} \\
	Y_{21} & Y_{22} \end{array} \right]
$$

    \item Multiply $ A_{\mathrm{wmhb}} = Y^T A_{\mathrm{wmhb}} $ as
	$$ \left[ \begin{array}{cc} A_{11} & A_{12} \\
	A_{21} & A_{22} \end{array} \right] +=
\left[ \begin{array}{cc} 0 & Y_{21}^T \\
	Y_{12}^T & Y_{22}^T \end{array} \right]
\left[ \begin{array}{cc} A_{11} & A_{12} \\
	A_{21} & A_{22} \end{array} \right]
$$
    \item Remove $A^{(j)}_{21}$, $A^{(j)}_{12}$, $A^{(j)}_{22}$ blocks of
$A_{\mathrm{wmhb}}$ storing in ROW, COL, and DRC formats respectively.
    \item After the removal, all that remains of $A_{\mathrm{wmhb}}$ is its $A^{(j)}_{11}$ block.
    \item Let j = j - 1, descending to level $j-1$.
    \end{enumerate}
 \item End While.
    \item Store the last $A_{\mathrm{wmhb}}$ as $A_{\mathrm{coarse}}$
\end{itemize} }
\end{algorithm}

We should note that in order to preserve the complexity of the overall
algorithm, all of the matrix-matrix algorithms must be carefully
implemented.  For example, the change of basis involves computing the
products of $A_{11}$ with $Y_{12}$ and $Y_{12}^T$.  To preserve storage
complexity, $Y_{12}$ must be kept in compressed column format, COL.
For the actual product, the loop over the columns of $Y_{12}$ must be
ordered first, then a loop over the nonzeros in each column, then a loop
over the corresponding row or column in $A_{11}$.
It is exactly for this reason, that one must
be able to traverse $A_{11}$ both by row and by column, which is why we
have chosen an orthogonal-linked matrix structure for $A$ during the
change of basis (and hence $A_{11}$).

To derive optimal complexity algorithms for the other products, it is
enough to ensure that the outer loop is always over a dimension
of size $n_j$.  Due to the limited ways in which a sparse matrix
can be traversed, the ordering of the remaining loops will usually
be completely determined.
Further gains can be obtained in the symmetric case, since
only the upper or lower portion of the matrix needs to be explicitly
computed and stored.

\section{The Finite Element ToolKit for the Einstein Constraints}
   \label{sec:fetk}

\FETK{}~\cite{Hols2001a} (see also~\cite{HBW99,BHW99,BaHo98a})
is an adaptive multilevel finite element code in ANSI C developed 
by one of the authors (M.H.) over several years at Caltech and UC San Diego.
It is designed to produce provably accurate numerical solutions to
nonlinear covariant elliptic systems of tensor equations
on 2- and 3-manifolds in an optimal or nearly-optimal way.
\FETK{} employs {\em a posteriori} error estimation, adaptive simplex 
subdivision, 
unstructured algebraic multilevel methods, global inexact Newton methods, 
and numerical continuation methods for the highly accurate numerical solution
of nonlinear covariant elliptic systems on (Riemannian) 2- and 3-manifolds.
In this section, we describe some of the key design features of \FETK{}.
A sequence of careful numerical examples producing initial data
for an evolution simulation appear in~\cite{KAHPT07a}.

\subsection{The overall design of \FETK{}}

The finite element kernel library in \FETK{} is referred to
as \MC{} (or Manifold Code).
\MC{} is an implementation of Algorithm~\ref{alg:MC},
employing Algorithm~\ref{alg:newton} for nonlinear elliptic systems that
arise in Step~1 of Algorithm~\ref{alg:MC}.
The linear Newton equations in each iteration of Algorithm~\ref{alg:newton}
are solved with algebraic multilevel methods, and the algorithm is
supplemented with a continuation technique when necessary.
Several of the features of \FETK{} are somewhat unusual, allowing for the 
treatment
of very general nonlinear elliptic systems of tensor equations on domains
with the structure of 2- and 3-manifolds.
In particular, some of these features are:
\begin{itemize}
\item {\em Abstraction of the elliptic system}: The elliptic system 
     is defined only through a nonlinear weak form over the domain manifold,
     along with an associated linearization form, also defined everywhere on
     the domain manifold
     (precisely the forms $\langle F(u),v \rangle$
     and $\langle DF(u)w,v \rangle$ in the discussions above).
     To use the {\em a posteriori} error estimator, a third function
     $F(u)$ must also be provided (essentially the strong form of the problem).
\item {\em Abstraction of the domain manifold}: The domain manifold is
     specified by giving a polyhedral representation of the topology, along
     with an abstract set of coordinate labels of the user's interpretation,
     possibly consisting of multiple charts.
     \FETK{} works only with the topology of the domain, the connectivity
     of the polyhedral representation.
     The geometry of the domain manifold is provided only through the form
     definitions, which contain the manifold metric information, and through
     a {\tt oneChart()} routine that the user provides to resolve chart
     boundaries.
\item {\em Dimension independence}: The same code paths are taken for
     both two- and three-dimensional problems (as well as for
     higher-dimensional problems).
     To achieve this dimension independence, \FETK{} employs the simplex as its 
     fundamental geometrical object for defining finite element bases.
\end{itemize}
As a consequence of the abstract weak form approach to defining the problem,
the complete definition of the constraints in the Einstein equations
requires writing only 1000 lines of C to define the two weak forms,
and to define the {\tt oneChart()} routine.
Changing to a different problem, e.g., large deformation nonlinear elasticity,
involves providing only a different definition of the forms and a different
domain description.

\subsection{Topology and geometry representation in \FETK{}}

A datastructure called the {\em ringed-vertex} (cf.~\cite{Hols2001a})
is used to represent meshes of $d$-simplices of arbitrary topology.
This datastructure is illustrated in Figure~\ref{fig:river}.
\begin{figure}[tbh]
    \label{fig:polyman}
\begin{center}
\mbox{\myfigpng{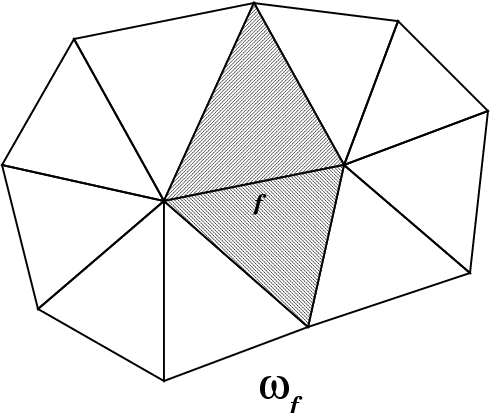}{1.4in}}
\hspace*{0.1cm}
\mbox{\myfigpng{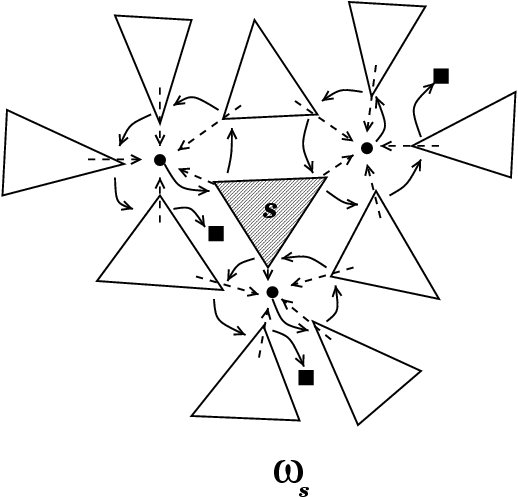}{1.4in}}
\end{center}
\caption{Polyhedral manifold representation.
The figure on the left shows two overlapping polyhedral (vertex) charts
consisting of the two rings of simplices around two vertices sharing an edge.
The region consisting of the two darkened triangles around the face $f$ is
denoted $\omega_f$, and represents the overlap of the two vertex charts.
Polyhedral manifold topology is represented by \FETK{} using the 
{\em ringed vertex}
datastructure.
The datastructure is illustrated for a given simplex $s$ in the figure on the
right; the topology primitives are vertices and $d$-simplices.
The collection of the simplices which meet the simplex $s$ at its vertices
(which then includes those simplices that share faces as well) is denoted
as $\omega_s$.
(The set $\omega_s$ includes $s$ itself.)
Edges are temporarily created during subdivision but are then destroyed 
(a similar ring datastructure is used to represent the edge topology). }
\label{fig:river}
\end{figure}
The ringed-vertex datastructure
is somewhat similar to the winged-edge, quad-edge, and edge-facet
datastructures commonly used in the computational geometry community for
representing 2-manifolds~\cite{Muck93}, but it can be used more generally
to represent arbitrary $d$-manifolds, $d=2,3,\ldots$.
It maintains a mesh of $d$-simplices with near minimal storage,
yet for shape-regular (non-degenerate) meshes, it provides $O(1)$-time access
to all information necessary for refinement, un-refinement, and discretization
of an elliptic operator.
The ringed-vertex datastructure also allows for dimension independent
implementations of mesh refinement and mesh manipulation, with one
implementation covering arbitrary dimension $d$.
An interesting feature of this datastructure is that the C structures used
for vertices, simplices, and edges are all of fixed size, so that a fast
array-based implementation is possible, as opposed to a less-efficient
list-based approach commonly taken for finite element implementations
on unstructured meshes.
A detailed description of the ringed-vertex datastructure, along with a
complexity analysis of various traversal algorithms, can be found
in~\cite{Hols2001a}.

Since \FETK{} is based entirely on the $d$-simplex, for adaptive refinement it
employs simplex bisection, using one of the simplex bisection strategies
outlined earlier.
Bisection is first used to refine an initial subset of the simplices in the 
mesh (selected according to some error estimates, discussed below), and then a
closure algorithm is performed in which bisection is used recursively on
any non-conforming simplices, until a conforming mesh is obtained.
If it is necessary to improve element shape, \FETK{} attempts to optimize the
following simplex shape measure function for a given $d$-simplex $s$, in an
iterative fashion, similar to the approach taken in~\cite{BaSm97}:
\begin{equation}
\eta(s,d) = \frac{2^{2(1-\frac{1}{d})} 3^{\frac{d-1}{2}} |s|^{\frac{2}{d}}}
                 {\sum_{0 \le i < j \le d} |e_{ij}|^2}.
\end{equation}
The quantity $|s|$ represents the (possibly negative) volume of the
$d$-simplex, and $|e_{ij}|$ represents the length of the edge that
connects vertex $i$ to vertex $j$ in the simplex.
For $d=2$ this is the shape-measure used in~\cite{BaSm97} with
a slightly different normalization.
For $d=3$ this is the shape-measure developed in~\cite{LiJo94} 
again with a slightly different normalization.
The shape measure above can be shown to be equivalent to the sphere ratio
shape measure commonly used; the sphere measure is used in the black hole
mesh generation algorithms we describe below.

\subsection{Discretization and adaptivity in \FETK{}}

Given a nonlinear form $\langle F(u),v \rangle$,
its linearization bilinear form
$\langle DF(u)w,v \rangle$,
a Dirichlet function $\bar{u}$, and collection of simplices
representing the domain, \FETK{} uses a default linear element to produce 
and then
solve the implicitly defined nonlinear algebraic equations for the
basis function coefficients in the
expansion~(\ref{eqn:soln}).
The user can also provide his own element, specifying the number of degrees
of freedom to be located on vertices, edges, faces, and in the interior of
simplices, along with a quadrature rule, and the values of the basis functions
at the quadrature points on the master element.
Different element types may be used for different components of a
coupled elliptic system.
The availability of a user-defined general element makes it possible to,
for example, use quadratic elements, which in the present setting would allow
for the differentiation of the resulting solutions to the constraints
for use as initial data in an evolution code.

Once the equations are assembled and solved (discussed below),
{\em a posteriori} error estimates are computed from the discrete
solution to drive adaptive mesh refinement.
The idea of adaptive error control in finite element methods is to estimate
the behavior of the actual solution to the problem using only a previously
computed numerical solution, and then use the estimate to build an improved
numerical solution by upping the polynomial order
($p$-refinement) or refining the mesh ($h$-refinement) where appropriate.
Note that this approach to adapting the mesh (or polynomial order) to the
local solution behavior affects not only approximation quality, but also
solution complexity: if a target solution accuracy can be obtained with
fewer mesh points by their judicious placement in the domain, the cost of
solving the discrete equations is reduced (sometimes dramatically) because
the number of unknowns is reduced (again, sometimes dramatically).
Generally speaking, if an elliptic equation has a solution with local singular
behavior, such as would result from the presence of abrupt changes in the
coefficients of the equation (e.g., the conformal metric components in the present case),
then adaptive methods tend to give dramatic improvements over
non-adaptive methods in terms of accuracy achieved
for a given complexity price.
Two examples illustrating bisection-based adaptivity patterns
(driven by a completely geometrical ``error'' indicator simply
for illustration) are shown in Figure~\ref{fig:refinements}.
\begin{figure}[tbh]
\begin{center}
\mbox{\myfigpng{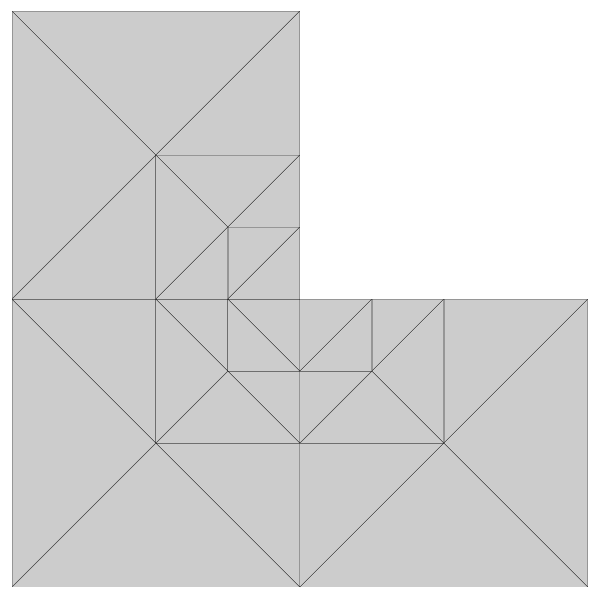}{0.75in}}
\hspace*{1.0cm}
\mbox{\myfigpng{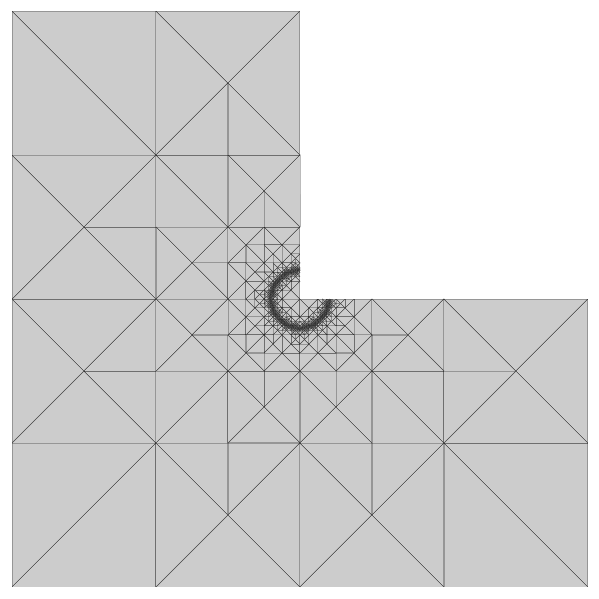}{0.75in}}
\hspace*{1.0cm}
\mbox{\myfigpng{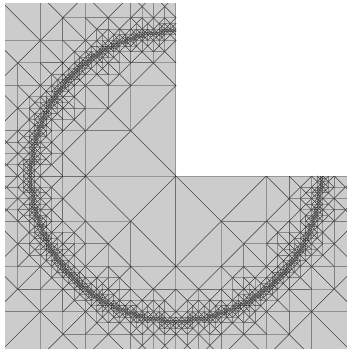}{0.75in}} \\
\mbox{\myfigpng{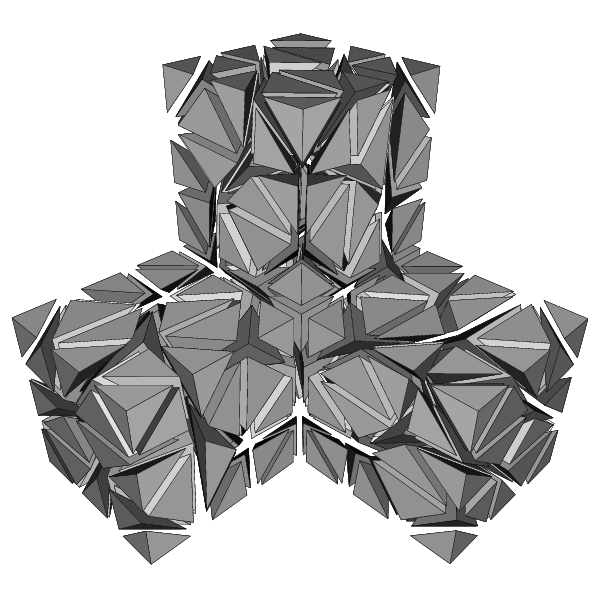}{0.75in}}
\hspace*{1.0cm}
\mbox{\myfigpng{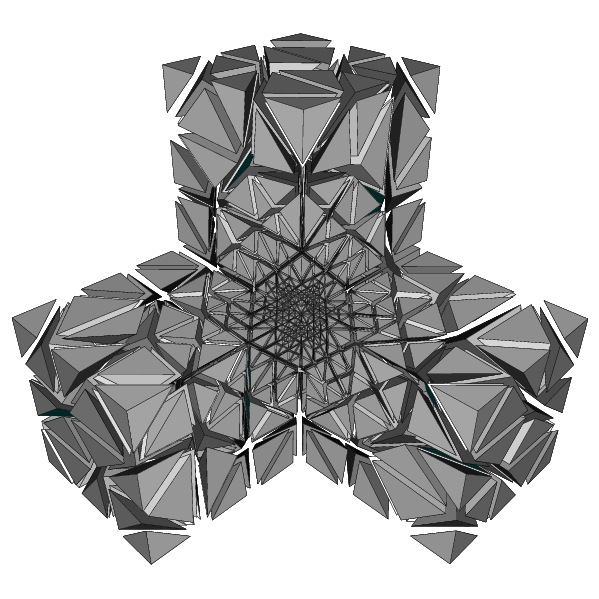}{0.75in}}
\hspace*{1.0cm}
\mbox{\myfigpng{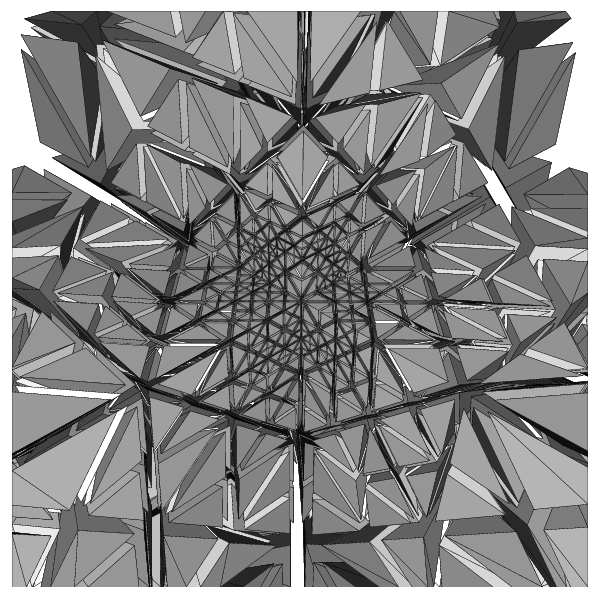}{0.75in}}
\end{center}
\caption{Examples illustrating the 2D and 3D adaptive mesh refinement
         algorithms in \FETK{}.  The right-most figure in each row shows
         a close-up of the area where most of the refinement occurred
         in each example.}
  \label{fig:refinements}
\end{figure}

To drive the adaptivity in \FETK{},
we employ a residual error estimator, based on computing an
upper bound on the nonlinear residual as given 
in~(\ref{eqn:indicator_residual}).
Reference~\cite{Hols2001a} contains the detailed derivation of the
indicator~(\ref{eqn:indicator_residual}) used in \FETK{} for
a general nonlinear elliptic system of tensor equations of the
form~(\ref{eqn:str1})--(\ref{eqn:str3}).
The error estimator provides a bound on
the error in the $W^{1,p}$-norm, $1<p<\infty$, for an approximation
of the form~(\ref{eqn:galerkin}) to the solution of the 
weak formulation~(\ref{eqn:weak}).

\subsection{Solution of linear and nonlinear systems with \FETK{}}
When a system of nonlinear finite element equations must be solved in \FETK{},
the global inexact-Newton Algorithm~\ref{alg:newton} is employed,
where the linearization systems are solved by linear multilevel methods.
When necessary, the Newton procedure in Algorithm~\ref{alg:newton} is
supplemented with a user-defined normalization equation for performing
an augmented system continuation algorithm.
The linear systems arising as the Newton equations in each iteration of
Algorithm~\ref{alg:newton} are solved using a completely algebraic multilevel
algorithm.
Either refinement-generated prolongation matrices $P_{j-1}^j$, or user-defined
prolongation matrices $P_{j-1}^j$ in a standard YSMP-row-wise sparse matrix format,
are used to define the multilevel hierarchy algebraically.
In particular, once the single ``fine'' mesh is used to produce the
discrete nonlinear problem $F(u)=0$ along with its linearization $Au=f$
for use in the Newton iteration in Algorithm~\ref{alg:newton}, a $J$-level
hierarchy of linear problems is produced algebraically using the
variational conditions~(\ref{varCond}) recursively.
As a result, the underlying multilevel algorithm is provably convergent
in the case of self-adjoint-positive matrices~\cite{HoVa97a}.
Moreover, the multilevel algorithm has provably optimal $O(N)$ convergence
properties under the standard assumptions for uniform
refinements~\cite{Xu92a},
and is nearly-optimal $O(N \log N)$ under very weak assumptions on
adaptively refined problems~\cite{BDY88}.

Coupled with the superlinear convergence properties of the outer inexact
Newton iteration in Algorithm~\ref{alg:newton}, this leads to
an overall complexity of $O(N)$ or $O(N \log N)$ for the solution of the
discrete nonlinear problems in Step~1 of Algorithm~\ref{alg:MC}.
Combining this low-complexity solver with the judicious placement of unknowns
only where needed due to the error estimation in Step~2 and the subdivision
algorithm in Steps~3-6 of Algorithm~\ref{alg:MC}, leads to a very
effective low-complexity approximation technique for solving a general
class of nonlinear elliptic systems on 2- and 3-manifolds.

\subsection{Availability of \FETK{}}

A fully-functional MATLAB version of the \MC{} kernel of \FETK{}
for domains with the structure
of Riemannian 2-manifolds, called \MCLite{},
is available under the GNU copyleft license at the following website:
\begin{center}
{{\tt http://www.FEtk.ORG/}}
\end{center}
\MCLite{} employs the ringed-vertex datastructure and implements the same
adaptivity and solution algorithms that are used in \MC{}.
A number of additional tools developed for use with \MC{} and \MCLite{}
are also available under a GNU license at the site above,
including MALOC (a Minimal Abstraction Layer for Object-oriented C)
and SG (an OpenGL-based X11/Win32 polygon display tool with a linear 
programming-based OpenGL-to-Postscript generator).
SG was used to generate the pictures of finite element meshes appearing
in this paper.

\subsection{Tetrahedral Mesh Generation for Single or Binary Compact Objects}
   \label{sec:meshgen}

The general problem of generating a finite element mesh given a
three dimensional domain is quite complex and has been the subject of
intensive research over the past 40 years (for a relatively old comprehensive
overview the reader is referred to~\cite{STHE96}).
However, for the purpose of generating a tetrahedral mesh on a domain
with the geometry necessary to describe a binary collision between compact
objects we do not need a completely general method.
In this section we will describe a rather simple method for generating high
quality
coarse meshes suitable for the computation of the initial data describing
binary systems. The only restriction of the method is that the geometry of
the domain must have an axis of symmetry. We note that this restriction
applies only to the domain geometry and not the physics (i.e., the source
terms in the constraint equations do not have to have any symmetries).

\begin{figure}[tbh]
\begin{center}
\mbox{\myfig{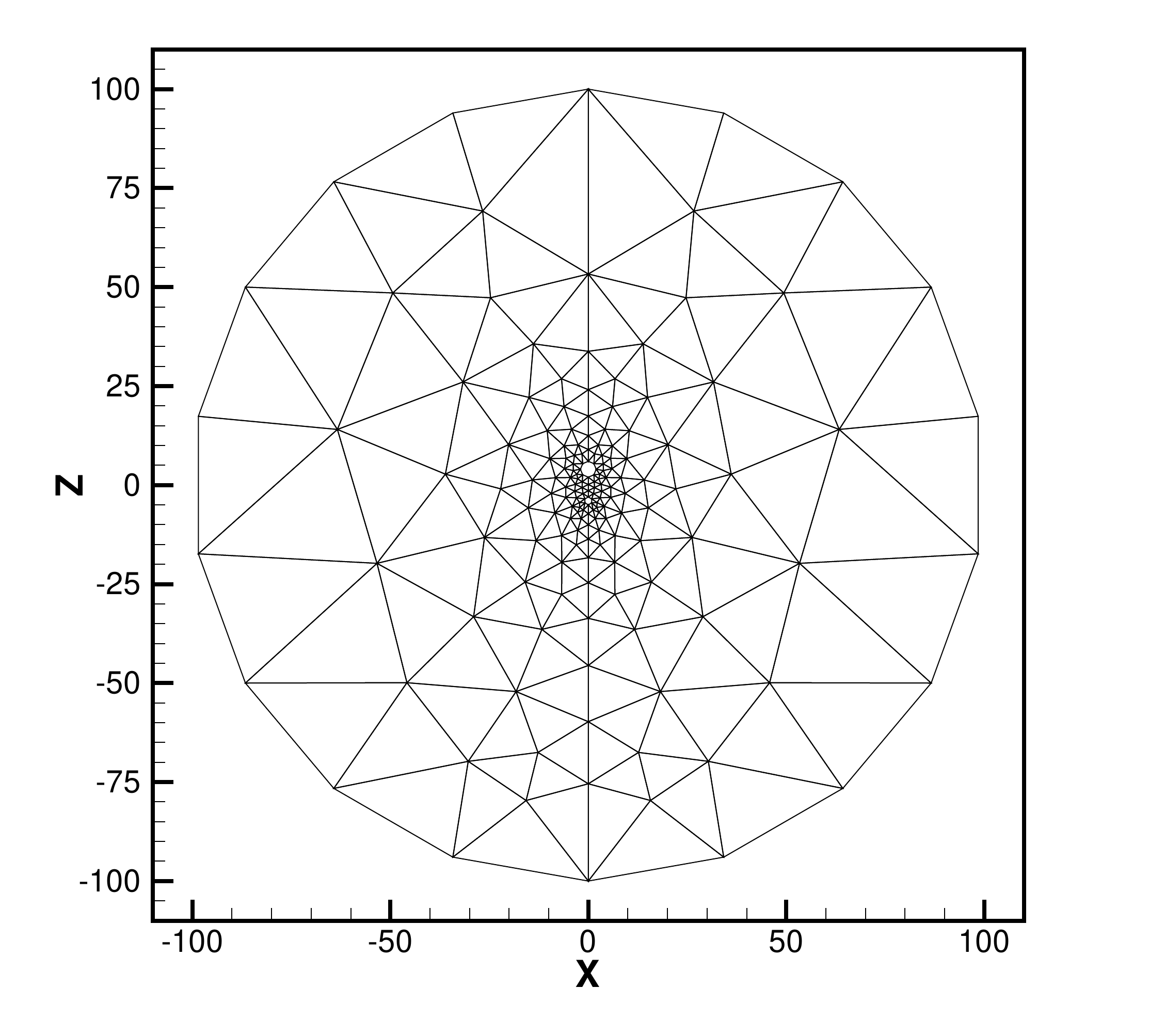}{3.2in}}
\end{center}
\caption{The planar triangulation step in the construction of a finite element 
mesh for a binary black hole collision.  
The figure shows the discretization of the entire domain;
due to the mirror symmetry of the domain only the region
$x\geq 0$ was meshed.}
\label{fig:meshgen1}
\end{figure}

A typical domain for a binary black hole collision consists of the interior
of a large sphere (the outer boundary) with two smaller spheres removed 
(the throats of the holes).  Since the exact placement of the outer sphere 
relative to the throats is not important we lose no generality in making the
centers of the three spheres colinear.  In this case the geometry is symmetric
about the line joining the three centers.  Choose this line to be the z axis.
Then the intersection of the domain with the x-z plane consists of a circle 
with two smaller circles removed from its interior.  The method consists of
first meshing this planar domain with triangular elements and then rotating the
elements around the z-axis to form triangular tori.  These tori are then 
subdivided into tetrahedra in such a way as to produce a conforming mesh.

We illustrate the method with an example.  Let the outer boundary have radius
100 and let the two throats be centered at $z = \pm 4$ with radii 2 and 1
respectively.  The geometry is chosen for illustrative purposes.  We discretize
the outer circle of the planar domain with 20 line segments and the inner
circles with 8 and 10 line segments respectively.  Two views of the 
triangulation are shown in figures~\ref{fig:meshgen1} and~\ref{fig:meshgen2}.

\begin{figure}[tbh]
\begin{center}
\mbox{\myfig{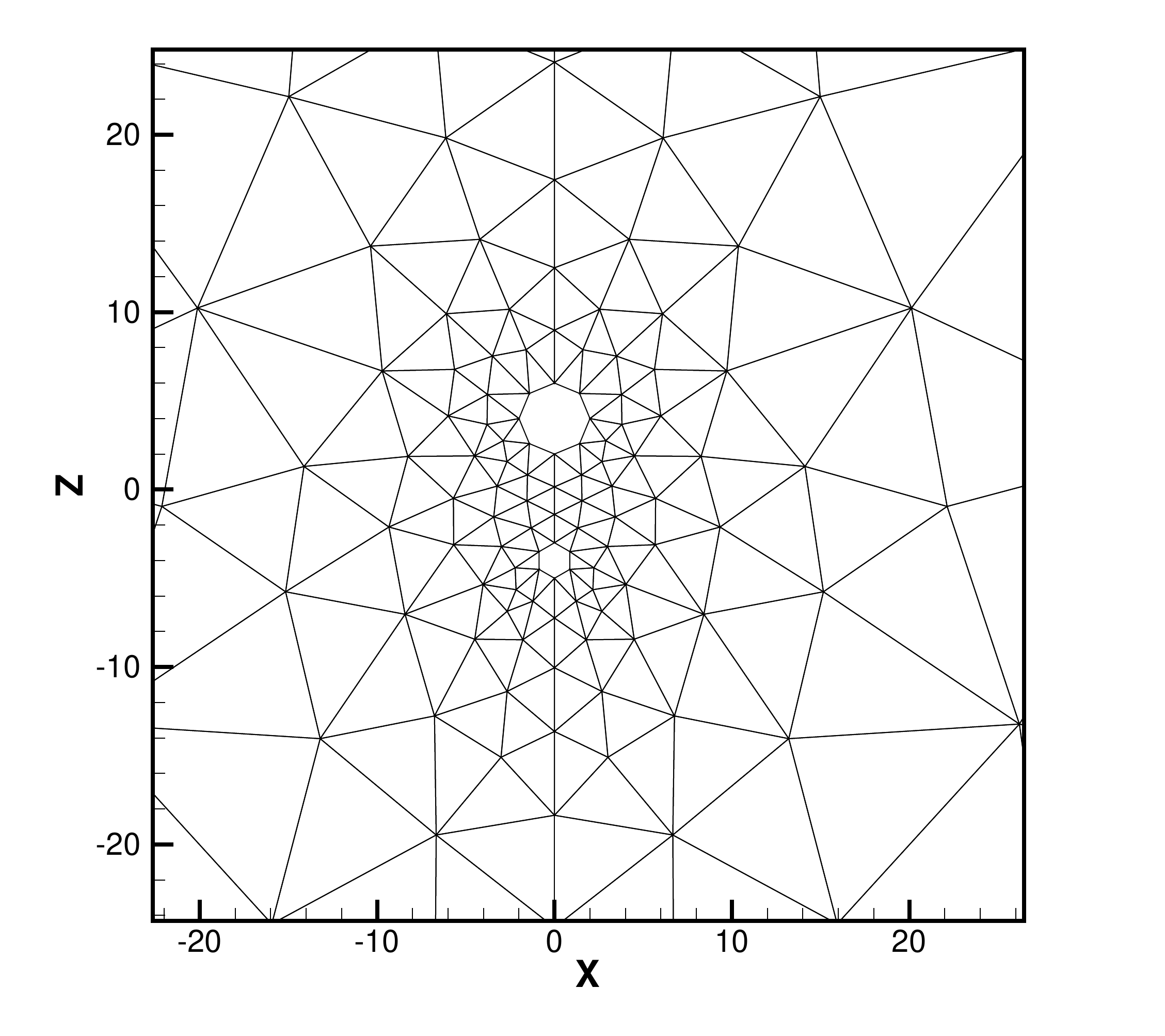}{3.2in}}
\end{center}
\caption{A closeup of the mesh from Figure~\ref{fig:meshgen1}
showing the details of the mesh around the two holes.}
\label{fig:meshgen2}
\end{figure}

The triangulation algorithm one employs for this part of the mesh generation
is not important.
In particular, one could use any one of a number of off-the-shelf codes
freely available on the internet (for example,
see {\tt www-2.cs.cmu.edu/$\sim$quake/triangle.html}).
A list of commercial and free mesh generators can be found at\\
{\tt www-users.informatik.rwth-aachen.de/$\sim$roberts/software.html}.\\
The triangulation code we have chosen to use is one developed
by one of the authors (D.B.).
We have the ability to produce highly graded meshes since the inner circles of
the planar domain will in general be much small than the outer boundary circle.
We would like to emphasize the distinguishing feature of the triangulation code. 
Although there are various codes that generate highly grades meshes with low quality,
the code developed by D.B. is superior than those because it generates highly graded meshes with 
high quality.

\begin{figure}[tbh]
\begin{center}
\mbox{\myfig{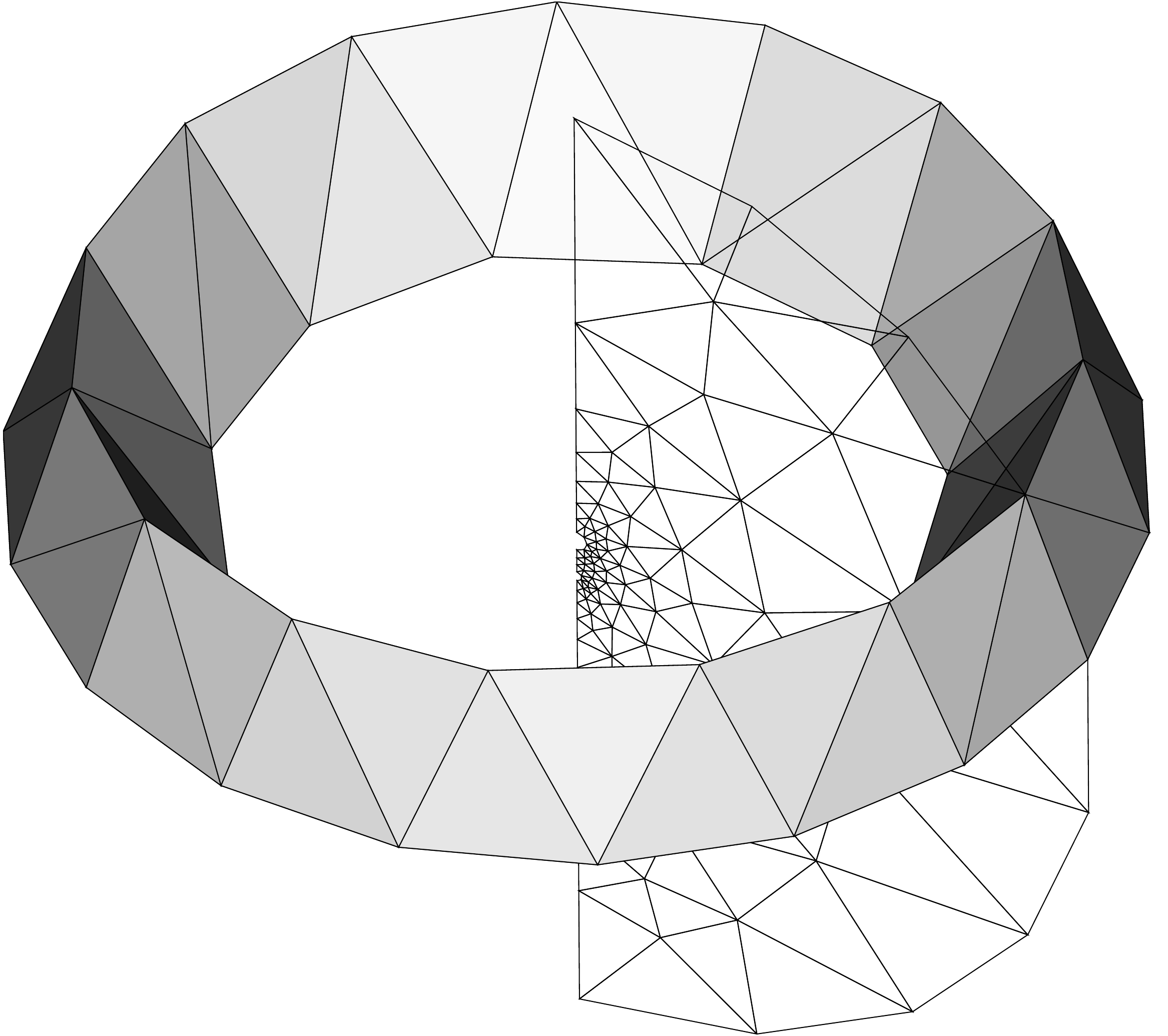}{3.0in}}
\end{center}
\caption{The second stage of mesh generation.  
One of the discretized tori is shown along with the planar 
mesh from figure~\ref{fig:meshgen1}.}
\label{fig:meshgen34}
\end{figure}

Once the planar triangulation is obtained the tetrahedra are constructed as
follows.
Each vertex of the triangulation defines a circle in the three-dimensional
space by rotation about the z-axis.
These circles are discretized by computing their circumference and dividing
by a measure of the size of the desired tetrahedra.
The measure we have used is the average distance to the neighbors of the
vertex in the triangulation.
For vertices with a neighbor on the axis of symmetry the discretization of
the circle is limited to between five and seven line segments, regardless
of the circumference or average neighbor distance.
This limits
the outdegree of the vertices which lie on the z-axis and prevents elements
with unusually small angles from being formed there.
Once each triangular torus is discretized along each of its edges it is a 
relatively simple matter to subdivide each torus into tetrahedra taking care
that the elements thus formed are reasonably shaped and consistent with their
neighbors.
In figure~\ref{fig:meshgen34} we show one of these tori and in 
figure~\ref{fig:meshgen3} we show the surface of the final mesh.

The quality of the meshes produced by this method is normally very good.
One measure of the quality of a tetrahedron which is often used is the ratio
of the radii of its circumscribing to inscribing spheres, $\alpha = R / 3r$.
The normalization factor is such that a regular tetrahedron has $\alpha=1$
and it has been proven~\cite{LiJo94} that this is a global minimum, i.e.,
no tetrahedron has $\alpha < 1$.
The maximum aspect ratio in the mesh shown in figure~\ref{fig:meshgen3}
is 1.91 while the average is 1.26. 
By comparison the aspect ratio of the reference element 
with vertices located at $(0,0,0)$, $(1,0,0)$, $(0,1,0)$, and 
$(0,0,1)$ is 1.37.

\begin{figure}[tbh]
\begin{center}
\mbox{\myfig{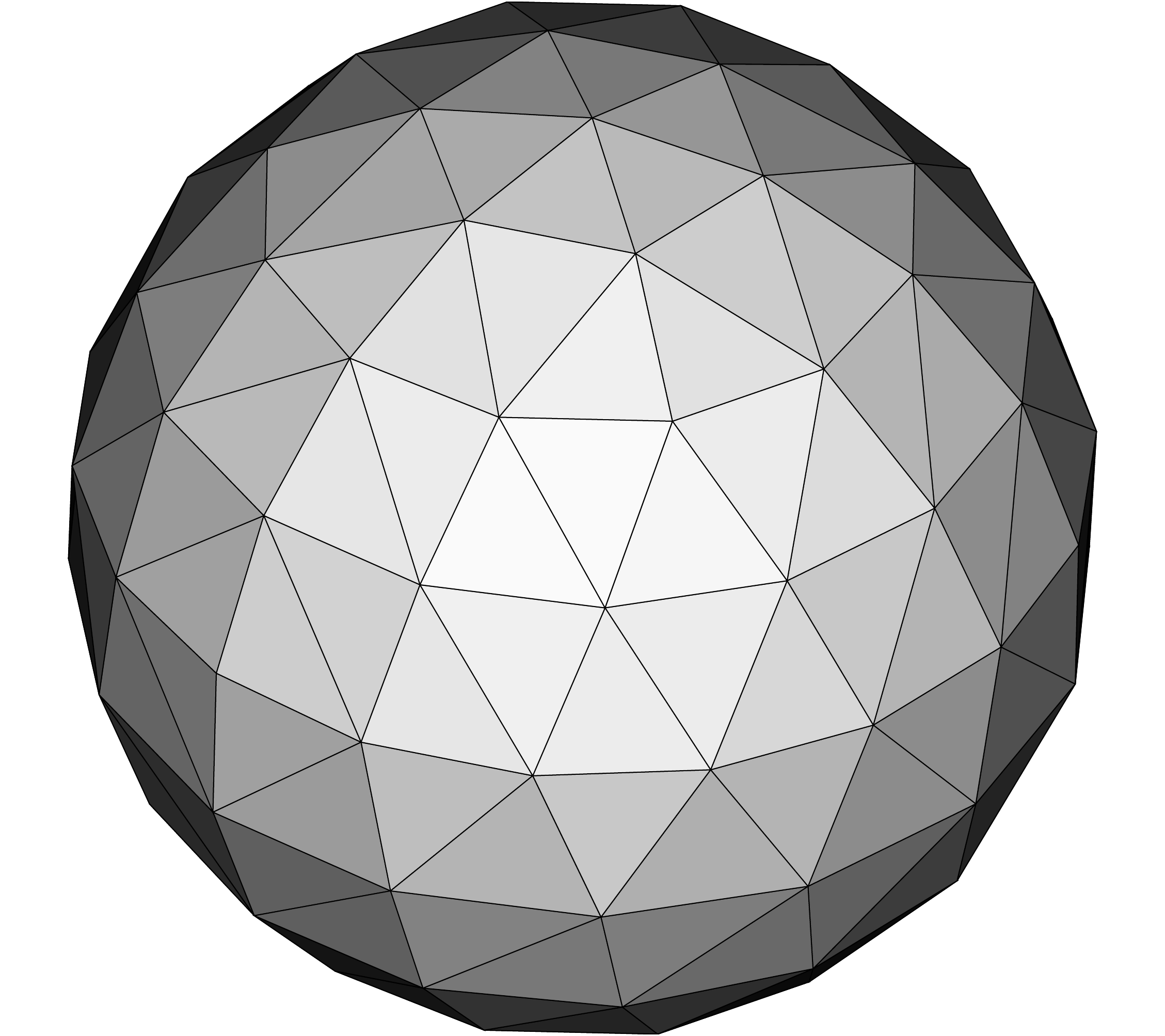}{1.49in}}
\mbox{\myfig{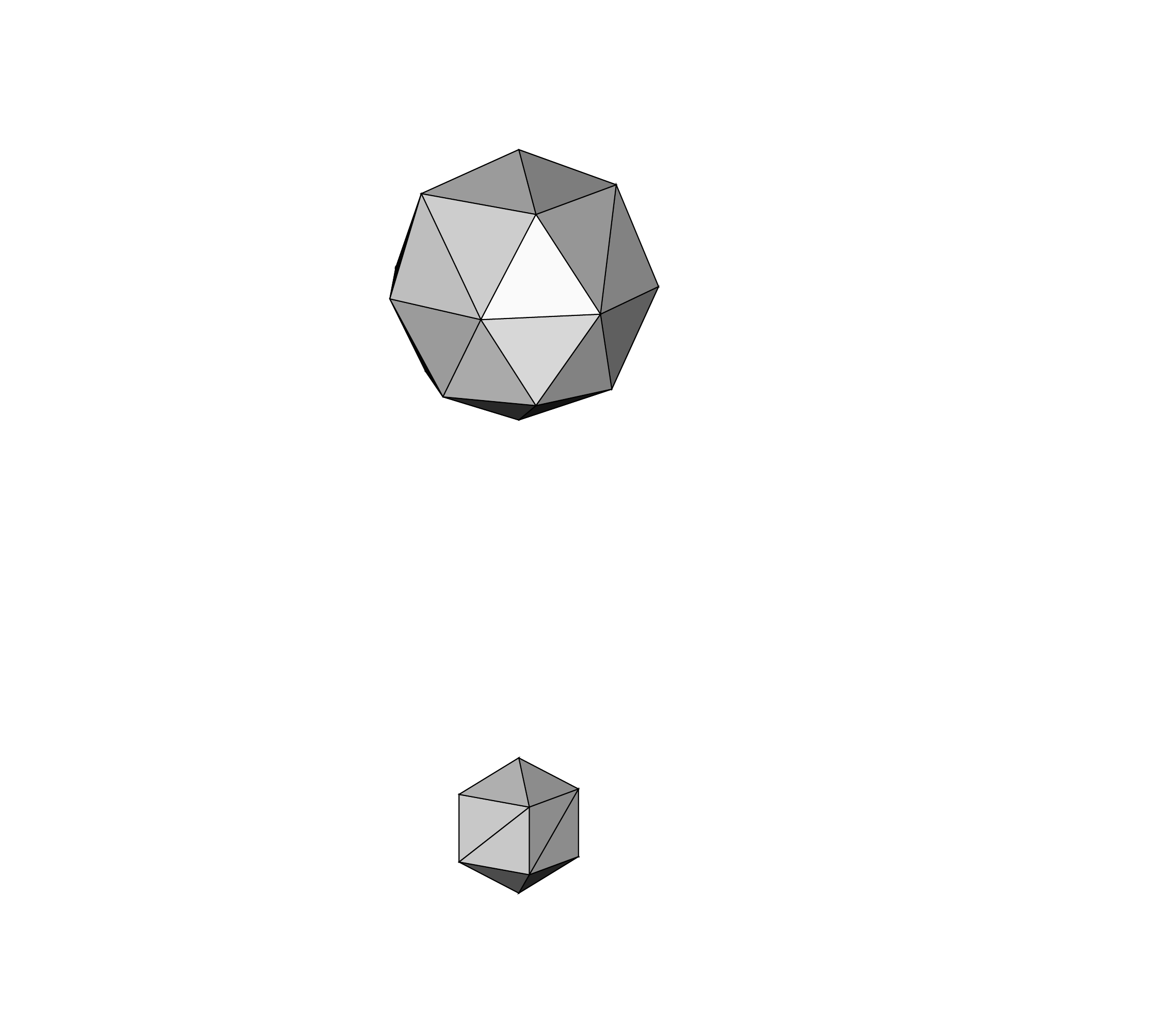}{1.49in}}
\end{center}
\caption{The final mesh generated from the planar domain shown in 
figure~\ref{fig:meshgen1}.  
On the left the surface of the outer boundary is shown, 
on the right the surfaces of the two holes.}
\label{fig:meshgen3}
\end{figure}

Finally, we note that although the example above describes
the generation of a mesh for a binary black hole simulation the method is
not restricted to this type of problem.  We can also generate meshes suitable
for any type of neutron star - black hole collision by replacing one or both of
the inner circles of the planar domain with a set of curves which are preserved
on the mesh.
The curves may represent (expected) isosurfaces of density or some other
quantity.  In Figures~\ref{fig:meshgen7} and~\ref{fig:meshgen8}
we show the mesh generated for
a collision between two neutron stars, one expected to be oblate and one
expected to be prolate.
Note that the initial data computed on such a mesh may or may not
have stars of such shape, we have merely chosen a coarse mesh which we
expect to reflect the nature of the solution to the constraints equations.
The actual mesh which solves the numerical problem to a given tolerance may
have a much different node distribution.
\begin{figure}[tbh]
\begin{center}
\mbox{\myfig{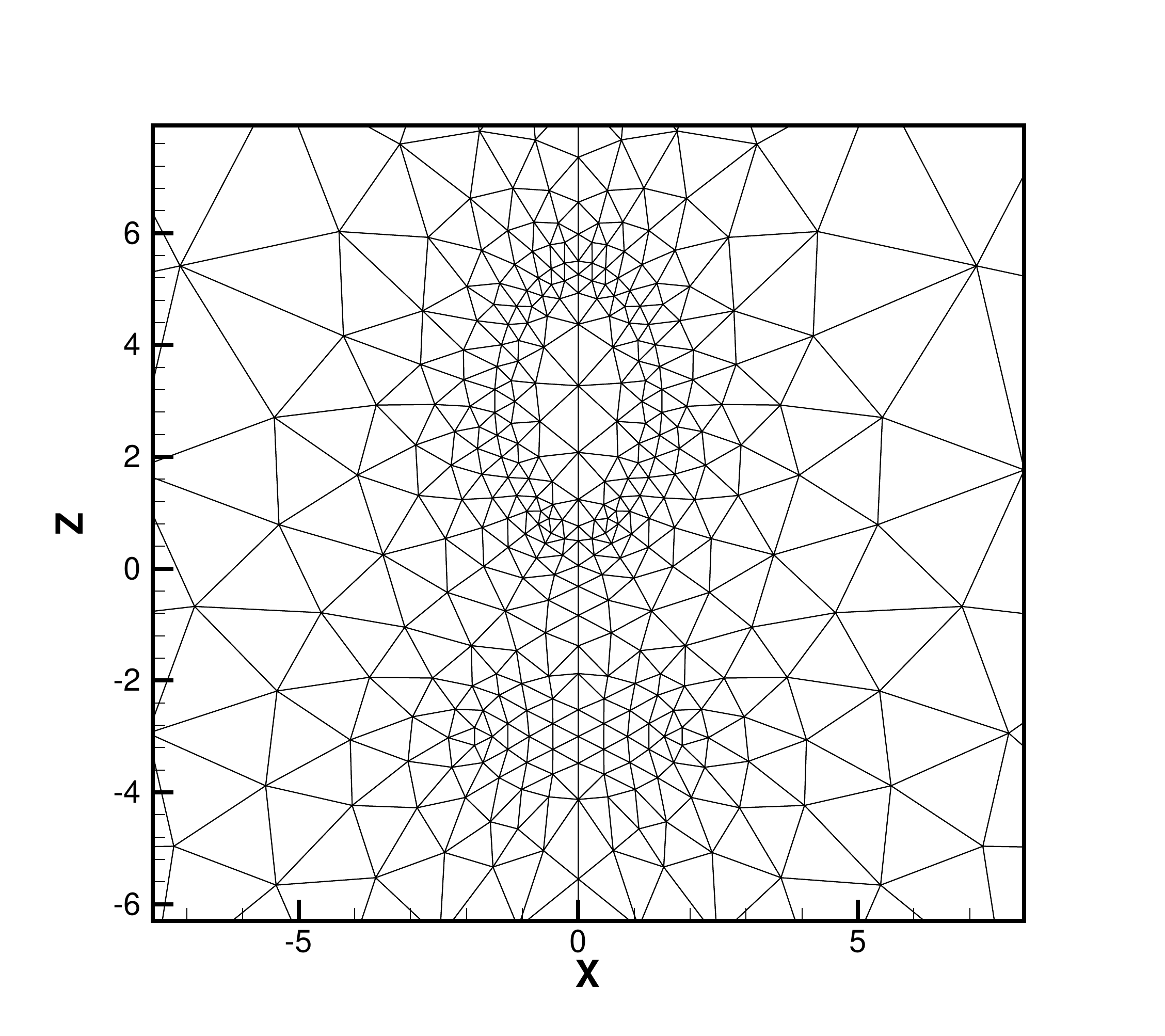}{3.0in}}
\end{center}
\caption{A mesh suitable for the calculation of the collision between two
stars, one initially prolate and one initially oblate.
The figure shows shows the planar triangulation about the two stars.}
\label{fig:meshgen7}
\end{figure}

\begin{figure}[tbh]
\begin{center}
\mbox{\myfig{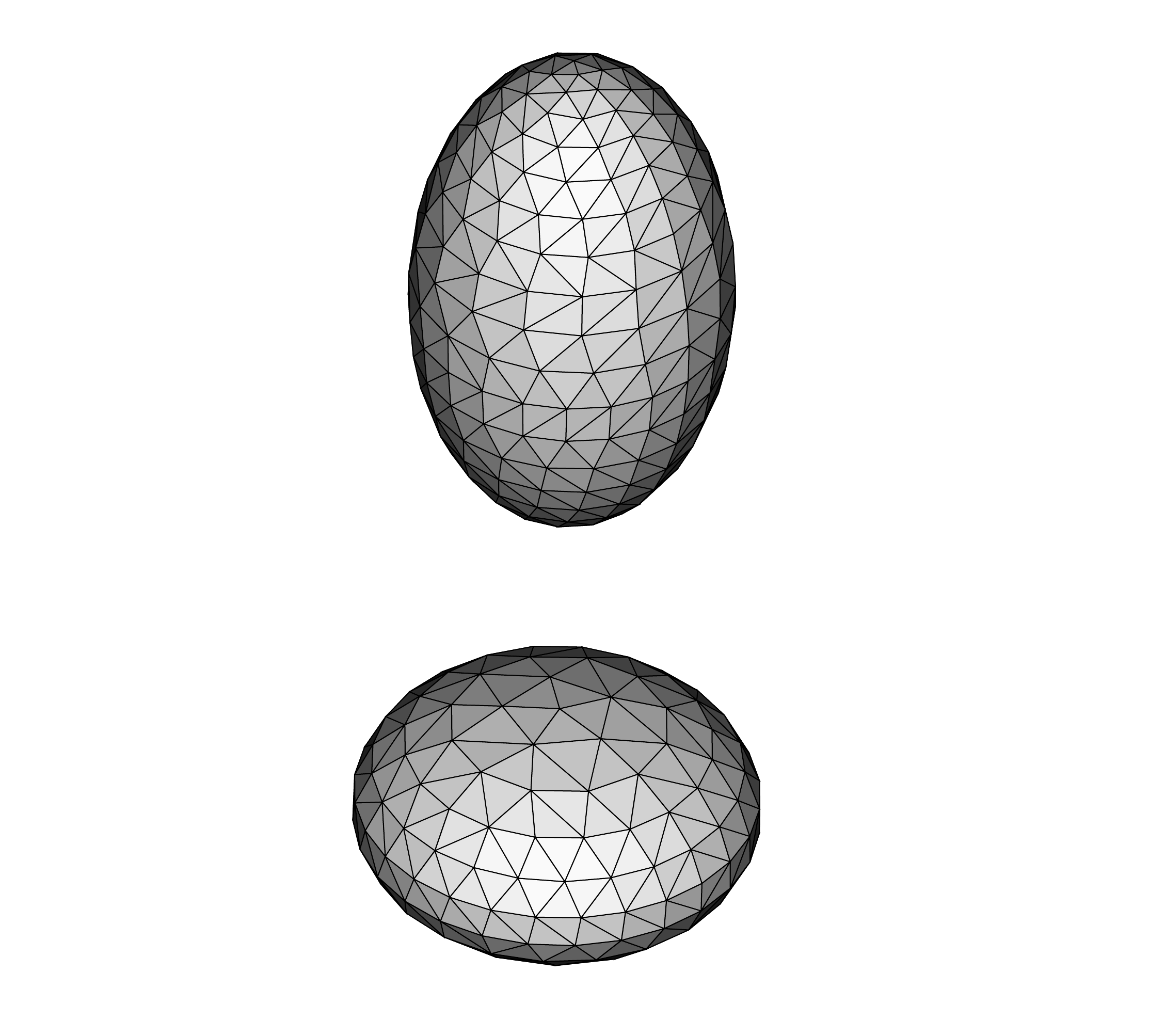}{3.0in}}
\end{center}
\caption{The ``surfaces'' of the two stars represented
by the mesh in Figure~\ref{fig:meshgen8}.}
\label{fig:meshgen8}
\end{figure}

\subsection{Computing Conformal Killing Vectors}
\label{app:B}

We describe a simple algorithm for computing conformal Killing
vectors and their Petrov-Galerkin approximations.
It can be shown that for Dirichlet and Robin
boundary conditions the 
homogeneous version of the momentum constraint
\begin{equation}
\hat{D}_a (\hat{L}V)^{ab} = 0
\label{KillingAppC}
\end{equation}
has only the trivial solution $V^a=0$.  However, using a pure Neumann
condition
removes the well-posedness of the problem and leads to a method for
computing
the conformal Killing vectors of $\hat{\gamma}_{ab}$.
Suppose that $K^a$ satisfies~(\ref{KillingAppC}).
This implies
\begin{equation}
0=(\hat{D}_a(\hat{L}K)^{ab},V^a)_{L^2(\calg{M})}
     = A(K^a,V^a), 
\end{equation}
$\forall~ V^a \in H^1_{0,D}$.
This can be posed as the generalized eigenvalue problem
\begin{equation}
  \label{eqn:find_killing}
\text{Find}~~K^a \in H^1_{0,D} ~\text{s.~t.}~
F(K^a,V^a) = \epsilon G(K^a,V^a),
\end{equation}
$\forall~ V^a \in H^1_{0,D}$, where
\begin{eqnarray}
F(K^a,V^a) &=& \int_M 2\mu (\hat{E}K)_{ab} (\hat{E}V)^{ab} dx,
\\
G(K^a,V^a) &=& \int_M \hat{D}_aK^a \hat{D}_bV^b dx.
\end{eqnarray}
If there exists a $K^a$ such that this holds for some $\epsilon \neq 0$
 then the 
operator $F-\epsilon G$ arising from the form through the
Riesz representation theorem is singular.  
For a given domain with a tetrahedral
decomposition, \FETK{} can discretize~(\ref{eqn:find_killing}) to produce
the matrix versions of $F$ and $G$ which can then be given to a
general eigenvalue package such as EISPACK.
The result would be a set of eigenvalues and eigenvectors.
The eigenvectors corresponding to the eigenvalue $\epsilon = 2/3$
form an (orthogonal) basis for the kernel of the discrete momentum
operator, representing a discrete approximation to the space
of conformal Killing vectors of $\hat{\gamma}_{ab}$.

\subsection{Computation of the ADM Mass on Adaptive Meshes}
\label{app:C}

Here we describe an alternative expression for the ADM mass that is more
appropriate for adaptive methods than that usually used in numerical
relativity.
The ADM mass is normally computed by the covariant surface integral
\begin{equation}
M = -\frac{1}{2 \pi} \int_{R} \hat{n}_a \hat{\Delta}^{a} \phi \; ds
\label{eq:adm1}
\end{equation}
where $R$ indicates the integral is to be taken over the outer sphere only.
The mass is actually defined as the limit $R \rightarrow \infty$ and assumes
that the equivalent energy associated with the conformal metric vanishes.
However, if this is the case, and if the outer sphere is reasonably far away
from the region where the source terms are varying rapidly then the adaptive
nature of \FETK{} will ensure that the outer sphere remains relatively coarsely
meshed, even in cases where the solution is computed to a high degree of 
accuracy.  This is because a $1/r$ variation in $\phi$ can be computed 
accurately with only a few elements.  Hence the outer boundary will remain 
coarsely meshed and the surface integral above will not be able to be
estimated with any degree of accuracy.  However, by 
converting~(\ref{eq:adm1}) to a volume integral using the Gauss theorem,
and by then evaluating it along with any resulting surface integrals over 
inner boundaries, it is possible to gain from the locally adapted mesh
rather than lose from it.
In particular, one has:
\begin{equation}
M = -\frac{1}{2 \pi} \left(
\int_{\calg{M}} \hat{\Delta} \phi \; dx
 - \int_{\partial \calg{M}-R} \hat{n}_a \hat{\nabla}^{a} \phi \; ds
\right).
\end{equation}

  \begin{figure}[!htp]
  \begin{center}
  \begin{tabular}{ccc}
  \includegraphics[width=0.25\textwidth]{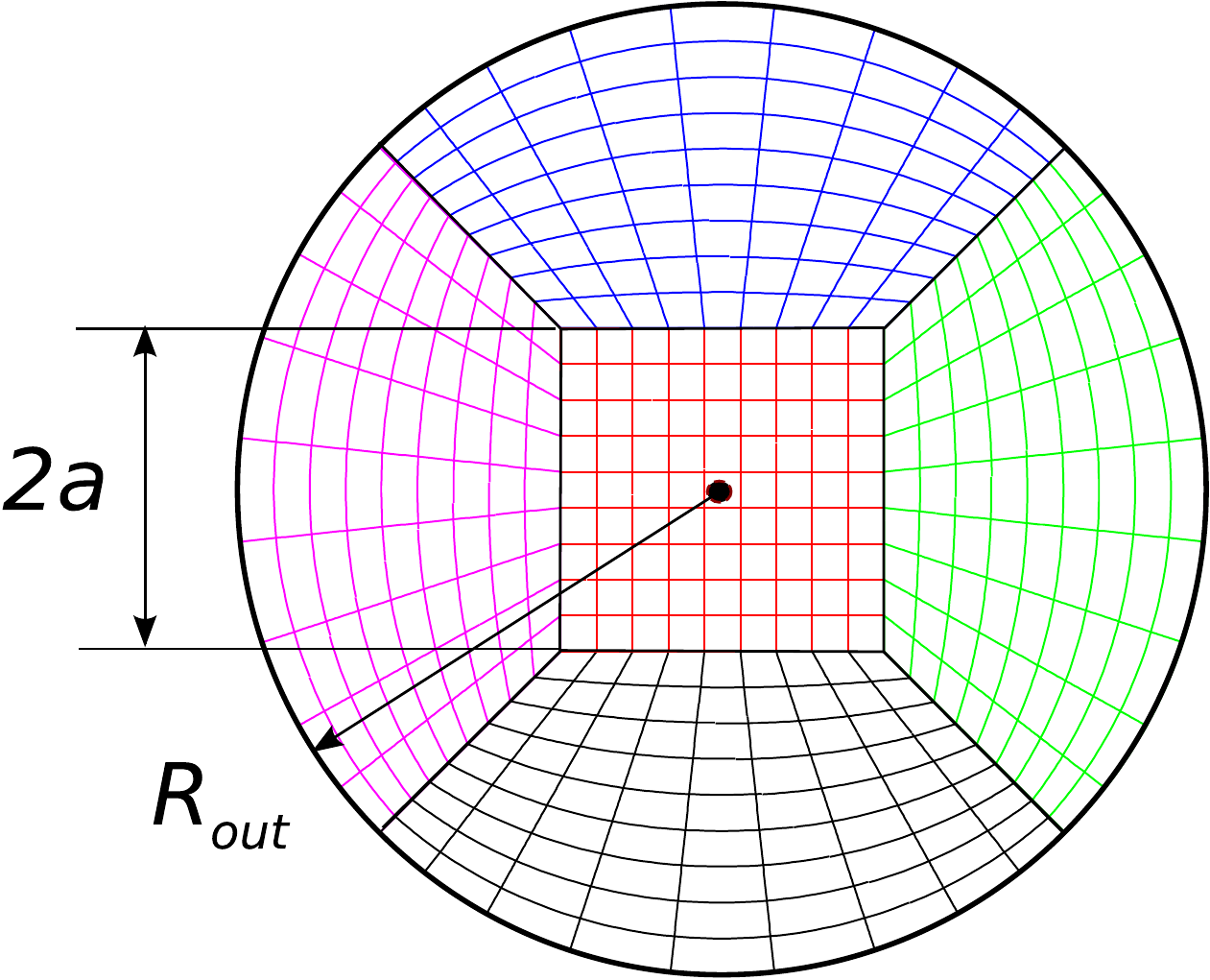} &
  \includegraphics[width=0.27\textwidth]{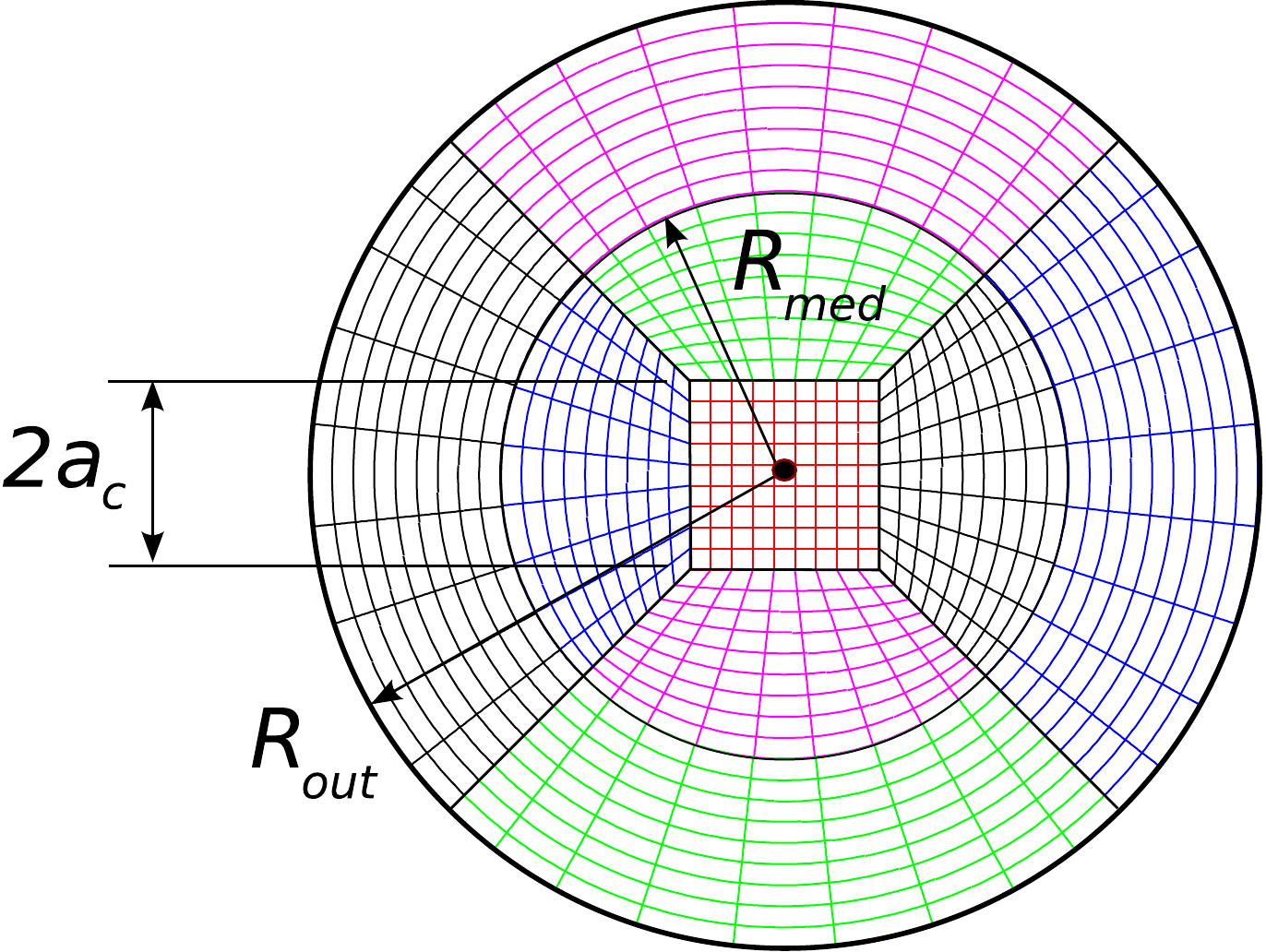} &
  \includegraphics[width=0.20\textwidth]{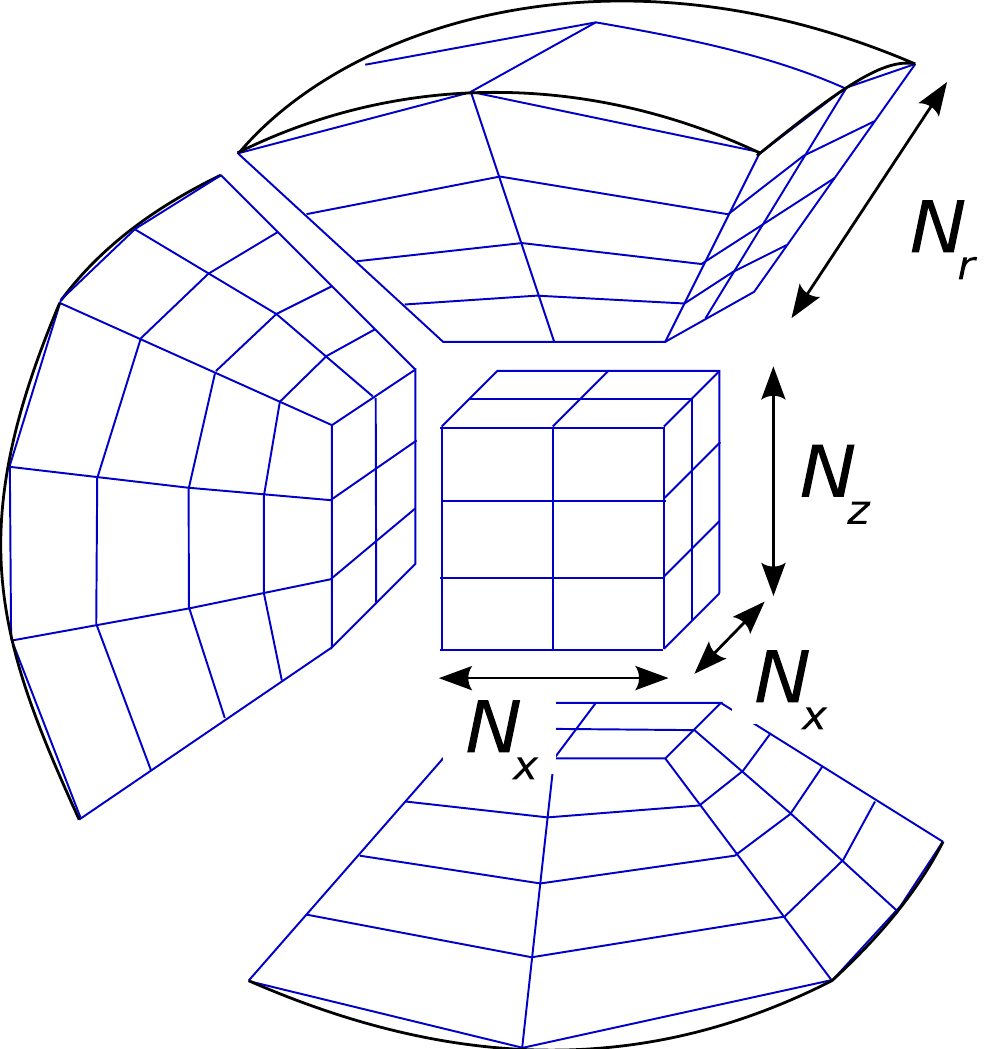} \\
  {\tiny (a)} & {\tiny (b)} & {\tiny (c)}
  \end{tabular}

  \caption{
    An equatorial cut of (a) seven-block and (b) thirteen-block systems.
    (c) Grid dimensions for the seven-block system.
  }
  \label{F:patchsystems}
  \end{center}
  \end{figure}

  \begin{figure}[htbp]
  \begin{center}
  \begin{tabular}{cc}
  \includegraphics[width=0.45\textwidth]{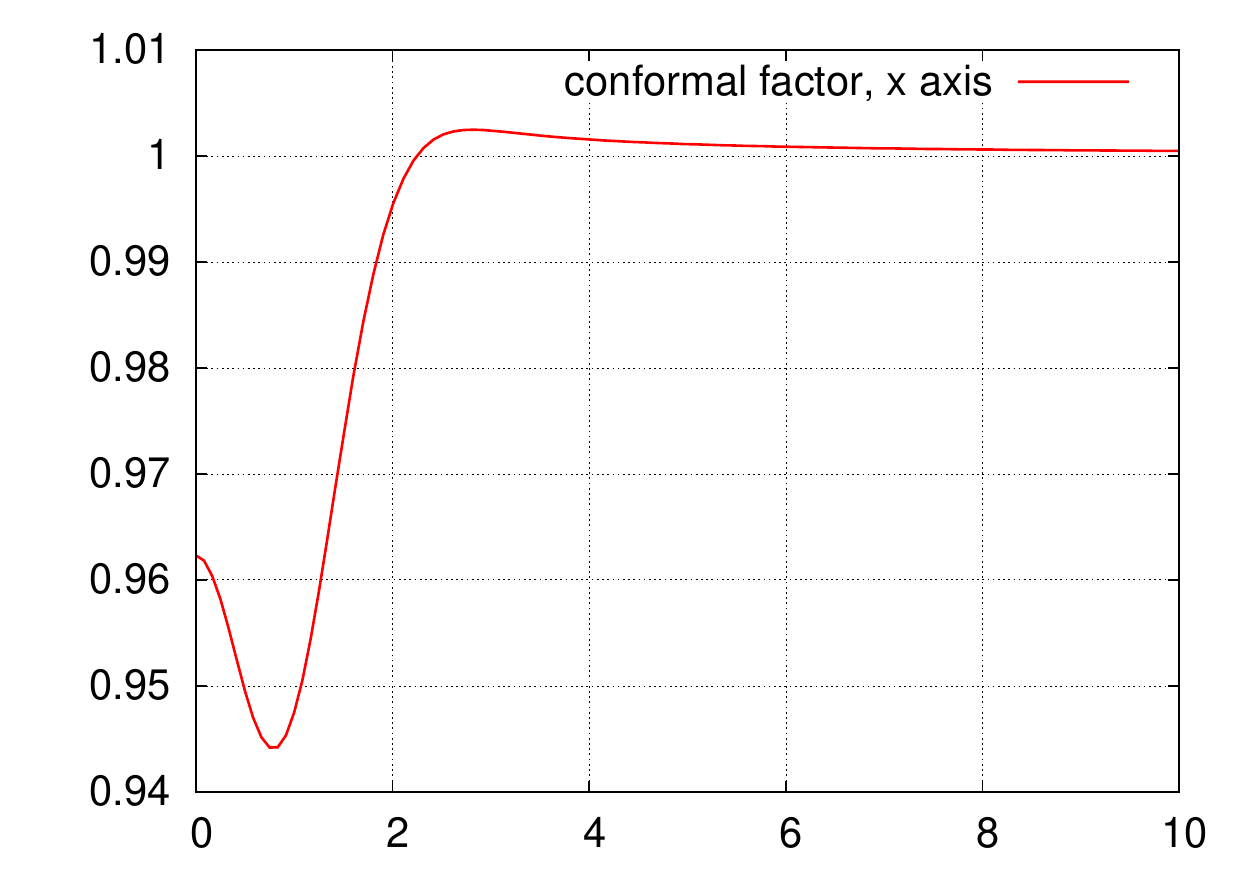} &
  \includegraphics[width=0.45\textwidth]{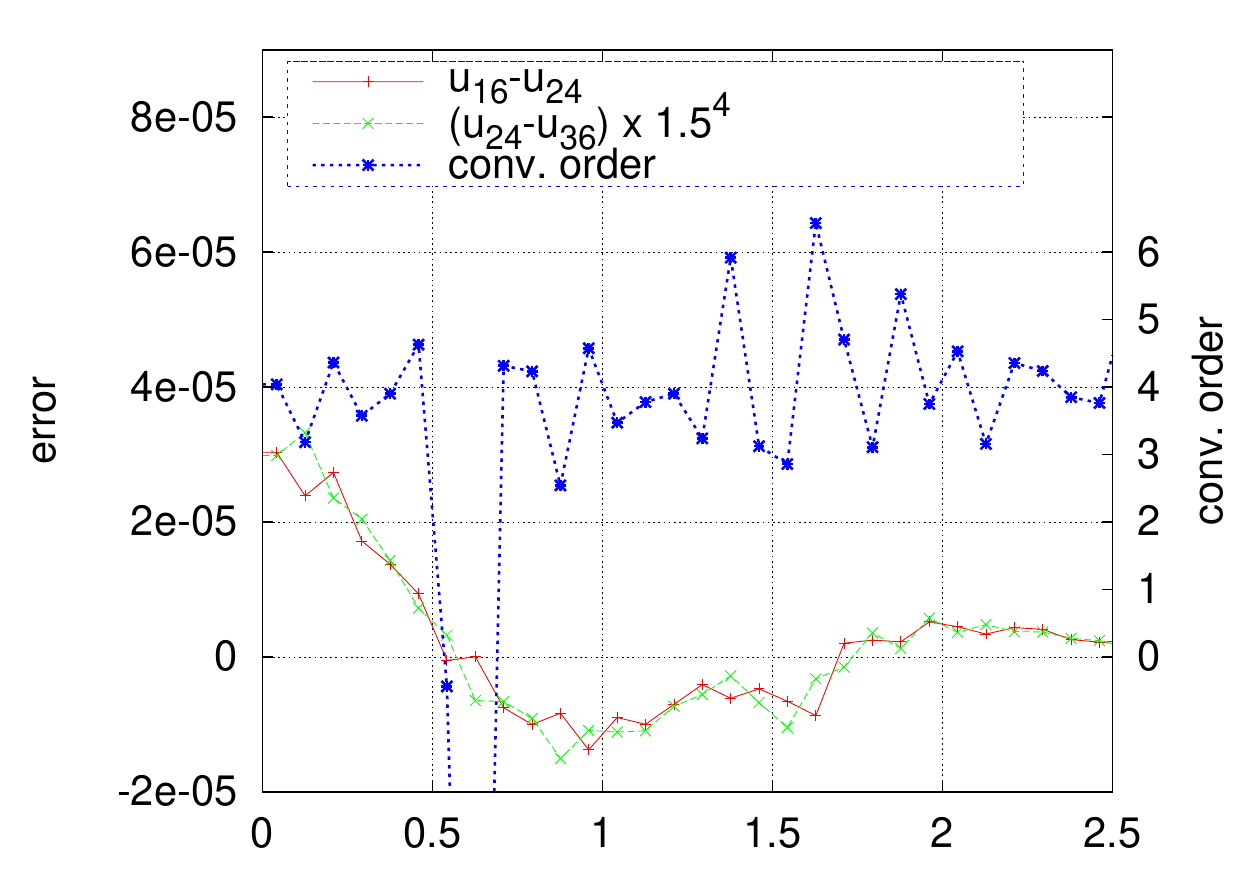} \\
  {\tiny (a)} & {\tiny (b)} \\
  \end{tabular}
  \caption{
    (a) 1-d cut through the Brill wave conformal factor $\psi_{fine}$
        for problem (a) along the $x$-axis.
    (b) Errors $\psi_{coarse}-\psi_{medium}$, 
        $\psi_{medium}-\psi_{fine}$, and pointwise convergence order on
        the 1-d cut along the $x$-axis.
  }
  \label{F:pointwise-conv-cpsi}
  \end{center}
  \end{figure}

Using the Hamiltonian constraint~(\ref{eqn:ham}) this is
\begin{eqnarray}
M & = & -\frac{1}{2 \pi} 
\int_{\calg{M}} \left(
\frac{1}{8} \hat{R} \phi 
+ \frac{1}{12} ({\rm tr} K)^2 \phi^5
\right. \\
 & & 
\left.
- \frac{1}{8} (\Ahatstar_{ab} + (\hat{L}W)_{ab})^2 \phi^{-7}
- 2 \pi \hat{\rho} \phi^{-3}
\right)
 \; dx
\\
 & & 
 + \frac{1}{2 \pi}
\int_{\partial \calg{M}-R} \hat{n}_a \hat{\nabla}^{a} \phi \; ds.
\label{eq:adm2}
\end{eqnarray}
Normally the inner boundaries will be highly refined (although one can 
construct cases where this does not happen) and so this method will give 
much more accurate results than the simple use of~(\ref{eq:adm1}).

\subsection{Brill waves initial data on multi-block domains}
\label{S:BrillWaveID}

As the final numerical experiment, we demonstrate the Brill wave initial
data generated by using~\FETK{} on a multi-block spherical domain.
The comprehensive treatment of this experiment can be found in~\cite{KAHPT07a}.

In General Relativity, initial data on a spatial 3D-slice has to satisfy the
Hamiltonian and momentum constraint equations~\cite{Arnowitt62},
\begin{eqnarray*}
{}^3R - K^{ij}K_{ij} + K^2 &= 0\\
\nabla_i(K^{ij} - g^{ij}K) &= 0
\end{eqnarray*}
where $K_{ij}$ and $K$ are the \textit{extrinsic curvature} of the 3D-slice and
its trace, respectively, and ${}^3R$ the Ricci scalar associated to the
spatial metric $g_{ij}$.

Brill waves~\cite{Bri59} 
constitute a simple yet rich example of initial data in numerical
relativity. In such case  the extrinsic curvature of the slice is zero, and
the above equations reduce to a single one, stating that the Ricci scalar has to
vanish:
\begin{equation}
 {}^3R = 0 \,. \label{Eq:Ris0}
\end{equation}

If the spatial metric is given up to one unknown function, Eq.~(\ref{Eq:Ris0})
in principle allows to solve for such function and thus complete the construction of
the initial data. The Brill equation is a special case of~(\ref{Eq:Ris0}), where the
3-metric is expressed through the conformal transformation $g_{ij} = \psi^4
\tilde{g}_{ij}$ of an unphysical metric $\tilde{g}_{ij}$, with an unknown
conformal factor $\psi$. Equation~(\ref{Eq:Ris0}) then becomes~\cite{B:Mur93}:
\begin{equation}
 (-\nabla_{\tilde{g}}^2 + \frac{1}{8}\tilde{R})\psi = 0
  \label{Eq:BrillWave2}
\end{equation}
where $\tilde{R}$ and $\nabla_{\tilde{g}}^2$ are the Ricci scalar curvature
and Laplacian of the unphysical metric $\tilde{g}_{ij}$, respectively.

  \begin{figure}[htbp]
  \begin{center}
  \begin{tabular}{cc}
  \includegraphics[width=0.45\textwidth]{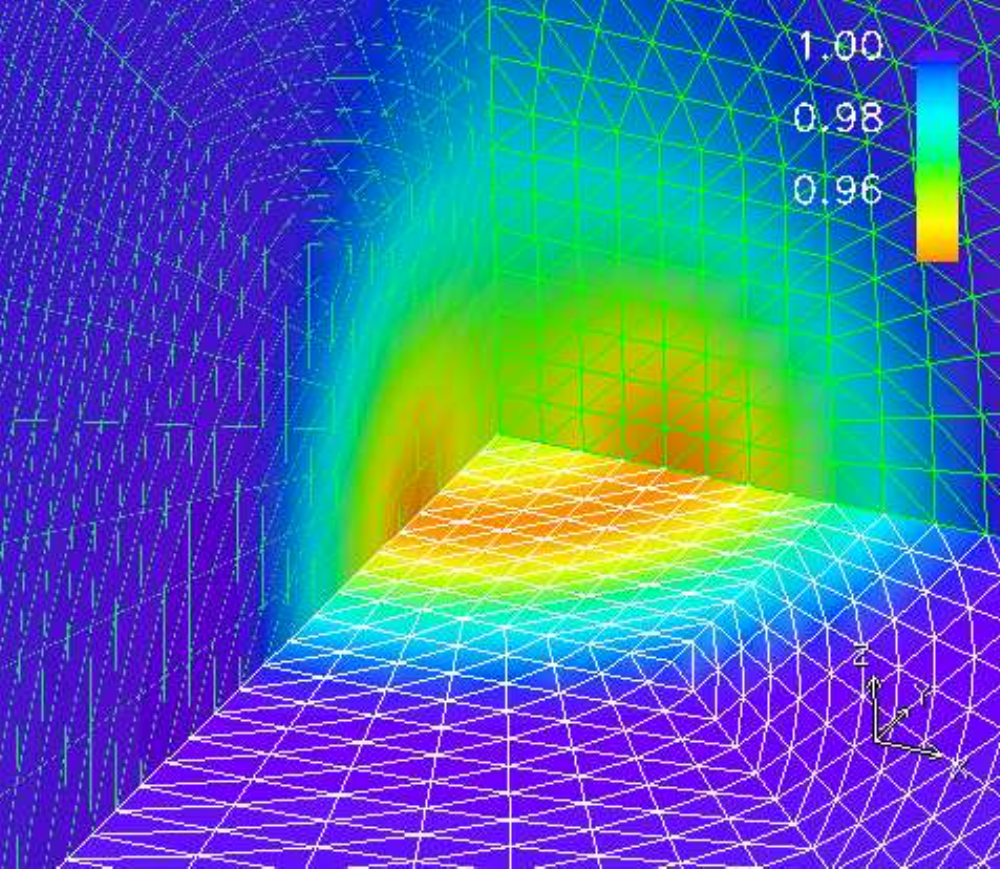}     &
  \includegraphics[width=0.45\textwidth]{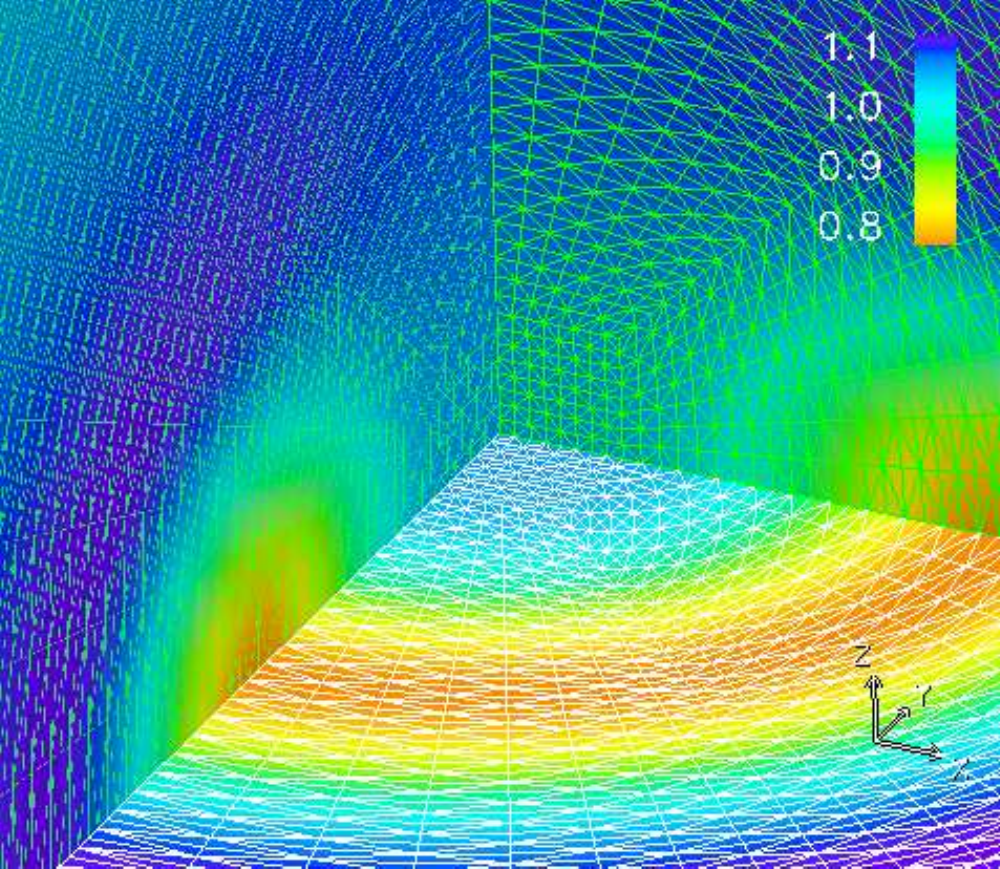} \\
  {\tiny (a)} & {\tiny (b)}
  \end{tabular}
  \caption{Potentials for the two types of Brill waves considered: 
           Holz' (a) and  toroidal (b) forms.}
  \label{F:bw-examples}
  \end{center}
  \end{figure}

We work with two specific choices for
$q(\rho,z)$:
\begin{itemize}
\item[(a)] Holz' form~\cite{alcubierre-2000-61}: 
           $q_H(\rho,z) = a_H \rho^2 e^{-r^2}$, with amplitude $a_H = 0.5$;
\item[(b)] toroidal form:  
           $q_t(\rho,z) = a_t \rho^2 \exp{\left(
               -\frac{(\rho-\rho_0)^2}{\sigma_{\rho}^2}
               -\frac{z^2}{\sigma_z^2}
           \right)}$, 
           with amplitude $a_t = 0.05$, radius $\rho_0=5$, width in
           $\rho$-direction $\sigma_{\rho}=3.0$ and width in
           $z$-direction $\sigma_z=2.5$.
\end{itemize}

Here we will focus on the axisymmetric case with the conformal metric given in
cylindrical coordinates by

\begin{equation}
 \tilde{g}_{ij} = e^{2q(\rho,z)}(d\rho^2+dz^2) + \rho^2 d\varphi^2 \,,
 \label{Eq:UnphMetric}
\end{equation}
where $q(\rho, z)$ is a function satisfying the following conditions:
\begin{enumerate}
\item regularity at the axis: $q(\rho=0,z) = 0$, $\frac{\partial q}{\partial\rho}|_{\rho=0} = 0$,
\item asymptotic flatness: $q(\rho,z)|_{r\to\infty} < O(1/r^2)$, 
      where $r$ is the spherical radius $r~=~\sqrt{\rho^2+z^2}$ \,.
\label{L:AsymptoticFlatnessBC}
\end{enumerate}    

The Hamiltonian constraint equation~(\ref{Eq:Ris0}) becomes a second order
elliptic PDE, which with asymptotically
flat boundary conditions at $r\to\infty$ takes the form
\begin{eqnarray}
\label{Eq:BrillWave}
-\nabla^2\psi(\rho,z) + V(\rho,z)\psi(\rho,z) & = & 0, \\
\label{Eq:PsiBoundaryCond}
 \psi|_{r\to\infty} &=& 1 + \frac{M}{2r} + O(1/r^2),
\end{eqnarray}
with the potential $V(\rho,z)$ given by 
$$
 V = -\frac{1}{4}(q''_{\rho\rho} + q''_{zz}).
$$

We numerically solve this equation using~\FETK{} on the 13-patch multi-block
spherical domain (see figure~\ref{F:patchsystems}).  We use domain 
 parameters $R_{out}~=~30$, $R_{med}~=~7$, $a_c~=~1.5$, and grid dimension
ratios $N:N_{r,inner}:N_{r,outer}=2:3:12$. Our low-medium-high resolution
triple is $N=32$, $N=36$ and $N=40$, except for pointwise convergence tests on
the $x$-axis (see figure~\ref{F:pointwise-conv-cpsi}), where we use $N=16$,
$N=24$ and $N=36$ (since they all differ by powers of $1.5$).

\section{Conclusion}
  \label{sec:conclusions}

In this paper we have described the use of finite element methods
to construct numerical solutions to the initial
value problem of general relativity.
We first reviewed the classical York conformal decomposition,
and gave a basic framework for deriving weak formulations.
We briefly outlined the notation used for the relevant function spaces,
and gave a simple weak formulation example.
We then derived an appropriate symmetric weak formulation of the 
coupled constraint equations, and summarized a number of basic 
theoretical properties of the constraints.
We also derived the linearization bilinear form of the weak form for use with
stability analysis or Newton-like numerical methods.
A brief introduction to adaptive finite element methods for
nonlinear elliptic systems was then presented,
and residual-type error indicators were derived.
We presented several general {\em a priori} error estimates 
from~\cite{Hols2001a,HoTs07a,HoTs07b}
for general Galerkin approximations to solutions equations such as the 
momentum and Hamiltonian constraints.
The numerical methods employed by MC were described in detail,
including the finite element discretization,
the residual-based {\em a posteriori} error estimator, the adaptive
simplex bisection strategy, the algebraic multilevel solver, and the
Newton-based continuation procedure for the solution of the nonlinear
algebraic equations which arise.
We described a mesh generation algorithm for modeling
compact binary objects,
outlined an algorithm for computing conformal Killing vectors, 
describes the numerical approximation of the ADM mass, and
gave an example showing the use of MC
for solution of the coupled constraints in the setting of a binary
compact object collision.

The implementation of these methods in the ANSI C finite element code
named MC was discussed in detail, including descriptions of the
algorithms and data structures it employs.
MC was designed specifically for solving general second-order
nonlinear elliptic systems of tensor equations on Riemannian manifolds
with boundary.  
The key feature of MC which makes it particularly useful for
relativity applications is the unusually high degree of abstraction
with which it can be used.
The user need only supply two functions (one for a linear problem)
in the form of a short C code file.
These functions are generally coded exactly as the weak form of the
equation and its linearization are written down
(our initial data constraint specification in MC is close to a
cut-and-paste of equations~(\ref{eqn:weakHamMom})
and~(\ref{eqn:weakHamMom_linearized}) into a C source file).
The user does not have to provide the elliptic system in
discrete form as is usually required in finite difference
implementations, and does not normally have to supply detailed coefficient
information.
In particular, the user provides only the two forms
$\langle F(u),v \rangle$ and $\langle DF(u)w,v \rangle$.
If {\em a posteriori} error estimation is to be used, then the user must
provide a third function $F(u)$, which is essentially the strong form of
the differential equation as needed for the error estimator given
in Section~\ref{sec:fem_residual}.
MC is also able to handle systems on a manifold whose atlas has
more than one chart.

\ack

M.~Holst thanks K.~Thorne, L. Lindblom, and H.~Keller for many fruitful
discussions over several years at Caltech.
D.~Bernstein thanks H.~Keller and Caltech Applied Mathematics for their
generous support during the initial phases of this research at Caltech.
M. Holst was supported in part by NSF Awards~0715146 and 0511766,
and DOE Awards DE-FG02-05ER25707 and DE-FG02-04ER25620.

B.~Aksoylu and M.~Holst thank O. Korobkin of Louisiana State University
for providing the numerical results on Brill wave initial data.

\section*{References}
\addcontentsline{toc}{section}{References}
\bibliographystyle{unsrt}
\bibliography{../bib/papers,../bib/books,../bib/mjh,../bib/burak}
\end{document}